\documentclass[11pt]{article}
\usepackage{cite, tikz, graphicx, subcaption, braket, amsmath, amssymb, amsthm, enumerate}
\usepackage[colorlinks = true]{hyperref}
\usepackage[margin = 2cm]{geometry}

\UseRawInputEncoding

\newtheorem{lemma}{Lemma}
\newtheorem{procedure}{Procedure}

\DeclareMathOperator{\tr}{Tr}
\DeclareMathOperator{\diag}{diag}

\title{Quantum teleportation of a qutrit state via a hypergraph state in a noisy environment}
\author{Souvik Giri\thanks{Email: \texttt{souvikgiri3010@gmail.com}}, Supriyo Dutta\thanks{Email: \texttt{dosupriyo@gmail.com}} \\
	\small{Department of Mathematics, National Institute of Technology Agartala}\\
	\small{Jirania, West Tripura, India - 799046.}}
\date{}

\begin{document}
	\maketitle
	\begin{abstract}
		The concept of quantum teleportation is fundamental in the theory of quantum communication. Developing models of quantum teleportation in different physical scenario is a modern trend of research in this direction. This work is at the interface of hypergraph theory and quantum teleportation. We propose a new teleportation protocol for qutrit states in a noisy environment. This protocol utilizes a shared quantum hypergraph state between parties as quantum channels, to carry quantum information. During the preparation of the shared state different quantum noise may act on it. In this article, we consider six types of noises which are generalized for qutrits using Wyel operators. They are the qutrit-flip noise, qutrit-phase-flip noise, depolarizing noise, Markovian and non-Markovian amplitude damping channel, Markovian and non-Markovian dephasing channel, and non-Markovian depolarization noise. There are only five qutrit hypergraph states which satisfies our requirements to be used as a channel. We consider all of them in this work. For different hypergraph states and different noises we work out the analytical expressions of quantum teleportation fidelity.\\
		
		\textbf{Keywords:} Quantum teleportation, Kraus operator, Weyl operator, noisy channel, qutrit, hypergraph quantum states, teleportation fidelity.
	\end{abstract}
	
	\tableofcontents
	
	\section{Introduction}
	
		Transmitting information from one party to another is a fundamental aspect of communication theory. In quantum communication, the quantum teleportation is a crucial tool to communicate quantum information \cite{mcmahon2007quantum}. A quantum state is a mathematical entity that symbolize the knowledge of a quantum system. Quantum teleportation  enables us to transfer quantum information without physically transmitting the quantum state. Therefore, investigations in quantum teleportation is essential theoretically and experimentally. Quantum teleportation was first introduced in \cite{bennett1993teleporting}. The initial experimental investigations on quantum teleportation was reported in \cite{boschi1998experimental, bouwmeester1997experimental}. Quantum teleportation is an active area of research in quantum information theory, in our age.
		
		A quantum state carries quantum information. Mathematically, a quantum states is a normalized vector in a Hilbert space. A qubit is a two-level quantum state, which is represented by a state vector of dimension two. A state of dimension three is called a qutrit. In general, a qudit is a state with dimension $d$. Initially, quantum teleportation was invented for qubit states. The higher-order qudits are an important resource in quantum-information processing \cite{lanyon2009simplifying, qudit2009emulation}. Hence, a teleportation scheme for them is certainly desirable. Quantum teleportation for higher dimensional states are investigated qutrits \cite{luo2019quantum, jie2008controlled, sebastian2023beyond, li2023improving}, and qudits \cite{al2010quantum, de2021efficient, wang2017generalized, fonseca2019high}. Experimental investigations on higher dimensional states are also reported in literature \cite{goyal2014qudit}.
		 
		Usually, in a quantum teleportation process we consider two parties, say Alice and Bob, who like to communicate a quantum state between them. We first distribute an entangled state between them, which acts as a channel to carry information. In the initial work of quantum teleportation the Bell state was used as a channel \cite{bennett1993teleporting}. Later the channel was replaced by GHZ state \cite{greenberger1989going, yang2009quantum, tabatabaei2023bi} , W state \cite{cabello2002bell, joo2003quantum}, cluster states \cite{briegel2001persistent, da2007teleportation}, etc. In this work, we consider a quantum hypergraph state \cite{qu2013encoding, rossi2013quantum} as a channel. These states are generalizations of graph states \cite{hein2004multiparty, van2005local, anders2006fast} and cluster states. The hypergraph states corresponding to the connected hypergraphs are entangled states with limited entanglement  \cite{guhne2014entanglement, dutta2019permutation}. Our motivation behind considering quantum hypergraph states is their potential applications in quantum communication, such as, in  building quantum repeater network \cite{hahn2019quantum, epping2016large, tzitrin2018local, azuma2023quantum}. The experimental implementation in this direction includes \cite{azuma2015all, borregaard2020one, zhan2023performance}. The graph and hypergraph states were initially developed for qubit states. Due to the experimental requirements they were generalized for qudits \cite{steinhoff2017qudit, xiong2018qudit, sarkar2021phase}. 
		
		The above model is an ideal description of quantum teleportation. However, in realistic physical implementations quantum teleportation noise is always present due to the inevitable interaction of the
		quantum states with the external environment. The effects of noise on quantum teleportation received considerable attention \cite{fonseca2019high, bennett1996purification, google2023measurement}. Noise may be applied at any phase of teleportation. In our investigation, we apply noise while establishing a state between two parties.
		
		Therefore, in this work we combine the following three fundamental aspects to build up our quantum teleportation procedure:
		\vspace{-0.2cm}
		\begin{itemize} 
			\itemsep0em 
			\item[-]
				The teleportation procedure is applicable for qutrits.
			\vspace{-0.2cm}
			\item[-] 
				It considers quantum hypergraph states, as a channel.
			\vspace{-0.2cm}
			\item[-] 
				The environmental noise is employed during the distribution of the channel.
		\end{itemize}
		\vspace{-0.2cm}
		The qutrit hypergraph states are generalization of qubit hypergraphh states. The idea of qutrit hypergraph states and their constructions are not readily available in literature, to the best of our knowledge. Therefore in this article we explain their construction in details. Till date, the qutrit hypergraph states are not used in quantum teleportation, to the best of our knowledge. This article fills this gap. Our teleportation protocol consists of three parties, which are Alice, Bob, and Dave. Dave prepares a qutrit hypergraph state and distribute between Alice and Bob via a noisy quantum channel. Then, Alice sends a qutrit quantum state to Bob via the shared state. We consider six types of quantum noise which may act on qutrit hypergraph state when Dave share it with Alice and Bob. They are the qutrit-flip noise, qutrit-phase-flip noise, depolarizing noise, Markovian and non-Markovian amplitude damping channel, Markovian and non-Markovian dephasing channel, and non-Markovian depolarization noise. The Kraus operators representing these noise are generalized for qutrits using Weyl operators. We have only five different options of quantum hypergraph states to be used as a quantum channel. We work out the analytical expressions of teleportation fidelity for different hypergraph states and be applying different noises on each of them. 
					
		This article is distributed into six sections and an appendix. In section 2, we discuss preliminary ideas on qutrit states and quantum channel to carry single-qutrit states. In section 3, we demonstrate qutrit hypergraph states and their construction which is an essential resource in our quantum teleportation. Section 4 is dedicated to the our teleportation protocol. Section 5 is distributed into six subsections. We consider qutrit-flip noise, qutrit-phase-flip noise, depolarizing noise, amplitude damping noise, dephasing, and non-markovian depolarization  channels in different subsections of section 5. Then we conclude this article. In appendix we present a calculation to find teleportation fidelity.

	\section{Qutrit quantum states and channels} 
	    
	    The fundamental component of quantum information in quantum computing is a qubit, also known as a quantum bit. We represent a qubit as a normalized vector in a two dimensional Hilbert space $\mathcal{H}_{2}$, which is spanned by $\ket{0} = \begin{bmatrix} 1 \\ 0 \end{bmatrix}$ and $\ket{1} = \begin{bmatrix} 0 \\ 1 \end{bmatrix}$. The Pauli operators play a significant role in quantum computation and information with qubits. Throughout this article $\sigma_x = \begin{bmatrix} 0 & 1 \\ 1 & 0 \end{bmatrix}, \sigma_y = \begin{bmatrix} 0 & -\iota \\ \iota & 0 \end{bmatrix}$ and $\sigma_z = \begin{bmatrix} 1 & 0 \\ 0 & -1 \end{bmatrix}$ denote the Pauli $X, Y$ and $Z$ operators.
	    
		In this work, we consider the qutrits which are the quantum states belonging to the three dimensional Hilbert space $\mathcal{H}_{3}$ \cite{byrd1998differential}. The set of vectors $\{\ket{j}: j = 0, 1, 2\}$ forms the computational basis of $\mathcal{H}_{3}$, where $\ket{0} = (1, 0, 0)^\dagger$, $\ket{1} = (0, 1, 0)^\dagger$, and $\ket{2} = (0, 0, 1)^\dagger$. The notations $\ket{0}$ and $\ket{1}$ are used for both qubits and qutrits. If we use them to denote qubits, we mention it specifically. Otherwise, they will indicate qutrit states. In this article, we consider a family of qutrit states to transmit utilizing a quantum teleportation protocol. In general, we can express them as
		\begin{equation}\label{quantum state vector}
		    \ket{\phi} = \cos (\text{$\theta $}_1) \ket{0} + \sin (\text{$\theta $}_1) \cos (\text{$\theta $}_2) \ket{1} + \sin (\text{$\theta $}_1) \sin (\text{$\theta $}_2) \ket{2}.
		\end{equation}
		Note that, $|\cos (\text{$\theta $}_1)|^2 + |\sin (\text{$\theta $}_1) \cos (\text{$\theta $}_2)|^2 + |\sin (\text{$\theta $}_1) \sin (\text{$\theta $}_2)|^2 = 1$. Here, $\theta_1$ and $\theta_2$ are the state parameters. In general there is no relationship between them. To make our calculations simpler, we sometime assume a linear relationship between $\theta_1$ and $\theta_2$, which is $\theta_1 = 3 \theta_2$. Under this condition we get
		\begin{equation}\label{quantum state vector_1}
			\ket{\phi}_{\theta_2 = \frac{\theta_1}{3}} = \cos(\theta_1) \ket{0} + \sin(\theta_1) \cos \left(\frac{\theta_1}{3} \right) \ket{1} + \sin(\theta_1) \sin \left(\frac{\theta_1}{3} \right) \ket{2},
		\end{equation}
		which has only one single state parameter $\theta_2$. We can be more specific. In equation (\ref{quantum state vector}), if $\theta_1 = \sin^{-1} \left(\sqrt{\frac{2}{3}}\right)$, and $\theta_2 = \frac{\pi}{4}$, then $\cos (\text{$\theta $}_1) = \sin (\text{$\theta $}_1) \cos (\text{$\theta $}_2) = \sin (\text{$\theta $}_1) \sin (\text{$\theta $}_2) = \frac{1}{\sqrt{3}}$. It generates the quantum state 
		\begin{equation}\label{plus_state}
		    	\ket{+} = \frac{1}{\sqrt{3}} (\ket{0} + \ket{1} + \ket{2}),
		\end{equation}
		which can be considered as a generalization of qubit $\ket{+}$ state. Also, putting $\theta_1 = \frac{\pi}{4}$ and $\theta_2 = \frac{\pi}{2}$ in equation (\ref{quantum state vector}), we get
		\begin{equation}\label{zero+two_state}
			\ket{\phi}_{\theta_1 = \frac{\pi}{4}, \theta_2 = \frac{\pi}{2}} = \frac{1}{\sqrt{2}} (\ket{0} + \ket{2}).
		\end{equation}
		In addition, putting $\theta_1 = 0$ and $\theta_2 = \frac{\pi}{2}$ in equation (\ref{quantum state vector}) we get
		\begin{equation}\label{zero_state}
			\ket{\phi}_{\theta_1 = 0, \theta_2 = \frac{\pi}{2}} = \ket{0}. 
		\end{equation}
	
		We can also describe a quantum state with a density matrix $\rho$, which is a positive semi-definite Harmitian matrix with trace $1$. The density matrix of the state vector $\ket{\phi}$ is $\rho = \ket{\phi} \bra{\phi}$.
	
		Let $A$ and $B$ be two qutrit systems. The space of composite two-qutrit system is represented by the Hilbert space $\mathcal{H}_3^{(2)} = \mathcal{H}_3 \otimes \mathcal{H}_3$, where $\otimes$ denotes tensor product. The idea of partial trace plays a crucial role, in this article. Let $\rho_{AB}$ be a density operator in $\mathcal{H}_3^{(2)}$. Then, the partial trace over the subsystem $B$ provides a density matrix in the system $A$, which is represented by 
		\begin{equation}
			\rho_A = \tr_{B}[\rho_{AB}] = \sum_{j = 0}^2 (I_A \otimes \bra{j}_B) \rho_{AB} (I_A \otimes \ket{j}_B).
		\end{equation}
		Recall that $\{\ket{j}: j = 0, 1, 2\}$ forms a computational basis of $\mathcal{H}_3$. Also, $\tr[A]$ denotes trace of matrix $A$ throughout this article. 
		
		A quantum channel carries quantum information from one site to another. Mathematically, a quantum channel is a completely positive and trace preserving linear map which transforms a density matrix to another density matrix \cite{kraus1971general}. A quantum channel is represented by a set of Kraus operators $\{K_1, K_2, \dots K_n\}$. The Kraus operators representing a quantum channel carrying a single qutrit state to another qutrit state are expressed by $3 \times 3$ matrices. Also, they should satisfy
		\begin{equation}\label{Kraus operators}
		   	\sum_{j = 1}^{n} K^{\dagger}_j K_j =I_3,
		\end{equation}
		where $I_3$ denotes the identity matrix of order $3$, and $n \leq 9$. Transferring a quantum state with density matrix $\rho$ via a quantum channel, we obtain a new state with density matrix 
		\begin{equation} \label{kraus}
		   	\rho' = \sum_{j = 1}^{n} K_j \rho K^{\dagger}_j.
		\end{equation}
		  	
		In this article, we utilize the Weyl operators for constructing quantum channels \cite{dutta2023qudit}. The Weyl operators are unitary operators. There are nine Weyl operators \cite{bennett1993teleporting, bertlmann2008bloch} for the qutrit system, which are represented by
		\begin{equation}
			{W}_{r,s} = \sum_{i=0}^{2} \omega^{ir} \ket{i} \bra{i \oplus s} ~\text{for}~ 0 \leq r,s \leq 2,
		\end{equation}
		where $\oplus$ denotes addition modulo $3$, throughout this article. We can mention them in matrix form as follows:
		\begin{equation} \label{Weyl operator}
			\small
			\begin{split}
				& W_{0,0} = \begin{bmatrix}
					1 & 0 & 0\\
					0 & 1 & 0\\ 
					0 & 0 & 1\\
				\end{bmatrix},
				W_{0,1} = \begin{bmatrix}
					0 & 1 & 0\\
					0 & 0 & 1\\ 
					1 & 0 & 0\\
				\end{bmatrix},
				W_{0,2} = \begin{bmatrix}
					0 & 0 & 1\\
					1 & 0 & 0\\ 
					0 & 1 & 0\\
				\end{bmatrix},\\
				&W_{1,0} = \begin{bmatrix}
					1 & 0 & 0\\
					0 & \omega & 0\\ 
					0 & 0 & \omega^2\\
				\end{bmatrix},
				W_{1,1} = \begin{bmatrix}
					0 & 1 & 0\\
					0 & 0 & \omega\\ 
					\omega^2 & 0 & 0\\
				\end{bmatrix},
				W_{1,2} = \begin{bmatrix}
					0 & 0 & 1\\
					\omega & 0 & 0\\ 
					0 & \omega^2 & 0
				\end{bmatrix},\\
				&W_{2,0} = \begin{bmatrix}
					1 & 0 & 0\\
					0 & \omega^2 & 0\\ 
					0 & 0 & \omega
				\end{bmatrix},
			W_{2,1} = \begin{bmatrix}
					0 & 1 & 0\\
					0 & 0 & \omega^2\\ 
					\omega & 0 & 0
				\end{bmatrix},
				W_{2,2} = \begin{bmatrix}
					0 & 0 & 1\\
					\omega^2 & 0 & 0\\ 
					0 & \omega & 0
				\end{bmatrix}.
			\end{split}
		\end{equation}
		It is clear that ${W}_{r,s}$ are the unitary operators for all $r$ and $s$, because ${W^{\dagger}}_{r,s} {W}_{r,s} = {W}_{r,s} {W^{\dagger}}_{r,s} = I_3$.  
	    	
    	At the end of quantum teleportation, we compare the initial and the final states using the quantum  teleportation fidelity. Our initial state is a pure state $\ket{\phi}$. After teleportation we get a mixed state $\rho$. The fidelity between $\ket{\phi}$ and $\rho$ is defined by \cite{mcmahon2007quantum}
	    \begin{equation}
	    	F(\ket{\phi}, \rho) = \sqrt{\braket{\phi | \rho | 	\phi}}.
	    \end{equation}

	\section{Quantum hypergraph states with qutrits}
	\label{Quantum_hypergraph_states_with_qutrits}

		A hypergraph \cite{bretto2013hypergraph} is a generalization of combinatorial graphs. We denote a hypergraph by $H = (V(H), E(H))$ which is a combination of a set of vertices $V(H)$ and a set of hyperedges $E(H) \subset \mathcal{P}(V(H))$, where $\mathcal{P}(V(H))$ denotes the power set of $V(H)$. A hyperedge $e = (v_1, v_2, \dots v_n)$ is a non-empty subset of the vertex set $V(H)$ consists of the vertices $v_1, v_2, \dots v_n$. The number of vertices $n$ in a hyperedge $e$ is called the cardinality of $e$. A loop is a hyperedge with cardinality $1$. Pictorially, we denote a vertex by a dot marked by a vertex labeling. We represent a hyperedge by a closed curve surrounding the vertices in it.

		For our teleportation protocol, we need three-qutrit entangled hypergraph states, which correspond the connected hypergraphs with three vertices. Therefore, the vertex set of all the hypergraphs considered in this article is $V(H) = \{0, 1, 2\}$. There are $2^3 = 8$ possible hyperedges with at most three vertices. A hyperedge may or may not be selected to construct a hypergraph. Thus, there are at most $64$ hypergraphs with three vertices. Many of them are isomorphic to each other, or may not be connected. In the construction of hypergraph states, a loop contributes a local $Z$-gate operation. We observe that there are only five non-isomorphic and connected hypergraphs with three vertices having no loop. The quantum hypergraph states corresponding to these hypergraphs are utilized in this article. We depict them in Figure \ref{Hypergraph}. 
		
		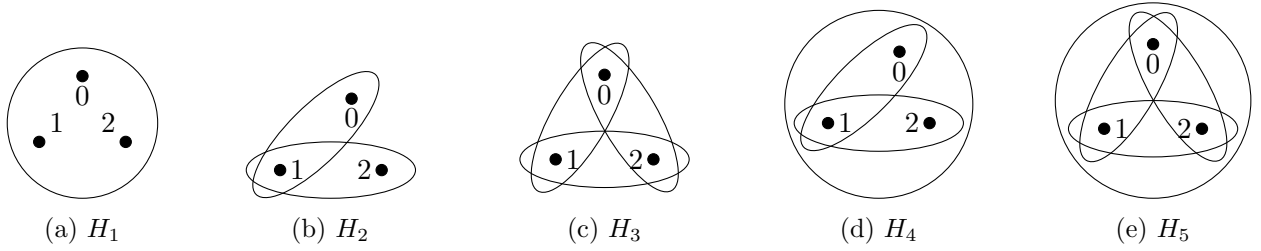
\begin{figure}
			\centering
			\begin{subfigure}{0.16\textwidth}
				\centering
				\begin{tikzpicture}[scale=0.25]
					\draw (0,0) circle [radius= 4];
					\draw [fill] (0,2.5) circle [radius=0.3] ;
					\node[below] at (0,2.5) {$0$};
					\draw [fill] (2.3,-1) circle [radius=0.3] ;
					\node[above left] at (2.3,-1) {$2$};
					\draw [fill] (-2.3,-1) circle [radius=0.3];
					\node[above right] at (-2.3,-1) {$1$};
				\end{tikzpicture}
				\caption{$H_1$}
				\label{fig:$H1$}
			\end{subfigure}
			\begin{subfigure}{0.2\textwidth}
				\centering
				\begin{tikzpicture}[scale=0.25]
					\draw (2,2) ellipse (4.5cm and 1.5cm);
					\draw[rotate=45] (3.6,1.9) ellipse (4.5cm and 1.5cm);
					\draw [fill] (-0.7,2) circle [radius=0.3] ;
					\node[right] at (-0.7,2) {$1$};
					\draw [fill] (4.7,2) circle [radius=0.3] ;
					\node[left] at (4.7,2) {$2$};
					\draw [fill] (3.1,5.8) circle [radius=0.3] ;
					\node[below] at (3.1,5.8) {$0$};
				\end{tikzpicture}
				\caption{$H_2$}
				\label{fig:$H2$}
			\end{subfigure}
			\begin{subfigure}{0.2\textwidth}
				\centering
				\begin{tikzpicture}[scale=0.25]
					\draw (2,2) ellipse (4.5cm and 1.5cm);
					\draw[rotate=60] (4,1.5) ellipse (4.5cm and 1.5cm);
					\draw [fill] (-0.6,2) circle [radius=0.3] ;
					\node[right] at (-0.6,2) {$1$};
					\draw [fill] (4.6,2) circle [radius=0.3] ;
					\node[left] at (4.6,2) {$2$};
					\draw [fill] (2,6.5) circle [radius=0.3] ;
					\node[below] at (2,6.5) {$0$};
					\draw[rotate=120] (2.01,-5) ellipse (4.5cm and 1.5cm);			
				\end{tikzpicture}
				\caption{$H_3$}
				\label{fig:$H3$}
			\end{subfigure}
			\begin{subfigure}{0.2\textwidth}
				\centering
				\begin{tikzpicture}[scale=0.25]
					\draw (2,2) ellipse (4.5cm and 1.5cm);
					\draw[rotate=45] (3.6,1.9) ellipse (4.5cm and 1.5cm);
					\draw [fill] (-0.7,2) circle [radius=0.3] ;
					\node[right] at (-0.7,2) {$1$};
					\draw [fill] (4.7,2) circle [radius=0.3] ;
					\node[left] at (4.7,2) {$2$};
					\draw [fill] (3.1,5.8) circle [radius=0.3] ;
					\node[below] at (3.1,5.8) {$0$};
					\draw (2,3) circle [radius=5];
				\end{tikzpicture}
				\caption{$H_4$}
				\label{fig:$H4$}
			\end{subfigure}
			\begin{subfigure}{0.2\textwidth}
				\centering
				\begin{tikzpicture}[scale=0.25]
					\draw (2,2) ellipse (4.5cm and 1.5cm);
					\draw[rotate=60] (4,1.5) ellipse (4.5cm and 1.5cm);
					\draw [fill] (-0.6,2) circle [radius=0.3] ;
					\node[right] at (-0.6,2) {$1$};
					\draw [fill] (4.6,2) circle [radius=0.3] ;
					\node[left] at (4.6,2) {$2$};
					\draw [fill] (2,6.5) circle [radius=0.3] ;
					\node[below] at (2,6.5) {$0$};
					\draw[rotate=120] (2.01,-5) ellipse (4.5cm and 1.5cm);		
					\draw (2,3.5) circle [radius=5.2];			
				\end{tikzpicture}
				\caption{$H_5$}
				\label{fig:$H5$}
			\end{subfigure}
			\caption{Hypergraphs representing 3 qutrit entangled states used in the article.}
			\label{Hypergraph}
		\end{figure}
		
		To construct a quantum hypergraph state from a hypergraph, we assign a $\ket{+}$ state corresponding to all the vertices which is defined in equation (\ref{plus_state}). Hence, the composite state is 
		\begin{equation}\label{+++}
			\begin{split}
				\ket{+++} = & \ket{+} \otimes \ket{+} \otimes \ket{+} \\
				= & \frac{1}{3\sqrt{3}} (\ket{000} + \ket{001} + \ket{002} + \ket{010} + \ket{011} + \ket{012} + \ket{020} + \ket{021} + \ket{022} + \ket{100}\\
				& + \ket{101} + \ket{102} + \ket{110}+\ket{111}+\ket{112}+\ket{120}+\ket{121}+\ket{122}+\ket{200}+\ket{201}+\ket{202}\\
				& +\ket{210}+\ket{211}+\ket{212}+ \ket{220}+ \ket{221}+ \ket{222}) \\
				= & \frac{1}{3\sqrt{3}} \sum_{q_0 = 0}^{2} \sum_{q_1 = 0}^{2} \sum_{q_2 = 0}^{2} \ket{q_0, q_1, q_2},
			\end{split}
		\end{equation}
		where $q_0, q_1$ and $q_2$ are the computational basis states of qutrit system corresponding to the vertices $0$, $1$, and $2$, respectively.  
		
		Now, we apply non-local operations on $\ket{+++}$ corresponding to the hyperedges to generate entangled states, which are different types of controlled-$Z^{(3)}$ operators. The standard generalization of Pauli $\sigma_z$ operator \cite{rossi2013quantum, volya2023qudcom} applicable on single qutrit state is given by
		\begin{equation}
			Z^{(3)} = \begin{bmatrix}
				1 & 0 & 0 \\
				0 & \omega & 0 \\
				0 & 0 & \omega^2
			\end{bmatrix} = \sum^{2} _{j=0}  \omega^j  \ket{j} \bra{j},
		\end{equation}
		where $\omega = e^{2\pi \iota/3}$ is the complex cubic root of $1$ that is $\omega = \left(\frac{-1 + \iota \sqrt{3}}{2}\right), \omega^2 = \left(\frac{-1 - \iota \sqrt{3}}{2}\right)$ and $\omega^3 = 1$.
		
		As there is no loop in our hypergraphs and they have three vertices, a hyperedge in them contains either two vertices or three vertices. Given any hyperedge $e = (c, t)$ with two vertices, we apply a two-qutrit controlled-$Z$ gate which is defined by a block diagonal matrix $CZ^{(3)} = \diag\{I_3, I_3, Z^{(3)}\}$ applicable on a state $\ket{c} \otimes \ket{t} \in \mathcal{H}_3^{(2)}$. It assumes the first qutrit $\ket{c}$ as control and the second qutrit $\ket{t}$ as target. Note that,
		\begin{equation}\label{controlled_z_for_cardinality_2}
			CZ^{(3)} \ket{c} \otimes \ket{t} = \begin{cases}
				\omega^t \ket{c} \otimes \ket{t} & ~\text{if}~ c = 2; \\ \ket{c} \otimes \ket{t} & ~\text{otherwise}.
			\end{cases}
		\end{equation}
		Note that, $CZ^{(3)}$ is applicable on two qutrits of $\ket{+++}$ depending on the hyperedge. We apply identity operation on the remaining qutrit
		
		A hypergraph may also contain a hyperedge with three vertices. We apply the control-control-$Z$ ($CCZ^{(3)}$) operation corresponding to them. It is defined by a block diagonal matrix $CCZ^{(3)} = \diag\{I_3, I_3, \dots, I_3(8\text{-times}), Z^{(3)}\}$. Note that,
		\begin{equation}\label{controlled_z_for_cardinality_3} 
			CCZ^{(3)} \ket{c_1} \otimes \ket{c_2} \otimes \ket{t} = \begin{cases}
				\omega^t \ket{c_1} \otimes \ket{c_2} \otimes \ket{t} & ~\text{if}~ c_1 = c_2 = 2; \\ \ket{c_1} \otimes \ket{c_2} \otimes \ket{t} & ~\text{otherwise}.
			\end{cases}
		\end{equation}
				
		The operators $CZ^{(3)}$, and $CCZ^{(3)}$ are commuting block diagonal matrices. Therefore, we can apply them on $\ket{+++}$, in any order. The resultant state depends on the existence of a particular hypereedges in the corresponding hypergraphs. After applying all the controlled-$Z^{(3)}$ operations corresponding to the hyperedges available in the hypergraph $H$ on $\ket{+++}$ we get the hypergraph state \cite{gachechiladze2016extreme, rossi2013quantum, xiong2018qudit, dutta2019permutation, morimae2017verification}
		\begin{equation}
			\ket{H} = \frac{1}{3\sqrt{3}} \sum_{q_0 = 0}^{2} \sum_{q_1 = 0}^{2} \sum_{q_2 = 0}^{2} h(q_0, q_1, q_2) \ket{q_0, q_1, q_2},
		\end{equation}
		where $h$ is a three variable function $h: \{0, 1, 2\}^{\times 3} \rightarrow \{1, \omega, \omega^2\}$.    
		
		Recall that we consider only five hypergraphs and their corresponding states which are $H_1, H_2, H_3, H_4$, and $H_5$, where $V(H_1) = V(H_2) = V(H_3) = V(H_4) = V(H_5) = \{0, 1, 2\}$. Also, $E(H_1) = \{(0, 1, 2)\}$, $E(H_2) = \{(0, 1), (1, 2)\}$, $E(H_3) = \{(0, 1), (1, 2), (2, 0)\}$, $E(H_4) = \{(0, 1), (1, 2), (0, 1, 2)\}$, and $E(H_5) = \{(0, 1), (1, 2), (2, 0), (0, 1, 2)\}$. Applying different combinations of $CZ^{(3)}$ and $CCZ^{(3)}$ operators on $\ket{+++}$ we construct the corresponding hypergraph states, which are mentioned below:
		\begin{equation} \label{h1}
			\small
			\begin{split}
				\ket{H_1}= & \frac{1}{3\sqrt{3}}(\ket{000}+\ket{001}+\ket{002}+\ket{010}+\ket{011}+\ket{012}+\ket{020}+\ket{021}+\ket{022}+\ket{100}+\ket{101}\\
				& +\ket{102} +\ket{110} +\ket{111}+\ket{112}+\ket{120}+\ket{121}+\ket{122}+\ket{200}+\ket{201}+\ket{202}+\ket{210}\\
				& +\ket{211}+\ket{212}+\ket{220} + \omega \ket{221}+ \omega^2\ket{222}).
			\end{split}
		\end{equation} 
		\begin{equation} \label{h2}
			\small
			\begin{split}
				\ket{H_2}= & \frac{1}{3\sqrt{3}}(\ket{000}+\ket{001}+\ket{002}+\ket{010}+\ket{011}+\ket{012}+\ket{020}+ \omega \ket{021}+ \omega^2 \ket{022}+\ket{100}\\
				& +\ket{101} +\ket{102} +\ket{110}+\ket{111}+\ket{112}+\ket{120}+ \omega \ket{121}+ \omega^2 \ket{122}+\ket{200}+\ket{201}\\
				& +\ket{202}+ \omega \ket{210}+ \omega \ket{211} + \omega \ket{212}+ \omega^2 \ket{220}+\ket{221}+ \omega \ket{222}.
			\end{split}
		\end{equation}
		\begin{equation} \label{h3}
			\small
			\begin{split}
				\ket{H_3}= & \frac{1}{3\sqrt{3}}(\ket{000}+\ket{001}+\ket{002}+\ket{010}+\ket{011}+\ket{012}+\ket{020}+ \omega \ket{021}+ \omega^2 \ket{022}+\ket{100}\\
				& +\ket{101}+ \omega \ket{102} +\ket{110}+\ket{111}+ \omega \ket{112}+ \ket{120}+ \omega^2 \ket{121}+ \ket{122}+\ket{200}+ \ket{201}\\
				& + \omega^2 \ket{202}+ \omega \ket{210}+ \omega \ket{211} + \ket{212}+ \omega^2 \ket{220}+ \ket{221}+ \ket{222}.
			\end{split}
		\end{equation}
		\begin{equation} \label{h4}
			\small
			\begin{split}
				\ket{H_4}= & \frac{1}{3\sqrt{3}}(\ket{000}+\ket{001}+\ket{002}+\ket{010}+\ket{011}+\ket{012}+\ket{020}+ \omega \ket{021}+ \omega^2 \ket{022}+\ket{100}\\
				& +\ket{101}+\ket{102} +\ket{110}+\ket{111}+\ket{112}+\ket{120}+ \omega \ket{121}+ \omega^2 \ket{122}+\ket{200}+\ket{201}\\
				& +\ket{202}+ \omega \ket{210}+ \omega \ket{211} + \omega \ket{212}+ \omega^2 \ket{220}+ \omega \ket{221}+\ket{222}.
			\end{split}
		\end{equation}
		\begin{equation} \label{h5}
			\small
			\begin{split}
				\ket{H_5}= & \frac{1}{3\sqrt{3}}(\ket{000}+\ket{001}+\ket{002}+\ket{010}+\ket{011}+\ket{012}+\ket{020}+ \omega \ket{021}+ \omega^2 \ket{022}+\ket{100}\\
				& +\ket{101}+ \omega \ket{102} +\ket{110}+\ket{111}+ \omega \ket{112}+ \ket{120}+ \omega^2 \ket{121}+ \ket{122}+\ket{200}+ \ket{201}\\
				& + \omega^2 \ket{202}+ \omega \ket{210}+ \omega \ket{211} + \ket{212}+ \omega^2 \ket{220}+ \omega \ket{221}+ \omega^2 \ket{222}.
			\end{split}
		\end{equation}

	\section{Schemes of teleportation}
		
		Now, we are in a position to discuss our quantum teleportation procedure. It involves three parties, which are Alice, Bob, and Dave. Dave prepares a quantum hypergraph state and transmit to Alice and Bob via a noisy quantum channel. After distribution of the state between Alice and Bob, Alice wants to sends an unknown state $\ket{\phi}_{a}$ to Bob. In general we assume that $\ket{\phi}_{a} = \ket{\phi}$, which is mentioned in equation (\ref{quantum state vector}). The density matrix of $\ket{\phi}_{a}$ is $\rho_{in} = \ket{\phi}_a\bra{\phi}$. Below is our procedure for quantum teleportation:
		
		\begin{procedure}\label{teleportation_procedure}
			Our teleportation protocol  consists of the following steps:
			\begin{enumerate}[Step 1:]
				\item \label{step_1}
					Dave creates one of the qutrit entanglement states $\ket{H_{i}}_{123}$, where $i=1,2,3,4,5$ mentioned in equation (\ref{h1}), (\ref{h2}), (\ref{h3}), (\ref{h4}), and (\ref{h5}). The density matrix of $\ket{H_{i}}_{123}$ is 
					\begin{equation} \label{H_i}
						\hat{\rho_{i}} = \ket{H_{i}}_{123} \bra{H_{i}}.
					\end{equation}	 
				\item \label{step_2}
					Dave transmits particles $0$ and $1$ to Alice and particle $2$ to Bob via a noisy quantum channel. We assume that same noise acts on all qutrits. Applying equation (\ref{kraus}), on the density matrix in equation (\ref{H_i}) we find a new density matrix
					\begin{equation}
						\hat{\hat{\rho_i}}=\sum_{j=1} K_j \hat{\rho_{i}} K_j^{\dagger}, \{i=1, 2, 3, 4, 5\}.
					\end{equation} 
					Thus, state of the whole system is
					\begin{equation}
						\rho_{H_i} = \rho_{in} \otimes \hat{\hat{\rho_i}}.
					\end{equation}
				\item \label{step_3}
					Using the following four orthogonal qutrit states $\{\ket{\psi_l}, l = 1, 2, 3, 4\}$ \cite{hu2020experimental}, Alice performs a von-Neumann measurement on her own qutrits $a, 0, 1$, with respect to the following orthogonal states: 
					\begin{equation}\begin{aligned} \label{orthogonal}
							\ket{\psi_1}=\frac{1}{\sqrt{3}}\Big[\ket{000} + \ket{111} + \ket{222}\Big], \ket{\psi_2}=\frac{1}{\sqrt{3}}\Big[\ket{000} + \omega \ket{111} + \omega^2 \ket{222}\Big],\\ \ket{\psi_3}=\frac{1}{\sqrt{3}}\Big[\ket{012} + \ket{120} + \ket{201}\Big], \ket{\psi_4}=\frac{1}{\sqrt{3}}\Big[\ket{012} + \omega \ket{120} + \omega^2 \ket{201}\Big].
						\end{aligned}
					\end{equation}
					
					 The density matrices which are generated after the measurement are
					\begin{equation} \label{M_l}
						\rho_l = \frac{M_l\rho_{H_i} M^{\dagger}_l}{\tr[M^{\dagger}_l M_l \rho_{H_i}]},
					\end{equation}
					where $M_l = \ket{\psi_l} \bra{\psi_l}$ are the measurement operators. Therefore, after these measurements the state $\rho_{H_i}$ is shifted to one of $\rho_l$ for $l = 1, 2, 3, 4$ with a assured probability.\\
				\item \label{step_4}
					Alice informs the measurement result to Bob using classical communication. The state of Bob is
					\begin{equation}
						\rho_{B_l} = \tr_{a01} [\rho_l] = \frac{\tr_{a01}[M_l \rho_{H_i} M^{\dagger}_l]}{\tr[M^{\dagger}_l M_l \rho_{H_i}]},
					\end{equation}
					where $\tr_{a01}$ stands for the partial trace over the particles $(a, 0, 1)$.
				\item \label{step_5}
					On his particle $2$, Bob applies a unitary operator $U_l=\Lambda_l, l = 1, 2, 3, 4$ \cite{zhang2023efficiency} which are 
					\begin{equation} \label{unitary}
						\Lambda_1 = \begin{bmatrix}
							1&0&0\\
							0&1&0\\ 
							0&0&1\\
						\end{bmatrix},
						\Lambda_2 = \begin{bmatrix}
							1&0&0\\
							0&-1&0\\ 
							0&0&1\\
						\end{bmatrix},
						\Lambda_3 = \begin{bmatrix}
							0&1&0\\
							1&0&0\\ 
							0&0&1\\
						\end{bmatrix},
						\Lambda_4 = \begin{bmatrix}
							0&-1&0\\
							1&0&0\\ 
							0&0&1\\
						\end{bmatrix}.
					\end{equation}
					Therefore, Bob's final state is\\
					\begin{equation} \label{output}
						\widetilde{\rho}_{B_l} = U_l \rho_{B_l} U^{\dagger}_l = \frac{U_l\tr_{a01}[M_l \rho_{H_i} M^{\dagger}_l]U^{\dagger}_l}{\tr[M^{\dagger}_l M_l \rho_{H_i}]}.
					\end{equation}
					To quantify the efficiency of the procedure, we calculate the teleportation fidelity, which measures the overlap between the initial and final states. Depending on the hypergraph state $\ket{H_i}$, the quantum teleportation fidelity is determined by  
					\begin{equation} \label{fidelity}
						F_{H_i} = \sum^{4}_{l=1} P_l F_l,
					\end{equation}
					where $P_l = \tr[M^{\dagger}_l M_l \rho_{H_i}]$ and $F_l = \bra{\phi}_a \widetilde{\rho}_{B_l} \ket{\phi}_a$ for $i = 1, 2, \dots 5$.
			\end{enumerate}
		\end{procedure}

	\section{Noisy quantum channels for 3-qutrit systems and quantum teleportation}
       
       In this section, we construct 3-qutrit quantum channels to share a state between Alice and Bob, which is prepared by Dave. A number of channels were initially defined for qubtis. We generalized them for qutrits. We consider all five hypergraph states for any channel. Then we develop the analytic expression of quantum teleportation fidelity. For transmitting a three qutrit state we need a tensor product structure in the quantum channel, which we discuss in the following lemma. 	
       
       \begin{lemma}\label{multiqutrit_channel_lemma}
       	Let $\{K_i: i=1, 2, \dots, 9\}$ forms a set of Kraus operators for a qutrit quantum channel. Define $K_{i_1} = K_{r,s} \otimes I_3 \otimes I_3, K_{i_2} = I_3 \otimes K_{r,s} \otimes I_3, K_{i_3} = I_3 \otimes I_3 \otimes K_{r,s}$ where $r = s = 0, 1, 2$. Then the set of operators $\{K_{i_1}, K_{i_2}, K_{i_3} : {i_1} = {i_2} = {i_3} = 1, 2, \dots, 9\}$ forms a set of Kraus operators for three qutrit quantum state.
       \end{lemma}
       
       \begin{proof}
       	Let $K_1, K_2, \dots, K_n$ where $n \leq 9$ are the Kraus operators each of dimension $3$ representing a single qutrit quantum channel. Therefore $\sum_{j=1}^{n} K_j^\dagger K_j = I_3$. For three qutrit system observe that $\sum_{i_1 =1}^{n} K_{i_1}^\dagger K_{i_1} + \sum_{i_2=1}^{n} K_{i_2}^\dagger K_{i_2} + \sum_{i_3 = 1}^{n} K_{i_3}^\dagger K_{i_3} = I_3 \otimes I_3 \otimes I_3$.
       \end{proof}
       
       Now we present teleportation fidelity which depends on noisy channels, hypergraph states, and the state to teleport.
       
		\subsection{Qutrit-flip noise and quantum teleportation:}\label{Qutrit-flip noise and quantum teleportation} 
			We generalize the bit-flip noise for qutrits and call it qutrit-flip noise. Recall that the bit-flip noise is represented by
			\begin{center}
				\begin{tabular}{c c}
					$K_1 = \sqrt{1-p} I_2,$ & $K_2 =\sqrt{p} \sigma_x$.
				\end{tabular}
			\end{center} 
			It changes qubit $\ket{0}$ to $\ket{1}$ or $\ket{1}$ to $\ket{0}$ with probability $p$ as well as keep them unchanged with probability $(1 - p)$. 
			
			Similarly, a qutrit-flip noise flips the state $\ket{i}$ to the states $\ket{i \oplus 1}$, and $\ket{i \oplus 2}$ with probability $p$ for $i = 0, 1, 2$ as well as keep them unchanged with probability $(1 - p)$. The corresponding Kraus operators are \cite{fonseca2019high, dutta2023qudit, fortes2015fighting}
      \begin{equation} \label{Qutrit flip noise}
      	K_{0,s} = \begin{cases}
      		\sqrt{\frac{1-p}{3}} I_3 & \text{when}~ r = 0, s = 0, \\
      		\sqrt{\frac {p} {6}} W_{0,s} & \text{for}~  r = 0, 1\leq s \leq 2,
      	\end{cases}
      \end{equation}
      where the Weyl operators are $W_{0,0}, W_{0,1}, \text{and}~ W_{0,2}$ define in equation (\ref{Weyl operator}). Using Lemma \ref{multiqutrit_channel_lemma} we can construct a three-qutrit channel which can be described by the following Kraus operators:

      \small
      \begin{tabular}{p{5.5cm} p{5.5cm} p{5.5cm}}
      	\centering
      	$K_1 = \sqrt{\frac{1-{p}}{3}} (K_{0,0} \otimes I_3 \otimes I_3)$, & $K_2 = \sqrt{\frac{{p}}{6}} (K_{0,1} \otimes I_3 \otimes I_3)$, & $K_3=\sqrt{\frac{{p}}{6}} (K_{0,2} \otimes I_3 \otimes I_3)$,\\
      	$K_4=\sqrt{\frac{1-{p}}{3}} (I_3 \otimes K_{0,0} \otimes I_3)$, & $K_5=\sqrt{\frac{{p}}{6}} (I_3 \otimes K_{0,1} \otimes I_3)$, & $K_6=\sqrt{\frac{{p}}{6}} (I_3 \otimes K_{0,2} \otimes I_3)$,\\
      	$K_7=\sqrt{\frac{1-{p}}{3}} (I_3 \otimes I_3 \otimes K_{0,0})$, & $K_8=\sqrt{\frac{{p}}{6}} (I_3 \otimes I_3 \otimes K_{0,1})$, & $K_9=\sqrt{\frac{{p}}{6}} (I_3 \otimes I_3 \otimes K_{0,2})$.
		\end{tabular}
      
		Recall that $\hat{\rho_i}$ is a density matrix of any one of our hypergraph states $\ket{H_i}$ for $i = 1, 2, \dots 5$, which is generated by Dave in Step \ref{step_1} of our teleportation procedure. Now in Step \ref{step_2} Dave transfers the quantum state $\hat{\rho_i}$ via qutrit-flip noise channel. It generates the new quantum state
		\begin{equation} \label{evaluation 1}
			\hat{\hat{\rho_i}}=\sum_{j=1}^{9} K_j \hat{\rho_{i}} K_j^{\dagger}, ~\text{for}~ i=1, 2, 3, 4, 5.
		\end{equation}
		Therefore after transportation of hypergraph states via the noisy quantum channel the combined state can by written as 
		\begin{equation} \label{H1}
			\rho_{H_i} = \rho_{in} \otimes \hat{\hat{\rho_i}},
		\end{equation}
		where $\rho_{in}$ is the density matrix of $\ket{\phi}_a$ which Alice wants to transfer to Bob, and mentioned in equation  (\ref{quantum state vector}). 
		
		Alice performs a von Neumann measurement on her particles $(a, 0, 1)$ and sends the outcomes of her measurements to Bob via classical communication. Bob applies an unitary operator to restore the state of his particle $2$ to state $\ket{\phi}_{a}$. Then we calculate the teleportation fidelity, which can be expressed in terms of the parameters of the initial state and the parameters of noisy channel as follows: 
		\begin{equation}
			\small
			\begin{split}
			F_{H_{1}} = & \frac{5}{5184}[8 \cos (2 \text{$\theta $}_1)-162 \cos (4 \text{$\theta $}_1)+24 (2-3 \text{p}) \sin ^4(\text{$\theta $}_1) \sin (4 \text{$\theta $}_2)\\
			& +8 (2-15 \text{p}) \sin ^4(\text{$\theta $}_1) \cos (4 \text{$\theta $}_2) -16 (15 \text{p}+8) \sin ^3(\text{$\theta $}_1) \cos (\text{$\theta $}_1) \cos (3 \text{$\theta $}_2)\\
			& +32 (3 \text{p}+17) \sin ^3(\text{$\theta $}_1) \cos (\text{$\theta $}_1) \sin (3 \text{$\theta $}_2)+8 \sin ^2(\text{$\theta $}_1) \sin (2 \text{$\theta $}_2)((74-25 \text{p}) \cos (2 \text{$\theta $}_1)\\
			& -19 \text{p}+110) +16 (\text{p}+4) \sin ^2(2 \text{$\theta $}_1) \cos (2 \text{$\theta $}_2)+4 \sin (2 \text{$\theta $}_1) \cos (\text{$\theta $}_2) (-(19 \text{p}+8)\cos (2 \text{$\theta $}_1)\\
			& +11 \text{p}+40)+8 \sin (2 \text{$\theta $}_1) \sin (\text{$\theta $}_2) ((7-17 \text{p}) \cos (2 \text{$\theta $}_1)+13 \text{p}+41)-60 \text{p} \cos (2 \text{$\theta $}_1)\\
			& +47 \text{p} \cos (4 \text{$\theta $}_1)+13 \text{p}+666].
		\end{split}
		\end{equation}
		\begin{equation}
			\small
			\begin{split}
				F_{H_{2}} = & \frac{5}{2592}[16 (9 p+2) \sin ^2(\text{$\theta $}_1) \sin (2 \text{$\theta $}_2)+8 (3 p-2) \sin ^4(\text{$\theta $}_1) \cos (4 \text{$\theta $}_2)\\
				& +16 (4 p-7) \sin ^3(\text{$\theta $}_1) \cos (\text{$\theta $}_1) \cos (3 \text{$\theta $}_2)+16 (29-26 p) \sin ^3(\text{$\theta $}_1) \cos (\text{$\theta $}_1) \sin (3 \text{$\theta $}_2)\\
				& -104 (p-1) \sin ^2(2 \text{$\theta $}_1) \cos (2 \text{$\theta $}_2)+4 \cos (2 \text{$\theta $}_1) (8 (7 p-5) \sin ^2(\text{$\theta $}_1) \sin (2 \text{$\theta $}_2)+3 p-2)\\
				& +4 \sin (2 \text{$\theta $}_1) \cos (\text{$\theta $}_2) ((20 p-19) \cos (2 \text{$\theta $}_1)+11)\\
				& +4 \sin (2 \text{$\theta $}_1) \sin (\text{$\theta $}_2) ((19-18 p) \cos (2 \text{$\theta $}_1)-6 p+29)+(p-18) \cos (4 \text{$\theta $}_1)-13 p+282].
			\end{split}
		\end{equation}
		\begin{equation}
			\small
			\begin{split}
				F_{H_{3}} = & \frac{5}{5184}[-40 \cos (2 \text{$\theta $}_1)-78 \cos (4 \text{$\theta $}_1)+48 (p-1) \sin ^4(\text{$\theta $}_1) \sin (4 \text{$\theta $}_2)\\
				& +40 (p-2) \sin ^4(\text{$\theta $}_1) \cos (4 \text{$\theta $}_2)+16 (9 p-14) \sin ^3(\text{$\theta $}_1) \cos (\text{$\theta $}_1) \cos (3 \text{$\theta $}_2)\\
				& +16 (16-3 p) \sin ^3(\text{$\theta $}_1) \cos (\text{$\theta $}_1) \sin (3 \text{$\theta $}_2)+16 \sin ^2(\text{$\theta $}_1) \sin (2 \text{$\theta $}_2) ((16 p+7) \cos (2 \text{$\theta $}_1)+10 p+13)\\
				& +16 (7-8 p) \sin ^2(2 \text{$\theta $}_1) \cos (2 \text{$\theta $}_2)+4 \sin (2 \text{$\theta $}_1) \cos (\text{$\theta $}_2) ((41 p-38) \cos (2 \text{$\theta $}_1)-9 p+22)\\
				& +4 \sin (2 \text{$\theta $}_1) \sin (\text{$\theta $}_2) ((32-21 p) \cos (2 \text{$\theta $}_1)-43 p+64)+20 p \cos (2 \text{$\theta $}_1)\\
				& +59 p \cos (4 \text{$\theta $}_1)-79 p+630]. 
			\end{split}
		\end{equation}
		\begin{equation}
			\small
			\begin{split}
				F_{H_{4}} = & \frac{5}{5184}[-16 \cos (2 \text{$\theta $}_1)-60 \cos (4 \text{$\theta $}_1)-8 p \sin ^4(\text{$\theta $}_1) \sin (4 \text{$\theta $}_2)+8 (5 p-4) \sin ^4(\text{$\theta $}_1) \cos (4 \text{$\theta $}_2)\\
				&+112 (p-2) \sin ^3(\text{$\theta $}_1) \cos (\text{$\theta $}_1) \cos (3 \text{$\theta $}_2)+32 (17-12 p) \sin ^3(\text{$\theta $}_1) \cos (\text{$\theta $}_1) \sin (3 \text{$\theta $}_2)\\
			& +8 \sin ^2(\text{$\theta $}_1) \sin (2 \text{$\theta $}_2) ((63 p-52) \cos (2 \text{$\theta $}_1)+45 p-4)+16 (10-11 p) \sin ^2(2 \text{$\theta $}_1) \cos (2 \text{$\theta $}_2)\\
				& +4 \sin (2 \text{$\theta $}_1) \cos (\text{$\theta $}_2) ((35 p-38) \cos (2 \text{$\theta $}_1)-3 p+22)-8 \sin (2 \text{$\theta $}_1) \sin (\text{$\theta $}_2) ((18 p-19) \cos (2 \text{$\theta $}_1)\\
				& +6 p-29) +20 p \cos (2 \text{$\theta $}_1)+43 p \cos (4 \text{$\theta $}_1)-63 p+588]. 
			\end{split}
		\end{equation}
		\begin{equation}
			\small
			\begin{split}
				F_{H_{5}} = & \frac{5}{2592}[12 (5 p-4) \sin ^4(\text{$\theta $}_1) \sin (4 \text{$\theta $}_2)+4 (13 p+20) \sin ^2(\text{$\theta $}_1) \sin (2 \text{$\theta $}_2)\\
				& +8 (4 p-5) \sin ^4(\text{$\theta $}_1) \cos (4 \text{$\theta $}_2)+8 (10-9 p) \sin ^3(\text{$\theta $}_1) \cos (\text{$\theta $}_1) \cos (3 \text{$\theta $}_2)\\
				& +16 (13 p-7) \sin ^3(\text{$\theta $}_1) \cos (\text{$\theta $}_1) \sin (3 \text{$\theta $}_2)+4 (14-17 p) \sin ^2(2 \text{$\theta $}_1) \cos (2 \text{$\theta $}_2)\\
				& +4 \cos (2 \text{$\theta $}_1) ((31 p-4) \sin ^2(\text{$\theta $}_1) \sin (2 \text{$\theta $}_2)+4 p-5)+2 \sin (2 \text{$\theta $}_1) \cos (\text{$\theta $}_2) ((19 p-14) \cos (2 \text{$\theta $}_1)\\
				& +5 p-2) +4 \sin (2 \text{$\theta $}_1) \sin (\text{$\theta $}_2) ((p+7) \cos (2 \text{$\theta $}_1)-37 p+41)+(38 p-39) \cos (4 \text{$\theta $}_1)\\
				& -54 p+315].
			\end{split}
		\end{equation}
		
			Now, we specify a quantum state by fixing $\ket{\phi}_a = \ket{+}$, which is in equation (\ref{plus_state}). It makes the teleportation fidelity $F_{H_{1}}, F_{H_{2}}, \dots F_{H_{5}}$ linear functions of $p$. We plot them in Figure \ref{fig:Qutritflip noise}. It can be noted from the figure that $F_{H_{1}}$ increases monotonically with respect to $p$, which the other decreases. It indicates that the hypergraph state $H_{1}$ is more efficient than the others.
			\begin{figure}
				\includegraphics[scale = 0.76]{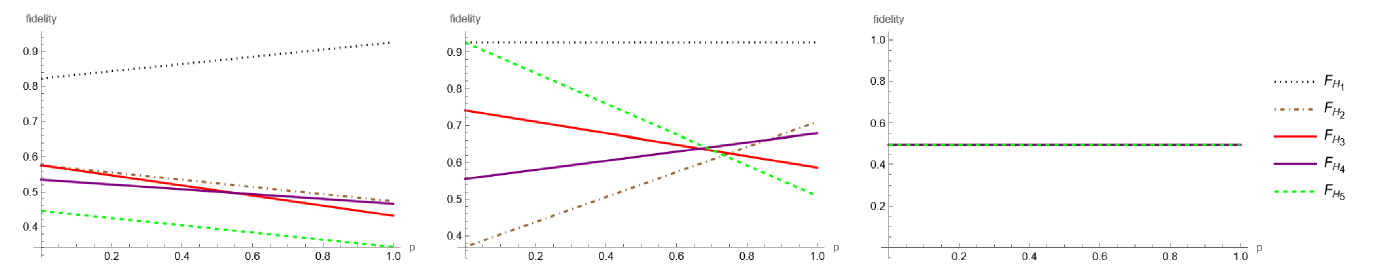}
				\caption{For all five hypergraph states we plot the teleportation fidelity with respect to the channel parameter $p$, in case of qutrit-flip noise. Three subfigues corresponds different choices of states $\ket{\phi}_a$ to teleport. We consider $\ket{\phi}_a = \ket{+}$, $\ket{\phi}_a = \frac{1}{\sqrt{2}}(\ket{0} + \ket{2})$, and $\ket{\phi}_a = \ket{0}$ in the left, middle, and right subfigure, respectively.}
				\label{fig:Qutritflip noise}
			\end{figure}
		   
			If we consider $\theta_1 = 3 \theta_2$, then the teleportation fidelity is a function of two variables, which are the state parameter $\theta_2$ and channel parameter $p$. With respect to the state and channel parameters we can plot the teleportation fidelity. The surface plots are available in Figure \ref{fig:Qutrit flip noise}.
			\begin{figure}[hbt!]
			   	\centering
			   	\begin{subfigure}{0.32\textwidth}
			   		\includegraphics[scale = .37]{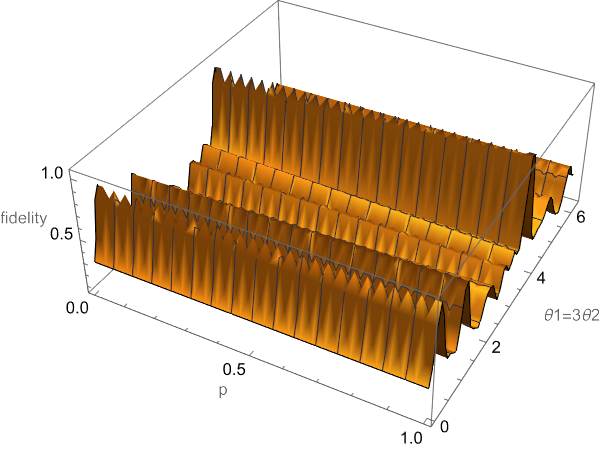}
			   		\caption{$F_{H_1}$}
			   	\end{subfigure}
			   	\begin{subfigure}{0.32\textwidth}
			   		\includegraphics[scale = .37]{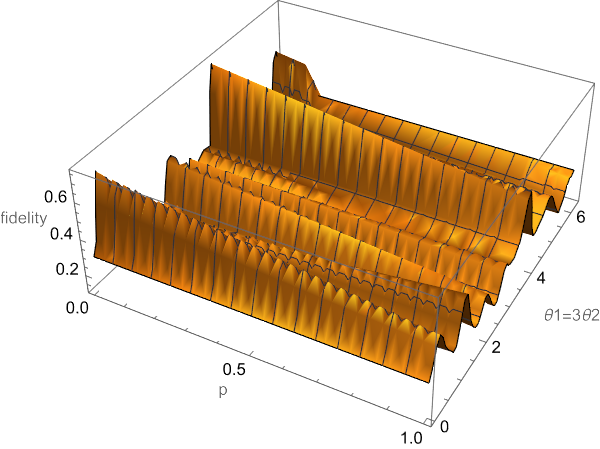}
			   		\caption{$F_{H_2}$}
			   	\end{subfigure}
			   	\begin{subfigure}{0.32\textwidth}
			   		\includegraphics[scale = .35]{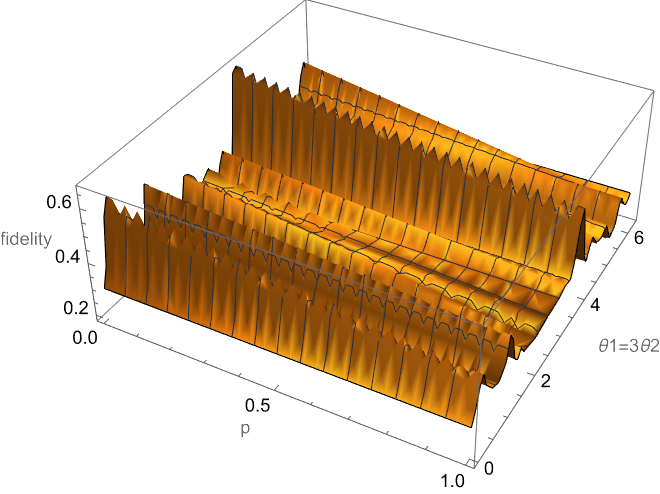}
			   		\caption{$F_{H_3}$}
			   	\end{subfigure}
			   	\\
			   	\begin{subfigure}{0.4\textwidth}
			   		\includegraphics[scale = .37]{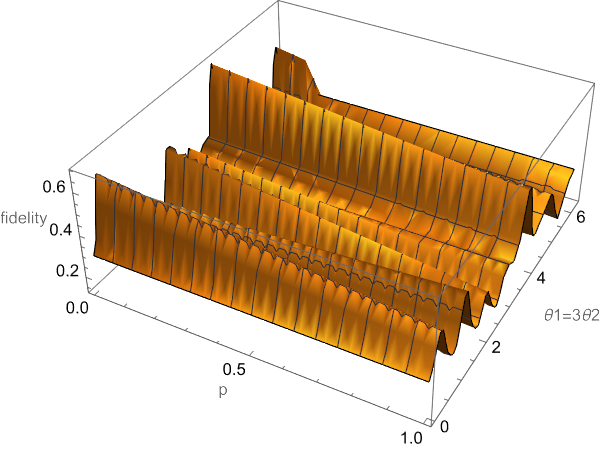}
			   		\caption{$F_{H_4}$}
			   	\end{subfigure}
			   	\begin{subfigure}{0.4\textwidth}
			   		\includegraphics[scale = .35]{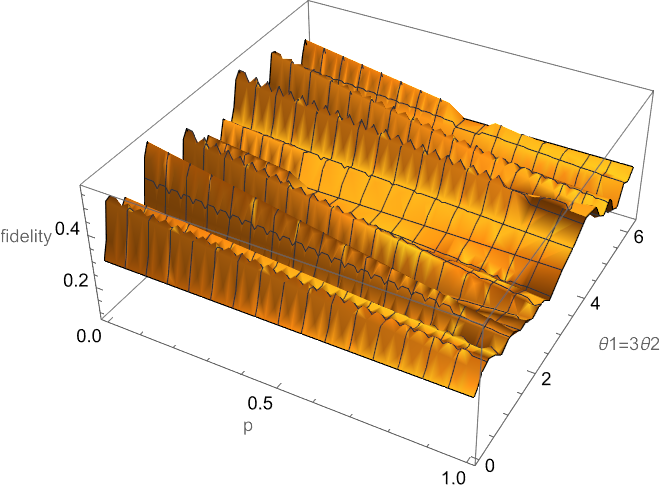}
			   		\caption{$F_{H_5}$}
			   	\end{subfigure}
			   	\caption{The quantum teleportation fidelity under qutrit-flip noise is plotted with respect to the channel parameter $p$ and the state parameter $\theta_2$. Different subfigure is generated for different hypergraph states.}
			   	\label{fig:Qutrit flip noise}
		   \end{figure}

		\subsection{Qutrit-phase-flip noise:}
		
			Next we consider the qutrit-phase-flip noise. This is a generalization of the qubit-phase-flip noise for qutrits. Recall that the qubit-phase-flip noise is represented by the kraus operators
			\begin{center}
				\begin{tabular}{c c}
					$K_1 =\sqrt{1-p} I_2,$ & $K_2 =\sqrt {p} \sigma_z$.
				\end{tabular}
			\end{center}
			Generalizing the qubit-phase flip noise \cite{fonseca2019high, dutta2023qudit, he2020effect}, we get the following Kraus operators representing a qutrit-phase-flip noise: 
			\begin{equation} \label{Qutrit phase flip noise}
				{K_{r,0}} = \begin{cases}
					\sqrt{\frac{1-p}{3}} I_3 & \text{when}~ r = 0, \\ 
					\sqrt{\frac {p} {6}} W_{r,0} & \text{for}~ s = 0, 1\leq r \leq 2, 
				\end{cases}   
			\end{equation}
			where the Weyl operators $W_{0,0}, W_{1,0}, \text{and}~ W_{2,0}$ are defined in equation (\ref{Weyl operator}). Using Lemma \ref{multiqutrit_channel_lemma}, we can construct three-qutrit quantum channels that can be represented by the following Kraus operators:\\
			\small
			\begin{tabular}{p{5.5cm} p{5.5cm} p{5.5cm}}
				\centering
				$K_1=\sqrt{\frac{1-{p}}{3}} (K_{0,0} \otimes I_3 \otimes I_3)$,& $K_2=\sqrt{\frac{{p}}{6}} (K_{1,0} \otimes I_3 \otimes I_3)$,& $K_3=\sqrt{\frac{{p}}{6}} (K_{2,0} \otimes I_3 \otimes I_3$,\\
				$K_4=\sqrt{\frac{1-{p}}{3}} (I_3 \otimes K_{0,0} \otimes I_3)$,& $K_5=\sqrt{\frac{{p}}{6}} (I_3 \otimes K_{1,0} \otimes I_3)$,& $K_6=\sqrt{\frac{{p}}{6}} (I_3 \otimes K_{2,0} \otimes I_3)$,\\
				$K_7=\sqrt{\frac{1-{p}}{3}} (I_3 \otimes I_3 \otimes K_{0,0})$,& $K_8=\sqrt{\frac{{p}}{6}} (I_3 \otimes I_3 \otimes K_{1,0})$,& $K_9=\sqrt{\frac{{p}}{6}} (I_3 \otimes I_3 \otimes K_{2,0})$.
			\end{tabular}

			Now we follow the Procedure \ref{teleportation_procedure} to transfer an arbitrary qutrit quantum state via one of the hyperrgraph states. The teleportation fidelity for different hypergraph states are as follows:
			\begin{equation}
				\small
				\begin{split}
					F_{H_{1}} = & \frac{5}{5184}[2 | p-1|  \{8 \sin ^2(\text{$\theta $}_1) (16 \cos ^2(\text{$\theta $}_1) \cos (2 \text{$\theta $}_2)+\sin ^2(\text{$\theta $}_1) (3 \sin (4 \text{$\theta $}_2)+\cos (4 \text{$\theta $}_2)) \\
					& +\sin (2 \text{$\theta $}_1) (17 \sin (3 \text{$\theta $}_2)  -4 \cos (3 \text{$\theta $}_2))+55 \sin (2 \text{$\theta $}_2)) + 4 \cos (2 \text{$\theta $}_1) (74 \sin ^2(\text{$\theta $}_1) \sin (2 \text{$\theta $}_2)\\
					& +\sin (2 \text{$\theta $}_1) (7 \sin (\text{$\theta $}_2)-4 \cos (\text{$\theta $}_2))+1)+4 \sin (2 \text{$\theta $}_1) (41 \sin (\text{$\theta $}_2)+20 \cos (\text{$\theta $}_2))-81 \cos (4 \text{$\theta $}_1)+333\} \\
					& +p \{64 \sin (4 \text{$\theta $}_1) \sin (\text{$\theta $}_2)-8 \sin ^2(\text{$\theta $}_1) (16 \cos ^2(\text{$\theta $}_1) \cos (2 \text{$\theta $}_2)+\sin ^2(\text{$\theta $}_1) \cos (4 \text{$\theta $}_2)+4 \sin (2 \text{$\theta $}_1) (2 \sin (3 \text{$\theta $}_2) \\
					& +\cos (\text{$\theta $}_2)-\cos (3 \text{$\theta $}_2))+18 \sin (2 \text{$\theta $}_2))-4 \cos (2 \text{$\theta $}_1) (68 \sin ^2(\text{$\theta $}_1) \sin (2 \text{$\theta $}_2)+1)+81 \cos (4 \text{$\theta $}_1)+435 \} ].
				\end{split}
			\end{equation}
			\begin{equation}
				\small
				\begin{split}
					F_{H_{2}} = & \frac{5}{2592} [p \{-6 \sin (2 \text{$\theta $}_1) \sin (\text{$\theta $}_2)+35 \sin (4 \text{$\theta $}_1) \sin (\text{$\theta $}_2)+8 \sin ^2(\text{$\theta $}_1) (\sin (2 \text{$\theta $}_2) (3-7 \sin (2 \text{$\theta $}_1) \sin (\text{$\theta $}_2)) \\
					& +\sin ^2(\text{$\theta $}_1) \cos (4 \text{$\theta $}_2)-\cos (\text{$\theta $}_1) (11 \sin (\text{$\theta $}_1) \sin (3 \text{$\theta $}_2)+26 \cos (\text{$\theta $}_1) \cos (2 \text{$\theta $}_2))) \\
					& +\cos (2 \text{$\theta $}_1) (4-40 \sin ^2(\text{$\theta $}_1) \sin (2 \text{$\theta $}_2))+9 \cos (4 \text{$\theta $}_1)+243\}-2 | p-1|  \{-58 \sin (2 \text{$\theta $}_1) \sin (\text{$\theta $}_2) \\
					& -4 \sin (\text{$\theta $}_1) (4 \sin (\text{$\theta $}_1) \sin (2 \text{$\theta $}_2)+52 \sin (\text{$\theta $}_1) \cos ^2(\text{$\theta $}_1) \cos (2 \text{$\theta $}_2)-2 \sin ^3(\text{$\theta $}_1) \cos (4 \text{$\theta $}_2) \\
					& +\cos (\text{$\theta $}_1) (2 \sin ^2(\text{$\theta $}_1) (29 \sin (3 \text{$\theta $}_2)-7 \cos (3 \text{$\theta $}_2))+11 \cos (\text{$\theta $}_2)))+\cos (2 \text{$\theta $}_1) (80 \sin ^2(\text{$\theta $}_1) \sin (2 \text{$\theta $}_2) \\
					& +38 \sin (2 \text{$\theta $}_1) (\cos (\text{$\theta $}_2)-\sin (\text{$\theta $}_2))+4)+9 \cos (4 \text{$\theta $}_1)-141\} ].
				\end{split}
			\end{equation}
			\begin{equation}
				\small
				\begin{split}
					F_{H_{3}} = & \frac{5}{5184}[2 | p-1|  \{8 \sin ^2(\text{$\theta $}_1) (28 \cos ^2(\text{$\theta $}_1) \cos (2 \text{$\theta $}_2)-\sin ^2(\text{$\theta $}_1) (3 \sin (4 \text{$\theta $}_2)+5 \cos (4 \text{$\theta $}_2)) \\
					& +\sin (2 \text{$\theta $}_1) (8 \sin (3 \text{$\theta $}_2)-7 \cos (3 \text{$\theta $}_2))+13 \sin (2 \text{$\theta $}_2))+4 \cos (2 \text{$\theta $}_1) (14 \sin ^2(\text{$\theta $}_1) \sin (2 \text{$\theta $}_2)-5) \\
					& +2 (22 \sin (2 \text{$\theta $}_1)-19 \sin (4 \text{$\theta $}_1)) \cos (\text{$\theta $}_2)+64 \sin (2 \text{$\theta $}_1) (\cos (2 \text{$\theta $}_1)+2) \sin (\text{$\theta $}_2)-39 \cos (4 \text{$\theta $}_1)+315\} \\
					& +p \{64 \sin (4 \text{$\theta $}_1) \sin (\text{$\theta $}_2)+8 \sin ^2(\text{$\theta $}_1) (-28 \cos ^2(\text{$\theta $}_1) \cos (2 \text{$\theta $}_2)+\sin ^2(\text{$\theta $}_1) (5 \cos (4 \text{$\theta $}_2)-6 \sin (4 \text{$\theta $}_2)) \\
					& +2 \sin (2 \text{$\theta $}_1) (-4 \sin (3 \text{$\theta $}_2)-5 \cos (\text{$\theta $}_2)+5 \cos (3 \text{$\theta $}_2))+6 \sin (2 \text{$\theta $}_2))+4 \cos (2 \text{$\theta $}_1) (4 \sin ^2(\text{$\theta $}_1) \sin (2 \text{$\theta $}_2)+5) \\
					& +39 \cos (4 \text{$\theta $}_1)+453\}].
				\end{split}
			\end{equation}
			\begin{equation}
				\small
				\begin{split}
					F_{H_{4}} = & \frac{5}{2592}[p \{6 \sin (2 \text{$\theta $}_1) \sin (\text{$\theta $}_2)+29 \sin (4 \text{$\theta $}_1) \sin (\text{$\theta $}_2)+8 \sin ^2(\text{$\theta $}_1) (-20 \cos ^2(\text{$\theta $}_1) \cos (2 \text{$\theta $}_2)+\sin ^2(\text{$\theta $}_1) \cos (4 \text{$\theta $}_2) \\
					& +\sin (\text{$\theta $}_1) \cos (\text{$\theta $}_1) (-5 \sin (3 \text{$\theta $}_2)-7 \cos (\text{$\theta $}_2)+7 \cos (3 \text{$\theta $}_2))+6 \sin (2 \text{$\theta $}_2))+\cos (2 \text{$\theta $}_1) (4-16 \sin ^2(\text{$\theta $}_1) \sin (2 \text{$\theta $}_2)) \\
					& +15 \cos (4 \text{$\theta $}_1)+237\}-2 | p-1|  \{8 \sin ^2(\text{$\theta $}_1) (-20 \cos ^2(\text{$\theta $}_1) \cos (2 \text{$\theta $}_2)+\sin ^2(\text{$\theta $}_1) \cos (4 \text{$\theta $}_2) \\
					& +\sin (\text{$\theta $}_1) \cos (\text{$\theta $}_1) (7 \cos (3 \text{$\theta $}_2)-17 \sin (3 \text{$\theta $}_2))+\sin (2 \text{$\theta $}_2))+19 \sin (4 \text{$\theta $}_1) \cos (\text{$\theta $}_2) \\
					& -2 \sin (2 \text{$\theta $}_1) (29 \sin (\text{$\theta $}_2)+11 \cos (\text{$\theta $}_2))+\cos (2 \text{$\theta $}_1) (4 \sin (\text{$\theta $}_1) \sin (\text{$\theta $}_2) (52 \sin (\text{$\theta $}_1) \cos (\text{$\theta $}_2)-19 \cos (\text{$\theta $}_1))+4) \\
					& +15 \cos (4 \text{$\theta $}_1)-147\}].
				\end{split}
			\end{equation}
			\begin{equation}
				\small
				\begin{split}
					F_{H_{5}} = & \frac{5}{5184} [-2 | p-1|  \{8 \sin ^2(\text{$\theta $}_1) (-28 \cos ^2(\text{$\theta $}_1) \cos (2 \text{$\theta $}_2)+\sin ^2(\text{$\theta $}_1) (6 \sin (4 \text{$\theta $}_2)+5 \cos (4 \text{$\theta $}_2)) \\
					& +\sin (2 \text{$\theta $}_1) (7 \sin (3 \text{$\theta $}_2)-5 \cos (3 \text{$\theta $}_2))-10 \sin (2 \text{$\theta $}_2))+4 \cos (2 \text{$\theta $}_1) (4 \sin ^2(\text{$\theta $}_1) \sin (2 \text{$\theta $}_2) \\
					& +7 \sin (2 \text{$\theta $}_1) (\cos (\text{$\theta $}_2)-\sin (\text{$\theta $}_2))+5)+4 \sin (2 \text{$\theta $}_1) (\cos (\text{$\theta $}_2)-41 \sin (\text{$\theta $}_2))+39 \cos (4 \text{$\theta $}_1)-315\} \\
					& +p \{-224 \sin ^2(\text{$\theta $}_1) \cos ^2(\text{$\theta $}_1) \cos (2 \text{$\theta $}_2)+20 \cos (2 \text{$\theta $}_1)+39 \cos (4 \text{$\theta $}_1)+453\} \\
					& +8 p \sin (\text{$\theta $}_1) \{-6 \sin ^3(\text{$\theta $}_1) \sin (4 \text{$\theta $}_2)-8 \sin (2 \text{$\theta $}_1) \sin (\text{$\theta $}_1) \sin (3 \text{$\theta $}_2)+4 (2 \sin (\text{$\theta $}_1)+\sin (3 \text{$\theta $}_1)) \sin (2 \text{$\theta $}_2) \\
					& +5 \sin ^3(\text{$\theta $}_1) \cos (4 \text{$\theta $}_2)+8 \cos (\text{$\theta $}_1) \sin (\text{$\theta $}_2) (\sin ^2(\text{$\theta $}_1) \sin (2 \text{$\theta $}_2)+\cos (2 \text{$\theta $}_1)+3)\}].
				\end{split}
			\end{equation}
			
			Now we plot the teleportation fidelity for different hypergraph states. For a number of specific state $\ket{\phi}_a$ we study the efficiency of different hypergraph states in in Figure \ref{fig:Qutrit-phase-flip noise}. Figure \ref{fig:Qutrit_phase_flip_noise} plots the teleportation fidelity with respect to the channel parameter $p$ and state parameter $\theta_2$. Note that, we consider a linear relationship $\theta_1 = 3 \theta_2$ to reduce a state parameter as discussed in equation (\ref{quantum state vector_1}). 
			\begin{figure}
				\centering 
	       		\includegraphics[scale = .79]{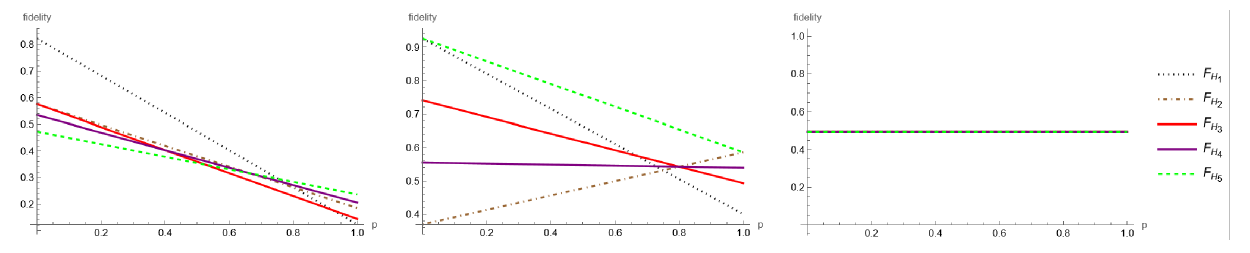}
	     		\caption{For all five hypergraph states we plot the teleportation fidelity with respect to the channel parameter $p$, in case of qutrit-phase-flip noise. Three subfigures consider different choices of states $\ket{\phi}_a$ to teleport.  We consider $\ket{\phi}_a = \ket{+}$, $\ket{\phi}_a = \frac{1}{\sqrt{2}}(\ket{0} + \ket{2})$, and $\ket{\phi}_a = \ket{0}$ in the left, middle, and right subfigure, respectively.}  		
	       		\label{fig:Qutrit-phase-flip noise}
	       	\end{figure}
       	
	       	\begin{figure}[hbt!]
       		\centering
       		\begin{subfigure}{0.32\textwidth}
       			\includegraphics[scale = .37]{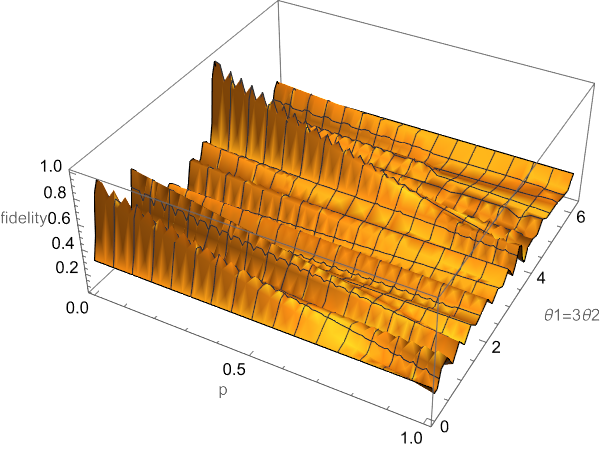}
       			\caption{$F_{H_1}$}
       		\end{subfigure}
       		\begin{subfigure}{0.32\textwidth}
       			\includegraphics[scale = .37]{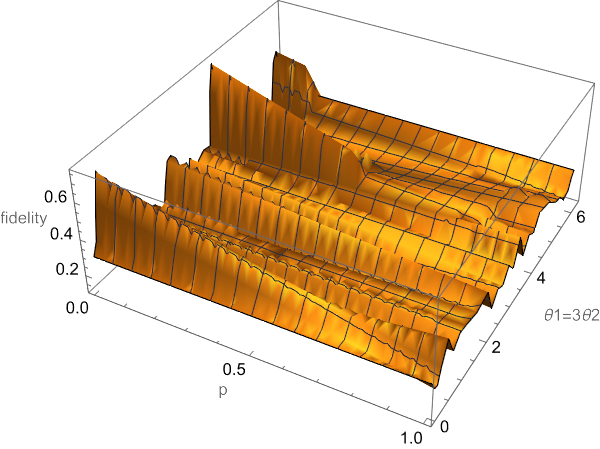}
       			\caption{$F_{H_2}$}
       		\end{subfigure}
       		\begin{subfigure}{0.32\textwidth}
       			\centering
       			\includegraphics[scale = .37]{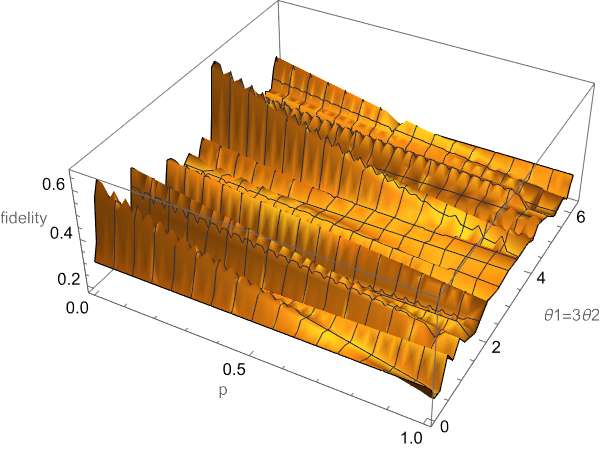}
       			\caption{$F_{H_3}$}
       		\end{subfigure}
       		\\
       		\begin{subfigure}{0.4\textwidth}
       			\centering
       			\includegraphics[scale = .37]{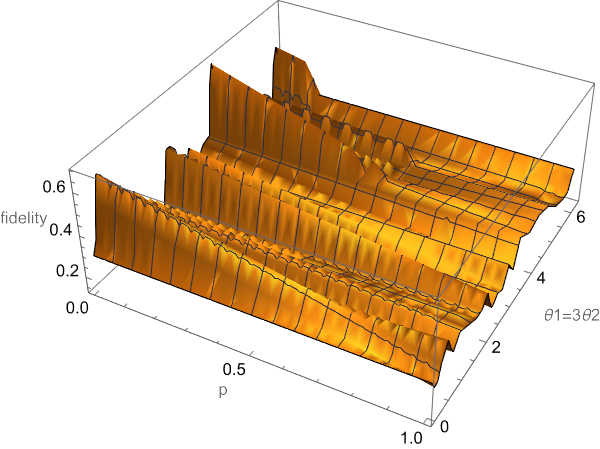}
       			\caption{$F_{H_4}$}
       		\end{subfigure}
       		\begin{subfigure}{0.4\textwidth}
       			\centering
       			\includegraphics[scale = .37]{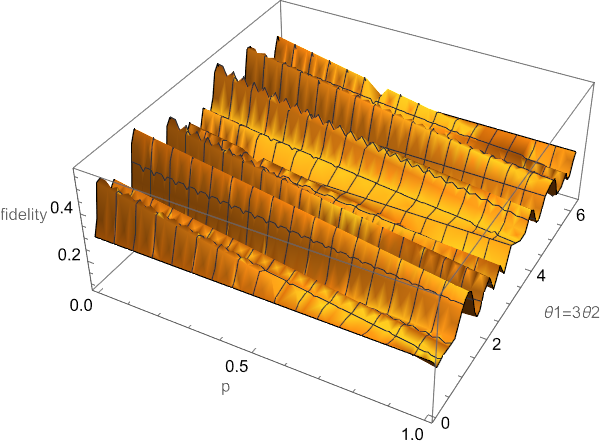}
       			\caption{$F_{H_5}$}
       		\end{subfigure}
       		\caption{The quantum teleportation fidelity under qutrit-phase-flip noise is plotted with respect to the channel parameter $p$ and the state parameter $\theta_2$. Different subfigure is generated for different hypergraphs. Here we consider the qutrit-phase-flip noise.}
       		\label{fig:Qutrit_phase_flip_noise}
    	   	\end{figure}

		\subsection{Depolarizing noise:}
	
			Now we consider the teleportation under depolarizing channel. Recall that the qubit depolarizing noise is represented by the following Kraus operators
		\begin{center} 
		\begin{tabular}{c  c c c}
			$K_1=\sqrt{\frac{4-{3p}}{4}} I_2$, & $K_2=\sqrt{\frac{p}{4}} \sigma_x$, 
			$K_3=\sqrt{\frac{p}{4}} \sigma_y$, & $K_4=\sqrt{\frac{p}{4}} \sigma_z$.
		\end{tabular}
		\end{center}
		The Kraus operators representing depolarizing channel for qutrits are \cite{fonseca2019high, dutta2023qudit}
		\begin{equation} \label{Depolarizing noise}
			{K_{r,s}} = \begin{cases}
				 \sqrt{\frac{9-{8p}}{27}} I_3 & \text{when}~  r = 0, s=0, \\
				 \sqrt{\frac{p}{27}}  W_{r,s} & \text{for}~  0\leq r,s \leq 2  ~\text{and}~ (r,s) \neq (0,0);
			\end{cases}
		\end{equation}
		 where the Weyl operators $W_{0,0}, W_{0,1}, W_{0,2}, W_{1,0}, W_{1,1}, W_{1,2}, W_{2,0}, W_{2,1}, \text{and}~ W_{2,2}$ are defined in equation (\ref{Weyl operator}). Using Lemma \ref{multiqutrit_channel_lemma}, we can construct three-qutrit channel which can be described by the following Kraus operators:
		 
		\small
		\begin{tabular}{p{5.5cm} p{5.5cm} p{5.5cm}}
			\centering
				$K_1=\sqrt{\frac{9-{8p}}{27}} (K_{0,0}\otimes I_3 \otimes I_3)$, & $K_2=\sqrt{\frac{{p}}{27}} (K_{0,1} \otimes I_3 \otimes I_3)$,& $K_3=\sqrt{\frac{{p}}{27}} (K_{0,2} \otimes I_3 \otimes I_3)$,\\
				$K_4=\sqrt{\frac{{p}}{27}} (K_{1,0} \otimes I_3 \otimes I_3)$,& $K_5=\sqrt{\frac{{p}}{27}} (K_{1,1} \otimes I_3 \otimes I_3)$,& $K_6=\sqrt{\frac{{p}}{27}} (K_{1,2} \otimes I_3 \otimes I_3)$,\\
				$K_7=\sqrt{\frac{{p}}{27}} (K_{2,0} \otimes I_3 \otimes I_3)$,& $K_8=\sqrt{\frac{{p}}{27}} (K_{2,1} \otimes I_3 \otimes I_3)$,& $K_9=\sqrt{\frac{{p}}{27}} (K_{2,2} \otimes I_3 \otimes I_3)$,\\
				$K_{10}=\sqrt{\frac{9-{8p}}{27}} (I_3\otimes K_{0,0} \otimes I_3)$,& $K_{11}=\sqrt{\frac{{p}}{27}} (I_3 \otimes K_{0,1} \otimes I_3)$,& $K_{12}=\sqrt{\frac{{p}}{27}} (I_3 \otimes K_{0,2} \otimes I_3)$,\\
				$K_{13}=\sqrt{\frac{{p}}{27}} (I_3 \otimes K_{1,0} \otimes I_3)$,& $K_{14}=\sqrt{\frac{{p}}{27}} (I_3 \otimes K_{1,1} \otimes I_3)$,& $K_{15}=\sqrt{\frac{{p}}{27}} (I_3 \otimes K_{1,2} \otimes I_3)$,\\
				$K_{16}=\sqrt{\frac{{p}}{27}} (I_3 \otimes K_{2,0} \otimes I_3)$,& $K_{17}=\sqrt{\frac{{p}}{27}} (I_3 \otimes K_{2,1} \otimes I_3)$,& $K_{18}=\sqrt{\frac{{p}}{27}} (I_3 \otimes K_{2,2} \otimes I_3)$,\\
				$K_{19}=\sqrt{\frac{9-{8p}}{27}} (I_3\otimes I_3 \otimes K_{0,0})$,& $K_{20}=\sqrt{\frac{{p}}{27}} (I_3\otimes I_3 \otimes K_{0,1})$,& $K_{21}=\sqrt{\frac{{p}}{27}} (I_3\otimes I_3 \otimes K_{0,2})$,\\
				$K_{22}=\sqrt{\frac{{p}}{27}} (I_3\otimes I_3 \otimes K_{1,0})$,& $K_{23}=\sqrt{\frac{{p}}{27}} (I_3\otimes I_3 \otimes K_{1,1})$,& $K_{24}=\sqrt{\frac{{p}}{27}} (I_3\otimes I_3 \otimes K_{1,2})$,\\
				$K_{25}=\sqrt{\frac{{p}}{27}} (I_3\otimes I_3 \otimes K_{2,0})$,& $K_{26}=\sqrt{\frac{{p}}{27}} (I_3\otimes I_3 \otimes K_{2,1})$,& $K_{27}=\sqrt{\frac{{p}}{27}} (I_3\otimes I_3 \otimes K_{2,2})$.
			\end{tabular}
		
		We now apply the quantum teleportation process described in Procedure \ref{teleportation_procedure} to transmit an qutrit quantum state as stated in equation (\ref{quantum state vector}). We find the following teleportation fidelity for different hypergraph states as follows: 
		
        \begin{equation}
        	\small
			\begin{split}
				F_{H_{1}} = &\frac{5}{23328}[| 9-8 p|  \{8 \sin ^2(\text{$\theta $}_1) (16 \cos ^2(\text{$\theta $}_1) \cos (2 \text{$\theta $}_2)+\sin ^2(\text{$\theta $}_1) (3 \sin (4 \text{$\theta $}_2)+\cos (4 \text{$\theta $}_2)) \\
				& +\sin (2 \text{$\theta $}_1) (17 \sin (3 \text{$\theta $}_2)-4 \cos (3 \text{$\theta $}_2))+55 \sin (2 \text{$\theta $}_2))+4 \cos (2 \text{$\theta $}_1) (74 \sin ^2(\text{$\theta $}_1) \sin (2 \text{$\theta $}_2) \\
				& +\sin (2 \text{$\theta $}_1) (7 \sin (\text{$\theta $}_2)-4 \cos (\text{$\theta $}_2))+1)+4 \sin (2 \text{$\theta $}_1) (41 \sin (\text{$\theta $}_2)+20 \cos (\text{$\theta $}_2))-81 \cos (4 \text{$\theta $}_1)+333\} \\
				& +p \{-128 \sin ^2(\text{$\theta $}_1) \cos ^2(\text{$\theta $}_1) \cos (2 \text{$\theta $}_2)-4 \cos (2 \text{$\theta $}_1)+81 \cos (4 \text{$\theta $}_1)+1971\} \\
				& +8 p \sin (\text{$\theta $}_1) \{ \cos (\text{$\theta $}_1) (4 (\cos (2 \text{$\theta $}_1)+7) \cos (\text{$\theta $}_2)+8 \sin ^2(\text{$\theta $}_1) \cos (3 \text{$\theta $}_2)+(29 \cos (2 \text{$\theta $}_1)+139) \sin (\text{$\theta $}_2)) \\
				& -\sin (\text{$\theta $}_1) (5 \sin (2 \text{$\theta $}_1) \sin (3 \text{$\theta $}_2)+\sin ^2(\text{$\theta $}_1) (9 \sin (4 \text{$\theta $}_2)+\cos (4 \text{$\theta $}_2))+(43 \cos (2 \text{$\theta $}_1)-35) \sin (2 \text{$\theta $}_2))\}].
			\end{split}
		\end{equation}
		\begin{equation}
			\small
			\begin{split}
				F_{H_{2}} = & \frac{5}{11664}[| 9-8 p|  \{58 \sin (2 \text{$\theta $}_1) \sin (\text{$\theta $}_2)+8 \sin ^2(\text{$\theta $}_1) (26 \cos ^2(\text{$\theta $}_1) \cos (2 \text{$\theta $}_2)-\sin ^2(\text{$\theta $}_1) \cos (4 \text{$\theta $}_2) \\
				& +\sin (\text{$\theta $}_1) \cos (\text{$\theta $}_1) (29 \sin (3 \text{$\theta $}_2)-7 \cos (3 \text{$\theta $}_2))+2 \sin (2 \text{$\theta $}_2))-4 \cos (2 \text{$\theta $}_1) (20 \sin ^2(\text{$\theta $}_1) \sin (2 \text{$\theta $}_2)+1) \\
				& +44 \sin (\text{$\theta $}_1) \cos (\text{$\theta $}_1) \cos (\text{$\theta $}_2)+19 \sin (4 \text{$\theta $}_1) (\sin (\text{$\theta $}_2)-\cos (\text{$\theta $}_2))-9 \cos (4 \text{$\theta $}_1)+141\} \\
				& +4 p \sin (\text{$\theta $}_1) \{-52 \sin (\text{$\theta $}_1) \cos ^2(\text{$\theta $}_1) \cos (2 \text{$\theta $}_2)+2 \sin ^3(\text{$\theta $}_1) \cos (4 \text{$\theta $}_2)+\cos (\text{$\theta $}_1) ((19 \cos (2 \text{$\theta $}_1)+13) \cos (\text{$\theta $}_2)\\
				& +2 \sin ^2(\text{$\theta $}_1) (7 \cos (3 \text{$\theta $}_2)-11 \sin (3 \text{$\theta $}_2))+(35 \cos (2 \text{$\theta $}_1)+61) \sin (\text{$\theta $}_2)) \\
				& +4 \sin (\text{$\theta $}_1) (11-4 \cos (2 \text{$\theta $}_1)) \sin (2 \text{$\theta $}_2)\} +p (4 \cos (2 \text{$\theta $}_1)+9 \cos (4 \text{$\theta $}_1)+1011)].
			\end{split}
		\end{equation}
		\begin{equation}
			\small
			\begin{split}
				F_{H_{3}} = & \frac{5}{23328}[| 9-8 p|  \{8 \sin ^2(\text{$\theta $}_1) (28 \cos ^2(\text{$\theta $}_1) \cos (2 \text{$\theta $}_2)-\sin ^2(\text{$\theta $}_1) (3 \sin (4 \text{$\theta $}_2)+5 \cos (4 \text{$\theta $}_2)) \\
				& +\sin (2 \text{$\theta $}_1) (8 \sin (3 \text{$\theta $}_2)-7 \cos (3 \text{$\theta $}_2))+13 \sin (2 \text{$\theta $}_2))+4 \cos (2 \text{$\theta $}_1) (14 \sin ^2(\text{$\theta $}_1) \sin (2 \text{$\theta $}_2)-5)\\
				& +2 (22 \sin (2 \text{$\theta $}_1)-19 \sin (4 \text{$\theta $}_1)) \cos (\text{$\theta $}_2)+64 \sin (2 \text{$\theta $}_1) (\cos (2 \text{$\theta $}_1)+2) \sin (\text{$\theta $}_2)-39 \cos (4 \text{$\theta $}_1)+315\}\\
				& +p \{8 \sin ^2(\text{$\theta $}_1) (-28 \cos ^2(\text{$\theta $}_1) \cos (2 \text{$\theta $}_2)+\sin ^2(\text{$\theta $}_1) (5 \cos (4 \text{$\theta $}_2)-9 \sin (4 \text{$\theta $}_2))\\
				& +\sin (2 \text{$\theta $}_1) (7 \cos (3 \text{$\theta $}_2)-2 \sin (3 \text{$\theta $}_2))-13 \sin (2 \text{$\theta $}_2))+4 \cos (2 \text{$\theta $}_1) (-14 \sin ^2(\text{$\theta $}_1) \sin (2 \text{$\theta $}_2)\\
				& +\sin (2 \text{$\theta $}_1) (26 \sin (\text{$\theta $}_2)+19 \cos (\text{$\theta $}_2))+5)-44 \sin (2 \text{$\theta $}_1) (\cos (\text{$\theta $}_2)-2 \sin (\text{$\theta $}_2))+39 \cos (4 \text{$\theta $}_1)+1989\}].
			\end{split}
		\end{equation}
		\begin{equation}
			\small
			\begin{split}
				F_{H_{4}} = & \frac{5}{11664}[| 9-8 p|  \{58 \sin (2 \text{$\theta $}_1) \sin (\text{$\theta $}_2)+4 \sin (\text{$\theta $}_1) (40 \sin (\text{$\theta $}_1) \cos ^2(\text{$\theta $}_1) \cos (2 \text{$\theta $}_2)\\
				& +\cos (\text{$\theta $}_1) (2 \sin ^2(\text{$\theta $}_1) (17 \sin (3 \text{$\theta $}_2)-7 \cos (3 \text{$\theta $}_2))+11 \cos (\text{$\theta $}_2))-2 \sin (\text{$\theta $}_1) (\sin ^2(\text{$\theta $}_1) \cos (4 \text{$\theta $}_2)+\sin (2 \text{$\theta $}_2)))\\
				& -2 \cos (2 \text{$\theta $}_1) (52 \sin ^2(\text{$\theta $}_1) \sin (2 \text{$\theta $}_2)+19 \sin (2 \text{$\theta $}_1) (\cos (\text{$\theta $}_2)-\sin (\text{$\theta $}_2))+2)-15 \cos (4 \text{$\theta $}_1)+147\}\\
				& +4 p \sin (\text{$\theta $}_1) \{-40 \sin (\text{$\theta $}_1) \cos ^2(\text{$\theta $}_1) \cos (2 \text{$\theta $}_2)+2 \sin ^3(\text{$\theta $}_1) \cos (4 \text{$\theta $}_2)+\cos (\text{$\theta $}_1) ((19 \cos (2 \text{$\theta $}_1)+13) \cos (\text{$\theta $}_2)\\
				& +2 \sin ^2(\text{$\theta $}_1) (\sin (3 \text{$\theta $}_2)+7 \cos (3 \text{$\theta $}_2))+(35 \cos (2 \text{$\theta $}_1)+61) \sin (\text{$\theta $}_2))-10 \sin (\text{$\theta $}_1) (\cos (2 \text{$\theta $}_1)-5) \sin (2 \text{$\theta $}_2)\}\\
				& +p (4 \cos (2 \text{$\theta $1})+15 (\cos (4 \text{$\theta $}_1)+67))].
			\end{split}
		\end{equation}
		\begin{equation}
			\small
			\begin{split}
				F_{H_{5}} = & \frac{5}{23328}[p \{8 \sin ^2(\text{$\theta $}_1) (-28 \cos ^2(\text{$\theta $}_1) \cos (2 \text{$\theta $}_2)+5 \sin ^2(\text{$\theta $}_1) \cos (4 \text{$\theta $}_2)+\sin (2 \text{$\theta $}_1) (\sin (3 \text{$\theta $}_2)-5 \cos (3 \text{$\theta $}_2))\\
				& +20 \sin (2 \text{$\theta $}_2))+4 \cos (2 \text{$\theta $}_1) (16 \sin ^2(\text{$\theta $}_1) \sin (2 \text{$\theta $}_2)+\sin (2 \text{$\theta $}_1) (7 \cos (\text{$\theta $}_2)-\sin (\text{$\theta $}_2))+5)\\
				& +4 \sin (2 \text{$\theta $}_1) (25 \sin (\text{$\theta $}_2)+\cos (\text{$\theta $}_2))+39 \cos (4 \text{$\theta $}_1)+1989\}-| 9-8 p|  \{8 \sin ^2(\text{$\theta $}_1) (-28 \cos ^2(\text{$\theta $}_1) \cos (2 \text{$\theta $}_2)\\
				& +\sin ^2(\text{$\theta $}_1) (6 \sin (4 \text{$\theta $}_2)+5 \cos (4 \text{$\theta $}_2))+\sin (2 \text{$\theta $}_1) (7 \sin (3 \text{$\theta $}_2)-5 \cos (3 \text{$\theta $}_2))-10 \sin (2 \text{$\theta $}_2))\\
				& +4 \cos (2 \text{$\theta $}_1) (4 \sin ^2(\text{$\theta $}_1) \sin (2 \text{$\theta $}_2)+7 \sin (2 \text{$\theta $}_1) (\cos (\text{$\theta $}_2)-\sin (\text{$\theta $}_2))+5)+4 \sin (2 \text{$\theta $}_1) (\cos (\text{$\theta $}_2)-41 \sin (\text{$\theta $}_2))\\
				& +39 \cos (4 \text{$\theta $}_1)-315\}].
			\end{split}
		\end{equation}
		
		Now we plot the fidelity $F_{H_1}, F_{H_2}, F_{H_3}, F_{H_4}$, and $F_{H_5}$ as a function of channel parameter $p$, for a number of states in Figure \ref{fig:Depolarizing}. Also, in Figure \ref{fig:Depolarizing noise}, we depict surface plots representing the dependency of fidelity with respect to the state and channel parameters.
		
		\begin{figure}
			\includegraphics[scale = .78]{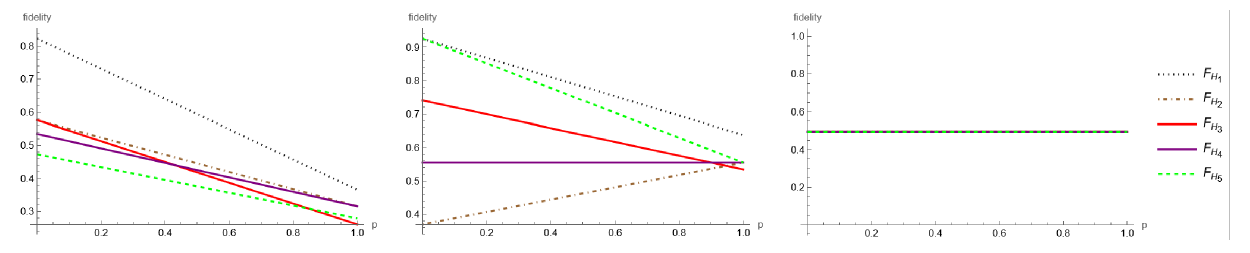}
			\caption{For all five hypergraph states we plot the teleportation fidelity with respect to the channel parameter $p$,  in case of depolarizing quantum state vector noise. Three subfigures consider different states $\ket{\phi}_a$ to teleport.  We consider $\ket{\phi}_a = \ket{+}$, $\ket{\phi}_a = \frac{1}{\sqrt{2}}(\ket{0} + \ket{2})$, and $\ket{\phi}_a = \ket{0}$ in the left, middle, and right subfigure, respectively.}
			\label{fig:Depolarizing}
		\end{figure}
		
		\begin{figure}[hbt!]
			\centering
			\begin{subfigure}{0.32\textwidth}
				\centering
				\includegraphics[scale = .39]{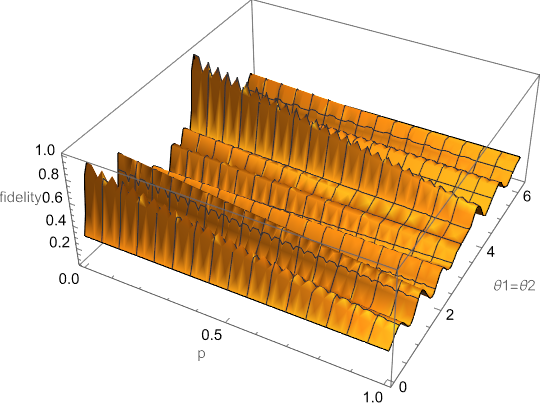}
				\caption{$F_{H_1}$}
			\end{subfigure}
			\begin{subfigure}{0.32\textwidth}
				\centering
				\includegraphics[scale = .37]{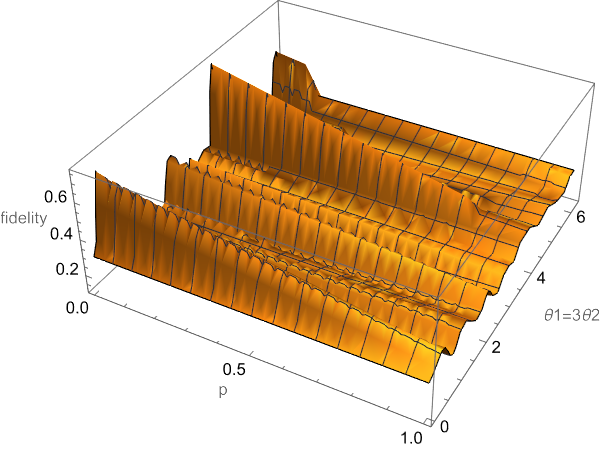}
				\caption{$F_{H_2}$}
			\end{subfigure}
			\begin{subfigure}{0.32\textwidth}
				\centering
				\includegraphics[scale = .37]{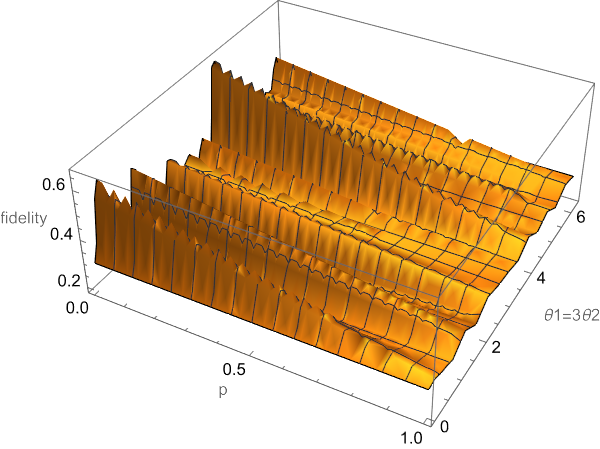}
				\caption{$F_{H_3}$}
			\end{subfigure}
			\\
			\begin{subfigure}{0.32\textwidth}
				\centering
				\includegraphics[scale = .37]{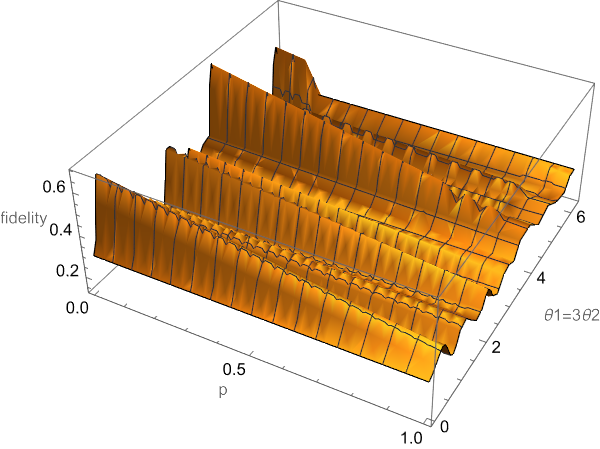}
				\caption{$F_{H_4}$}
			\end{subfigure}
			\begin{subfigure}{0.32\textwidth}
				\centering
				\includegraphics[scale = .37]{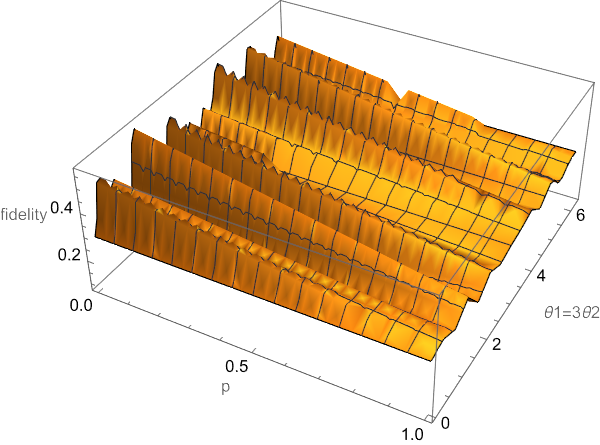}
				\caption{$F_{H_5}$}
			\end{subfigure}
			\caption{The quantum teleportation fidelity under depolarizing noise is plotted with respect to the channel parameter $p$ and the state parameter $\theta_2$. Different subfigure represents the fidelity for different hypergraph states.}
			\label{fig:Depolarizing noise}
		\end{figure}

		\subsection{Markovian and non-Markovian amplitude damping channel:}
		
			The Markovian amplitude damping channel for a qubit is represented by the following Kraus operators
		\begin{center}
			\begin{tabular}{c c}
				$K_{0,0} = \ket{0} \bra{0} + \sqrt{1 - p} \ket{1} \bra{1},$ & $K_{0,1} = \sqrt{p} \ket{0} \bra{1}.$
			\end{tabular}
		\end{center}
		 For a qutrit system we generalize the amplitude damping channel using the Kraus operators  
		\begin{equation} \label{Amplitude damping channel}
			\begin{split}
				&K_{0,0} = \ket{0} \bra{0} + \sqrt{1 - p} \sum_{i=1}^{2} \ket{i} \bra{i},\\
				&K_{0,i} = \sqrt{p} \ket{0} \bra{i}  \text{for}~  1 \le i \le 2,
			\end{split}
		\end{equation}
		where
		\begin{equation}\label{Amplitude damping channel_matrices}
			\begin{split}
				K_{0,0} = \begin{bmatrix}
					1&0&0\\
					0&\sqrt{1-p}&0\\ 
					0&0&\sqrt{1-p}\\
				\end{bmatrix},
				K_{0,1} = \begin{bmatrix}
					0&\sqrt{p}&0\\
					0&0&0\\ 
					0&0&0\\
				\end{bmatrix},
				K_{0,2} = \begin{bmatrix}
					0&0&\sqrt{p}\\
					0&0&0\\ 
					0&0&0\\
				\end{bmatrix}.
			\end{split}
		\end{equation}

		Using Lemma \ref{multiqutrit_channel_lemma} we can construct three-qutrit channel which can be described by the following Kraus operators:
		 
		\small
		\begin{tabular}{p{5.5cm} p{5.5cm} p{5.5cm}}
			\centering
				$K_1=\sqrt{\frac{1}{3}} (K_{0,0} \otimes I_3 \otimes I_3)$,& $K_2=\sqrt{\frac{{1}}{3}} (K_{0,1} \otimes I_3 \otimes I_3)$,& $K_3=\sqrt{\frac{{1}}{3}} (K_{0,2} \otimes I_3 \otimes I_3)$,\\
				$K_4=\sqrt{\frac{{1}}{3}} (I_3 \otimes K_{0,0} \otimes I_3)$,& $K_5=\sqrt{\frac{{1}}{3}} (I_3 \otimes K_{0,1} \otimes I_3)$,& $K_6=\sqrt{\frac{{1}}{3}} (I_3 \otimes E_{0,2} \otimes I_3)$,\\
				$K_7=\sqrt{\frac{{1}}{3}} (I_3 \otimes I_3 \otimes K_{0,0})$,& $K_8=\sqrt{\frac{{1}}{3}} (I_3 \otimes I_3 \otimes K_{0,1})$,& $K_9=\sqrt{\frac{{1}}{3}} (I_3 \otimes I_3 \otimes K_{0,2})$.
			\end{tabular}
	 
	 	Now applying the Procedure \ref{teleportation_procedure}, we teleport any arbitary qutrit state via a noise affected hypergraph state. We work out teleportation fidelity for hypergraph states. We obtain the following expressions:  
	 
		\begin{equation}
			\small
			\begin{split}\label{F_H1}
				F_{H_{1}} = & \frac{1}{15552}[-48 \sqrt{1-p} \sin (2 (\text{$\theta $}_1+\text{$\theta $}_2))-66 \sqrt{1-p} \sin (4 \text{$\theta $}_1+2 \text{$\theta $}_2)+84 \sqrt{1-p} \sin (2 \text{$\theta $}_1-\text{$\theta $}_2) \\
				& +6 \sqrt{1-p} \sin (4 \text{$\theta $}_1-\text{$\theta $}_2)+84 \sqrt{1-p} \sin (2 \text{$\theta $}_1+\text{$\theta $}_2)+6 \sqrt{1-p} \sin (4 \text{$\theta $}_1+\text{$\theta $}_2)\\
				& +48 \sqrt{1-p} \sin (2 \text{$\theta $}_1-2 \text{$\theta $}_2)+66 \sqrt{1-p} \sin (4 \text{$\theta $}_1-2 \text{$\theta $}_2)-64 \sin (\text{$\theta $}_1) \cos ^3(\text{$\theta $}_1) \{(3 \sqrt{1-p}+5) \cos (\text{$\theta $}_2)\\
				& -(7 \sqrt{1-p}+11) \sin (\text{$\theta $}_2)\}+576 \sqrt{1-p} \sin ^2(\text{$\theta $}_1) \cos ^2(\text{$\theta $}_1) \cos (2 \text{$\theta $}_2)-48 \sqrt{1-p} \sin ^4(\text{$\theta $}_1) \cos (4 \text{$\theta $}_2)\\
				& +24 \sin ^4(\text{$\theta $}_1) ((3-4 p) \sin (4 \text{$\theta $}_2)-3 (p-1) \cos (4 \text{$\theta $}_2))+16 \sin ^3(\text{$\theta $}_1) \cos (\text{$\theta $}_1) \{13 \sin (\text{$\theta $}_2)+17 \sin (3 \text{$\theta $}_2)\\
				& -3 \cos (3 \text{$\theta $}_2)+(3 \sqrt{1-p}-11) \cos (\text{$\theta $}_2)+\sqrt{1-p} (38 \sin (\text{$\theta $}_2)+34 \sin (3 \text{$\theta $}_2)-9 \cos (3 \text{$\theta $}_2))\}\\
				& +8 \sin ^2(\text{$\theta $}_1) (2 p \sin (2 \text{$\theta $}_1) (-7 \sin (3 \text{$\theta $}_2)-3 \cos (\text{$\theta $}_2)+3 \cos (3 \text{$\theta $}_2))+3 (25-18 p) \sin (2 \text{$\theta $}_2)\\
				& +12 (3 p-1) \cos (2 \text{$\theta $}_2))-12 \cos (2 \text{$\theta $}_1) \{2 \sin ^2(\text{$\theta $}_1) ((22 p-15) \sin (2 \text{$\theta $}_2)-4 (p-1) \cos (2 \text{$\theta $}_2))\\
				& -21 p+2 \sqrt{1-p}-3\}+8 \sin (2 \text{$\theta $}_1) (29 \cos (\text{$\theta $}_2)+21 p \sin (\text{$\theta $}_2))+4 \sin (4 \text{$\theta $}_1) (3 \cos (\text{$\theta $}_2)+31 p \sin (\text{$\theta $}_2))\\
				& -81 (-p+2 \sqrt{1-p}+1) \cos (4 \text{$\theta $}_1)+228 \sqrt{1-p} \sin (2 \text{$\theta $}_2)+51 p+186 \sqrt{1-p}+813].
			\end{split}
		\end{equation}
		\begin{equation}\label{F_H2}
			\small
			\begin{split}
				F_{H_{2}} = & \frac{1}{15552}[-24 p \sin ^4(\text{$\theta $}_1) \sin (4 \text{$\theta $}_2)+24 (9 p-2) \sin ^2(\text{$\theta $}_1) \sin (2 \text{$\theta $}_2)-27 \sqrt{1-p} \sin (4 \text{$\theta $}_1-\text{$\theta $}_2)\\
				& -48 \sqrt{1-p} \sin (2 (\text{$\theta $}_1+\text{$\theta $}_2))-27 \sqrt{1-p} \sin (4 \text{$\theta $}_1+\text{$\theta $}_2)-30 \sqrt{1-p} \sin (2 \text{$\theta $}_1+3 \text{$\theta $}_2)\\
				& -6 \sqrt{1-p} \sin (4 \text{$\theta $}_1-2 \text{$\theta $}_2)-30 \sqrt{1-p} \sin (2 \text{$\theta $}_1-3 \text{$\theta $}_2)+54 \sqrt{1-p} \sin (2 \text{$\theta $}_1-\text{$\theta $}_2)\\
				& +54 \sqrt{1-p} \sin (2 \text{$\theta $}_1+\text{$\theta $}_2)+6 \sqrt{1-p} \sin (4 \text{$\theta $}_1+2 \text{$\theta $}_2)+15 \sqrt{1-p} \sin (4 \text{$\theta $}_1+3 \text{$\theta $}_2)\\
				& +48 \sqrt{1-p} \sin (2 \text{$\theta $}_1-2 \text{$\theta $}_2)+15 \sqrt{1-p} \sin (4 \text{$\theta $}_1-3 \text{$\theta $}_2)-3 \sqrt{1-p} \cos (4 (\text{$\theta $}_1+\text{$\theta $}_2))\\
				& -96 \sqrt{1-p} \cos (2 \text{$\theta $}_1+\text{$\theta $}_2)-8 \sqrt{1-p} \cos (4 \text{$\theta $}_1+\text{$\theta $}_2)-60 \sqrt{1-p} \cos (4 \text{$\theta $}_1+2 \text{$\theta $}_2)\\
				& -104 \sqrt{1-p} \cos (2 \text{$\theta $}_1+3 \text{$\theta $}_2)-60 \sqrt{1-p} \cos (4 \text{$\theta $}_1-2 \text{$\theta $}_2)-52 \sqrt{1-p} \cos (4 \text{$\theta $}_1-3 \text{$\theta $}_2)\\
				& -3 \sqrt{1-p} \cos (4 \text{$\theta $}_1-4 \text{$\theta $}_2)+96 \sqrt{1-p} \cos (2 \text{$\theta $}_1-\text{$\theta $}_2)+8 \sqrt{1-p} \cos (4 \text{$\theta $}_1-\text{$\theta $}_2)\\
				& +52 \sqrt{1-p} \cos (4 \text{$\theta $}_1+3 \text{$\theta $}_2)+12 \sqrt{1-p} \cos (2 \text{$\theta $}_1+4 \text{$\theta $}_2)+104 \sqrt{1-p} \cos (2 \text{$\theta $}_1-3 \text{$\theta $}_2)\\
				& +12 \sqrt{1-p} \cos (2 \text{$\theta $}_1-4 \text{$\theta $}_2)+96 (p-1) \sin ^3(\text{$\theta $}_1) \cos (\text{$\theta $}_1) \cos (3 \text{$\theta $}_2)\\
				& -560 (p-1) \sin ^3(\text{$\theta $}_1) \cos (\text{$\theta $}_1) \sin (3 \text{$\theta $}_2)+48 (p+3) \sin ^2(\text{$\theta $}_1) \cos (2 \text{$\theta $}_2)\\
				& -24 \cos (2 \text{$\theta $}_1) \{\sin ^2(\text{$\theta $}_1) ((18-11 p) \sin (2 \text{$\theta $}_2)+6 (p-1) \cos (2 \text{$\theta $}_2))-12 p+\sqrt{1-p}\}\\
				& +24 \sin (2 \text{$\theta $}_1) \cos (\text{$\theta $}_2) ((7 p-5) \cos (2 \text{$\theta $}_1)+p+1)+4 \sin (2 \text{$\theta $}_1) \sin (\text{$\theta $}_2) ((23 p+49) \cos (2 \text{$\theta $}_1)\\
				& +33 p+39)-6 (2 p+11 \sqrt{1-p}-2) \cos (4 \text{$\theta $}_1)+84 \sqrt{1-p} \sin (2 \text{$\theta $}_2)-18 \sqrt{1-p} \cos (4 \text{$\theta $}_2)\\
				& +120 \sqrt{1-p} \cos (2 \text{$\theta $}_2)+108 (p+7)+90 \sqrt{1-p}].
			\end{split}
		\end{equation}
		\begin{equation}\label{F_H3}
			\small
			\begin{split}
				F_{H_{3}} = & \frac{5}{7776}[-24 \sin ^4(\text{$\theta $}_1) \sin (4 \text{$\theta $}_2)-27 \sqrt{1-p} \sin (4 \text{$\theta $}_1-\text{$\theta $}_2)-24 \sqrt{1-p} \sin (2 (\text{$\theta $}_1+\text{$\theta $}_2))\\
				& -3 \sqrt{1-p} \sin (4 (\text{$\theta $}_1+\text{$\theta $}_2))-27 \sqrt{1-p} \sin (4 \text{$\theta $}_1+\text{$\theta $}_2)-18 \sqrt{1-p} \sin (4 \text{$\theta $}_1+2 \text{$\theta $}_2)\\
				& -30 \sqrt{1-p} \sin (2 \text{$\theta $}_1+3 \text{$\theta $}_2)-30 \sqrt{1-p} \sin (2 \text{$\theta $}_1-3 \text{$\theta $}_2)-12 \sqrt{1-p} \sin (2 \text{$\theta $}_1-4 \text{$\theta $}_2)\\
				& +54 \sqrt{1-p} \sin (2 \text{$\theta $}_1-\text{$\theta $}_2)+54 \sqrt{1-p} \sin (2 \text{$\theta $}_1+\text{$\theta $}_2)+15 \sqrt{1-p} \sin (4 \text{$\theta $}_1+3 \text{$\theta $}_2)\\
				& +12 \sqrt{1-p} \sin (2 \text{$\theta $}_1+4 \text{$\theta $}_2)+24 \sqrt{1-p} \sin (2 \text{$\theta $}_1-2 \text{$\theta $}_2)+18 \sqrt{1-p} \sin (4 \text{$\theta $}_1-2 \text{$\theta $}_2)\\
				& +15 \sqrt{1-p} \sin (4 \text{$\theta $}_1-3 \text{$\theta $}_2)+3 \sqrt{1-p} \sin (4 \text{$\theta $}_1-4 \text{$\theta $}_2)-3 \sqrt{1-p} \cos (4 (\text{$\theta $}_1+\text{$\theta $}_2))\\
				& -102 \sqrt{1-p} \cos (2 \text{$\theta $}_1+\text{$\theta $}_2)-5 \sqrt{1-p} \cos (4 \text{$\theta $}_1+\text{$\theta $}_2)-36 \sqrt{1-p} \cos (4 \text{$\theta $}_1+2 \text{$\theta $}_2)\\
				& -38 \sqrt{1-p} \cos (2 \text{$\theta $}_1+3 \text{$\theta $}_2)-36 \sqrt{1-p} \cos (4 \text{$\theta $}_1-2 \text{$\theta $}_2)-19 \sqrt{1-p} \cos (4 \text{$\theta $}_1-3 \text{$\theta $}_2)\\
				& -3 \sqrt{1-p} \cos (4 \text{$\theta $}_1-4 \text{$\theta $}_2)+102 \sqrt{1-p} \cos (2 \text{$\theta $}_1-\text{$\theta $}_2)+5 \sqrt{1-p} \cos (4 \text{$\theta $}_1-\text{$\theta $}_2)\\
				& +19 \sqrt{1-p} \cos (4 \text{$\theta $}_1+3 \text{$\theta $}_2)+12 \sqrt{1-p} \cos (2 \text{$\theta $}_1+4 \text{$\theta $}_2)+38 \sqrt{1-p} \cos (2 \text{$\theta $}_1-3 \text{$\theta $}_2)\\
				& +12 \sqrt{1-p} \cos (2 \text{$\theta $}_1-4 \text{$\theta $}_2)+72 (p-1) \sin ^4(\text{$\theta $}_1) \cos (4 \text{$\theta $}_2)+48 (p-2) \sin ^3(\text{$\theta $}_1) \cos (\text{$\theta $}_1) \cos (3 \text{$\theta $}_2)\\
				& +16 (5-2 p) \sin ^3(\text{$\theta $}_1) \cos (\text{$\theta $}_1) \sin (3 \text{$\theta $}_2)+48 \sin ^2(\text{$\theta $}_1) \cos (2 \text{$\theta $}_2) (-(p-1) \cos (2 \text{$\theta $}_1)\\
				& +3 p+1)-24 (2 p-1) \sin ^2(\text{$\theta $}_1) (\cos (2 \text{$\theta $}_1)+3) \sin (2 \text{$\theta $}_2)\\
				& +12 \sin (2 \text{$\theta $}_1) \cos (\text{$\theta $}_2) ((17 p-10) \cos (2 \text{$\theta $}_1)-p+2)+4 \sin (2 \text{$\theta $}_1) \sin (\text{$\theta $}_2) ((43-10 p) \cos (2 \text{$\theta $}_1)-6 p+45)\\
				& +36 (9 p-1) \cos (2 \text{$\theta $}_1)-24 \sqrt{1-p} \cos (2 \text{$\theta $}_1)+3 (p-1) \cos (4 \text{$\theta $}_1)-114 \sqrt{1-p} \cos (4 \text{$\theta $}_1)\\
				& -18 \sqrt{1-p} \sin (4 \text{$\theta $}_2)+84 \sqrt{1-p} \sin (2 \text{$\theta $}_2)-18 \sqrt{1-p} \cos (4 \text{$\theta $}_2)+72 \sqrt{1-p} \cos (2 \text{$\theta $}_2)\\
				& +57 p+138 \sqrt{1-p}+807].
			\end{split}
		\end{equation}
		\begin{equation}\label{F_H4}
			\small
			\begin{split}
				F_{H_{4}} = & \frac{5}{7776}[-24 p \sin ^4(\text{$\theta $}_1) \sin (4 \text{$\theta $}_2)+24 (9 p-2) \sin ^2(\text{$\theta $}_1) \sin (2 \text{$\theta $}_2)-27 \sqrt{1-p} \sin (4 \text{$\theta $}_1-\text{$\theta $}_2)\\
				& -48 \sqrt{1-p} \sin (2 (\text{$\theta $}_1+\text{$\theta $}_2))-27 \sqrt{1-p} \sin (4 \text{$\theta $}_1+\text{$\theta $}_2)-30 \sqrt{1-p} \sin (2 \text{$\theta $}_1+3 \text{$\theta $}_2)\\
				& -24 \sqrt{1-p} \sin (4 \text{$\theta $}_1-2 \text{$\theta $}_2)-30 \sqrt{1-p} \sin (2 \text{$\theta $}_1-3 \text{$\theta $}_2)+54 \sqrt{1-p} \sin (2 \text{$\theta $}_1-\text{$\theta $}_2)\\
				& +54 \sqrt{1-p} \sin (2 \text{$\theta $}_1+\text{$\theta $}_2)+24 \sqrt{1-p} \sin (4 \text{$\theta $}_1+2 \text{$\theta $}_2)+15 \sqrt{1-p} \sin (4 \text{$\theta $}_1+3 \text{$\theta $}_2)\\
				& +48 \sqrt{1-p} \sin (2 \text{$\theta $}_1-2 \text{$\theta $}_2)+15 \sqrt{1-p} \sin (4 \text{$\theta $}_1-3 \text{$\theta $}_2)-3 \sqrt{1-p} \cos (4 (\text{$\theta $}_1+\text{$\theta $}_2))\\
				& -108 \sqrt{1-p} \cos (2 \text{$\theta $}_1+\text{$\theta $}_2)-2 \sqrt{1-p} \cos (4 \text{$\theta $}_1+\text{$\theta $}_2)-42 \sqrt{1-p} \cos (4 \text{$\theta $}_1\\
				& +2 \text{$\theta $}_2)-68 \sqrt{1-p} \cos (2 \text{$\theta $}_1+3 \text{$\theta $}_2)-42 \sqrt{1-p} \cos (4 \text{$\theta $}_1-2 \text{$\theta $}_2)-34 \sqrt{1-p} \cos (4 \text{$\theta $}_1\\
				& -3 \text{$\theta $}_2)-3 \sqrt{1-p} \cos (4 \text{$\theta $}_1-4 \text{$\theta $}_2)+108 \sqrt{1-p} \cos (2 \text{$\theta $}_1-\text{$\theta $}_2)+2 \sqrt{1-p} \cos (4 \text{$\theta $}_1-\text{$\theta $}_2)\\
				& +34 \sqrt{1-p} \cos (4 \text{$\theta $}_1+3 \text{$\theta $}_2)+12 \sqrt{1-p} \cos (2 \text{$\theta $}_1+4 \text{$\theta $}_2)+68 \sqrt{1-p} \cos (2 \text{$\theta $}_1-3 \text{$\theta $}_2)\\
				& +12 \sqrt{1-p} \cos (2 \text{$\theta $}_1-4 \text{$\theta $}_2)+96 (p-1) \sin ^3(\text{$\theta $}_1) \cos (\text{$\theta $}_1) \cos (3 \text{$\theta $}_2)\\
				& -272 (p-1) \sin ^3(\text{$\theta $}_1) \cos (\text{$\theta $}_1) \sin (3 \text{$\theta $}_2)+48 (p+3) \sin ^2(\text{$\theta $}_1) \cos (2 \text{$\theta $}_2)\\
				& -24 \cos (2 \text{$\theta $}_1) \{\sin ^2(\text{$\theta $}_1) ((18-11 p) \sin (2 \text{$\theta $}_2)+6 (p-1) \cos (2 \text{$\theta $}_2))-12 p+\sqrt{1-p}\}\\
				& +24 \sin (2 \text{$\theta $}_1) \cos (\text{$\theta $}_2) ((7 p-5) \cos (2 \text{$\theta $}_1)+p+1)+4 \sin (2 \text{$\theta $}_1) \sin (\text{$\theta $}_2) ((17 p+55) \cos (2 \text{$\theta $}_1)\\
				& +39 p+33)-6 (2 p+17 \sqrt{1-p}-2) \cos (4 \text{$\theta $}_1)+48 \sqrt{1-p} \sin (2 \text{$\theta $}_2)-18 \sqrt{1-p} \cos (4 \text{$\theta $}_2)\\
				& +84 \sqrt{1-p} \cos (2 \text{$\theta $}_2)+108 (p+7)+126 \sqrt{1-p}].
			\end{split}
		\end{equation}
		\begin{equation}\label{F_H5}
			\small
			\begin{split}
				F_{H_{5}} = & \frac{5}{7776}[24 (5 p-4) \sin ^4(\text{$\theta $}_1) \sin (4 \text{$\theta $}_2)-6 \sqrt{1-p} \sin (4 \text{$\theta $}_1-\text{$\theta $}_2)-24 \sqrt{1-p} \sin (2 (\text{$\theta $}_1+\text{$\theta $}_2))\\
				& -3 \sqrt{1-p} \sin (4 (\text{$\theta $}_1+\text{$\theta $}_2))-6 \sqrt{1-p} \sin (4 \text{$\theta $}_1+\text{$\theta $}_2)-6 \sqrt{1-p} \sin (4 \text{$\theta $}_1+3 \text{$\theta $}_2)\\
				& -6 \sqrt{1-p} \sin (4 \text{$\theta $}_1-3 \text{$\theta $}_2)-12 \sqrt{1-p} \sin (2 \text{$\theta $}_1-4 \text{$\theta $}_2)+12 \sqrt{1-p} \sin (2 \text{$\theta $}_1-\text{$\theta $}_2)\\
				& +12 \sqrt{1-p} \sin (2 \text{$\theta $}_1+\text{$\theta $}_2)+12 \sqrt{1-p} \sin (2 \text{$\theta $}_1+3 \text{$\theta $}_2)+12 \sqrt{1-p} \sin (2 \text{$\theta $}_1+4 \text{$\theta $}_2)\\
				& +24 \sqrt{1-p} \sin (2 \text{$\theta $}_1-2 \text{$\theta $}_2)+12 \sqrt{1-p} \sin (2 \text{$\theta $}_1-3 \text{$\theta $}_2)+3 \sqrt{1-p} \sin (4 \text{$\theta $}_1-4 \text{$\theta $}_2)\\
				& -\sqrt{1-p} \cos (4 \text{$\theta $}_1-\text{$\theta $}_2)-3 \sqrt{1-p} \cos (4 (\text{$\theta $}_1+\text{$\theta $}_2))-114 \sqrt{1-p} \cos (2 \text{$\theta $}_1+\text{$\theta $}_2)\\
				& -36 \sqrt{1-p} \cos (4 \text{$\theta $}_1+2 \text{$\theta $}_2)-2 \sqrt{1-p} \cos (2 \text{$\theta $}_1+3 \text{$\theta $}_2)-36 \sqrt{1-p} \cos (4 \text{$\theta $}_1-2 \text{$\theta $}_2)\\
				& -\sqrt{1-p} \cos (4 \text{$\theta $}_1-3 \text{$\theta $}_2)-3 \sqrt{1-p} \cos (4 \text{$\theta $}_1-4 \text{$\theta $}_2)+114 \sqrt{1-p} \cos (2 \text{$\theta $}_1-\text{$\theta $}_2)\\
				&+\sqrt{1-p} \cos (4 \text{$\theta $}_1+\text{$\theta $}_2) +\sqrt{1-p} \cos (4 \text{$\theta $}_1+3 \text{$\theta $}_2)+12 \sqrt{1-p} \cos (2 \text{$\theta $}_1+4 \text{$\theta $}_2)\\
				& +2 \sqrt{1-p} \cos (2 \text{$\theta $}_1-3 \text{$\theta $}_2)+12 \sqrt{1-p} \cos (2 \text{$\theta $}_1-4 \text{$\theta $}_2)+72 (p-1) \sin ^4(\text{$\theta $}_1) \cos (4 \text{$\theta $}_2)\\
				& +48 (3-4 p) \sin ^3(\text{$\theta $}_1) \cos (\text{$\theta $}_1) \cos (3 \text{$\theta $}_2)+16 (25 p-22) \sin ^3(\text{$\theta $}_1) \cos (\text{$\theta $}_1) \sin (3 \text{$\theta $}_2)\\
				& +48 \sin ^2(\text{$\theta $}_1) \cos (2 \text{$\theta $}_2) (-(p-1) \cos (2 \text{$\theta $}_1)+3 p+1)+24 \sin ^2(\text{$\theta $}_1) \sin (2 \text{$\theta $}_2) ((p-2) \cos (2 \text{$\theta $}_1)-5 p+6)\\
				& +12 \sin (2 \text{$\theta $}_1) \cos (\text{$\theta $}_2) ((12 p-5) \cos (2 \text{$\theta $}_1)+4 p-3)-4 \sin (2 \text{$\theta $}_1) \sin (\text{$\theta $}_2) ((p-22) \cos (2 \text{$\theta $}_1)+39 p-66)\\
				& +36 (9 p-1) \cos (2 \text{$\theta $}_1)-24 \sqrt{1-p} \cos (2 \text{$\theta $}_1)+3 (p-1) \cos (4 \text{$\theta $}_1)-114 \sqrt{1-p} \cos (4 \text{$\theta $}_1)\\
				& -18 \sqrt{1-p} \sin (4 \text{$\theta $}_2)+48 \sqrt{1-p} \sin (2 \text{$\theta $}_2)-18 \sqrt{1-p} \cos (4 \text{$\theta $}_2)+72 \sqrt{1-p} \cos (2 \text{$\theta $}_2)\\
				& +57 p+138 \sqrt{1-p}+807].
			\end{split}
		\end{equation}
       
       \begin{figure}
       	\includegraphics[scale = .75]{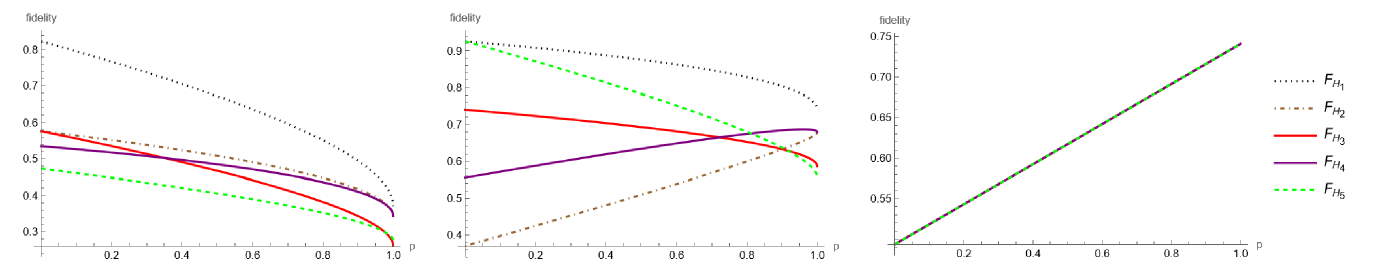}
       	\caption{For all five hypergraph states we plot the teleportation fidelity with respect to the channel parameter $p$, in case of markovian amplitude damping channel. Three subfigures consider different states $\ket{\phi}_a$ to teleport. We consider $\ket{\phi}_a = \ket{+}$, $\ket{\phi}_a = \frac{1}{\sqrt{2}}(\ket{0} + \ket{2})$, and $\ket{\phi}_a = \ket{0}$ in the left, middle, and right subfigure, respectively.}
       	\label{fig:Markovian_amplitude_damping_channel}
       \end{figure}
       
       \begin{figure}
       	\centering
       	\begin{subfigure}{0.32\textwidth}
       		\centering
       		\includegraphics[scale = .37]{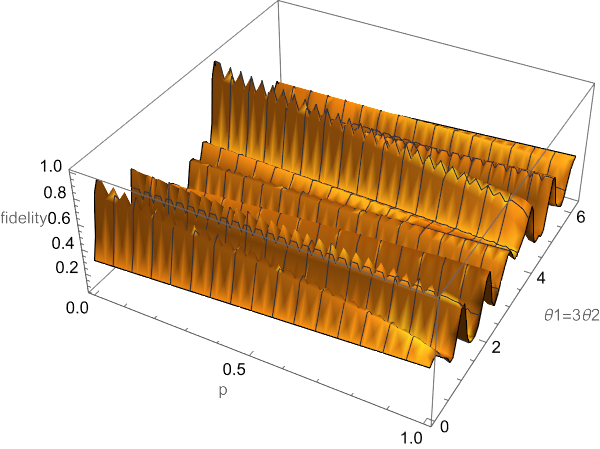}
       		\caption{$F_{H_1}$}
       	\end{subfigure}
       	\begin{subfigure}{0.32\textwidth}
       		\centering
       		\includegraphics[scale = .37]{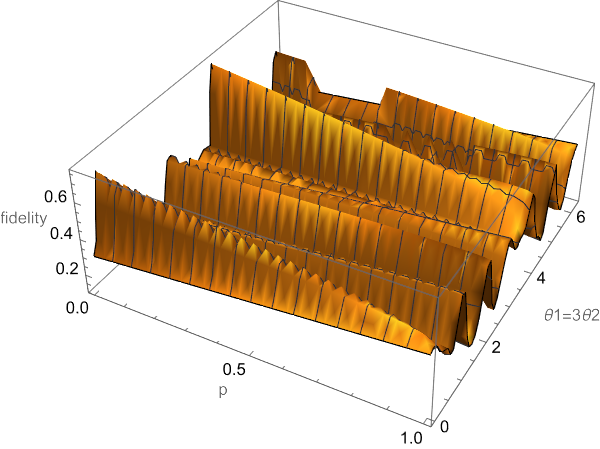}
       		\caption{$F_{H_2}$}
       	\end{subfigure}
       	\begin{subfigure}{0.32\textwidth}
       		\centering
       		\includegraphics[scale = .37]{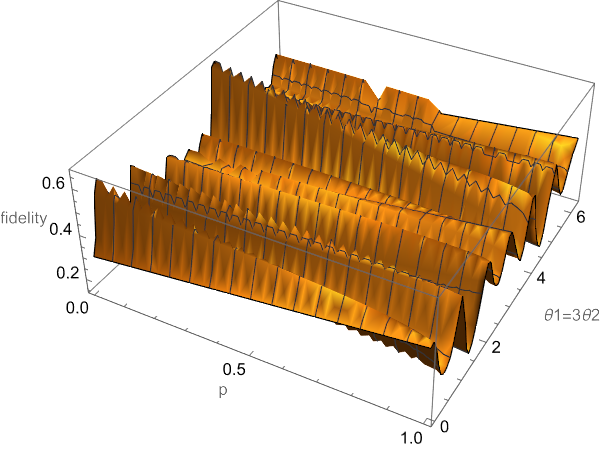}
       		\caption{$F_{H_1}$}
       	\end{subfigure}
       	\\
       	\begin{subfigure}{0.32\textwidth}
       		\centering
       		\includegraphics[scale = .37]{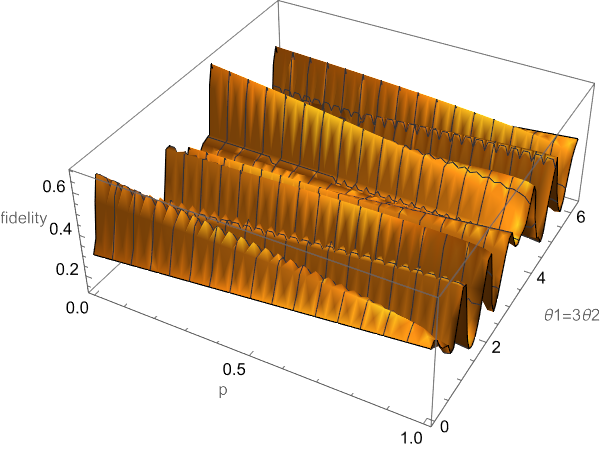}
       		\caption{$F_{H_4}$}
       	\end{subfigure}
       	\begin{subfigure}{0.32\textwidth}
       		\centering
       		\includegraphics[scale = .35]{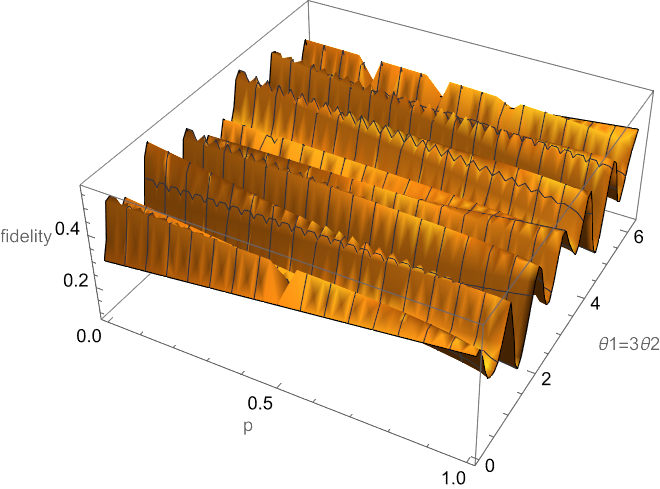}
       		\caption{$F_{H_5}$}
       	\end{subfigure}
       	\label{fig:Markovian Amplitude Damping Channel}
       	\caption{The quantum teleportation fidelity under Markovian amplitude damping noise is plotted with respect to the channel parameter $p$ and the state parameter $\theta_2$. Different subfigure is generated for different hypergraph states.}
       \end{figure}
       
       We also consider the non-Markovian amplitude damping channel for our discussion. We generalize it for qubit to qutrit. The Kraus operations mentioned in equation (\ref{Amplitude damping channel}) will be updated as:
       \begin{equation}
       	 K_{0,0} = \ket{0} \bra{0} + \sqrt{1 - \lambda(t)} \sum_{i=1}^{2} \ket{i} \bra{i},\\ K_{0,i} = \sqrt{\lambda(t)} \ket{0} \bra{i} \text{for}~  1 \le i \le 2.
       \end{equation}
       where $\lambda(t) = 1 - e^{-gt} ( \frac{g}{l} \sinh[\frac{lt}{2}] + \cosh[\frac{lt}{2}])^2$, and $l = \sqrt{g^2 - 2 \gamma g}$. Recall that $g$ is the spectral width of the system-environment coupling and $\gamma$ is the spontaneous emission rate.
       
       We work out the expression of teleportation fidelity for the non-Markovian amplitude damping channel. These expression are similar to equations (\ref{F_H1}), (\ref{F_H2}), (\ref{F_H3}), (\ref{F_H4}), and (\ref{F_H5}) where $p$ is replaced by $\lambda(t)$. We consider three qutrit states for our calculations, which are $\ket{+}$, $\frac{1}{\sqrt{2}} (\ket{0} + \ket{2})$ and $\ket{0}$. Fixing $g = 1$, and $\gamma = 10$ we plot the teleportation fidelity with respect to the parameter $t$, which is available in Figure \ref{fig:Non-Markovian Amplitude Damping channel}. The signature of non-Markovianity is observed.
       
       If we consider $\theta_1 = 3 \theta_2$, equation (\ref{quantum state vector}) then the number of state parameters is reduced to one $\theta_2$. Then we can plot fidelity as a function of the state parameter $\theta_2$ and the channel parameter $t$. The corresponding surface plots of fidelity for different hypergraph states are in Figure \ref{fig:Non_Markovian_Amplitude_Damping_Channel}. 
       
       \begin{figure}
       	\includegraphics[scale = .9]{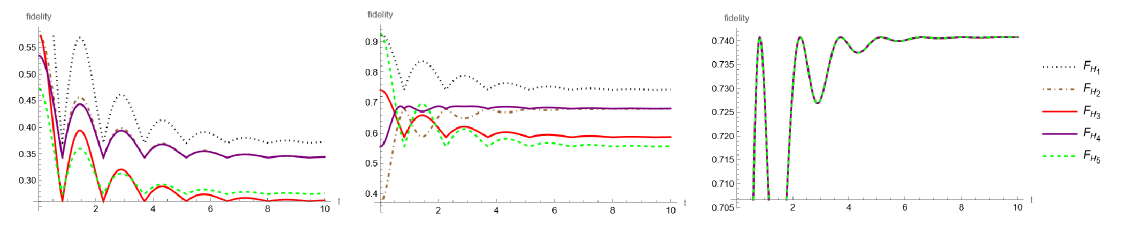}
       	\caption{For all five hypergraph states we plot the teleportation fidelity with respect to the channel parameter $p$, in case of non-Markovian amplitude damping channel. Three subfigues consider different states $\ket{\phi}_a$ to teleport. We consider $\ket{\phi}_a = \ket{+}$, $\ket{\phi}_a = \frac{1}{\sqrt{2}}(\ket{0} + \ket{2})$, and $\ket{\phi}_a = \ket{0}$ in the left, middle, and right subfigure, respectively.}
       	\label{fig:Non-Markovian Amplitude Damping channel}
       \end{figure}
       
       \begin{figure}[hbt!]
       	\centering
       	\begin{subfigure}{0.32\textwidth}
       		\centering
       		\includegraphics[scale = .34]{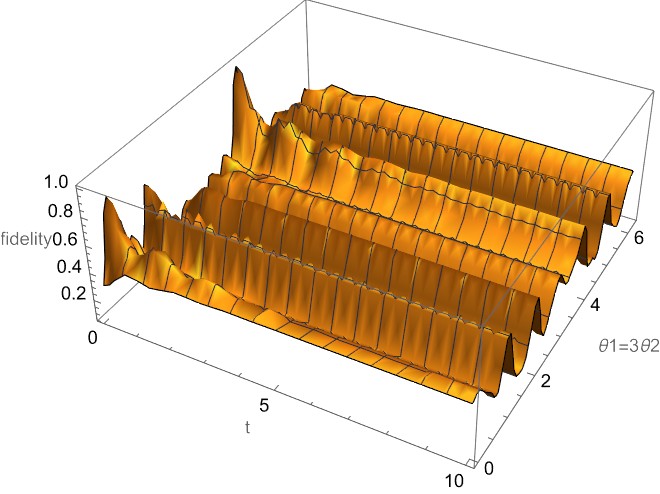}
       		\caption{$F_{H_1}$}
       	\end{subfigure}
       	\begin{subfigure}{0.32\textwidth}
       		\centering
       		\includegraphics[scale = .37]{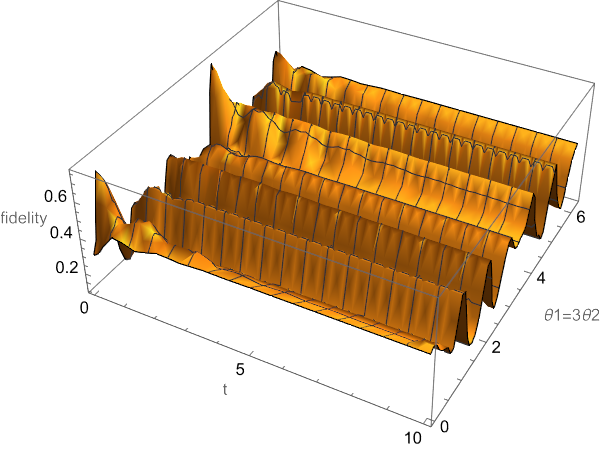}
       		\caption{$F_{H_2}$}
       	\end{subfigure}
       	\begin{subfigure}{0.32\textwidth}
       		\centering
       		\includegraphics[scale = .37]{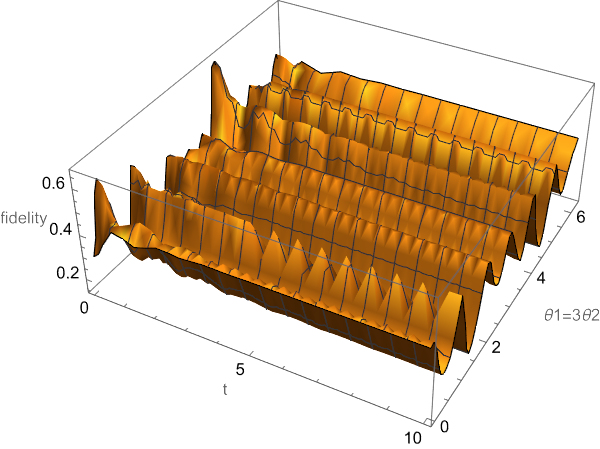}
       		\caption{$F_{H_3}$}
       	\end{subfigure}
       	\\
       	\begin{subfigure}{0.32\textwidth}
       		\centering
       		\includegraphics[scale = .37]{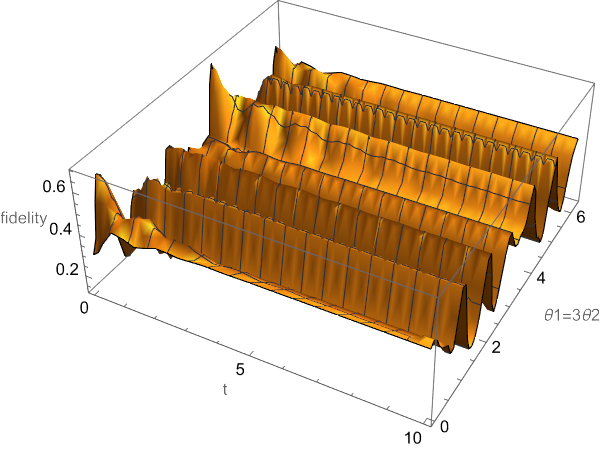}
       		\caption{$F_{H_4}$}
       	\end{subfigure}
       	\begin{subfigure}{0.32\textwidth}
       		\centering
       		\includegraphics[scale = .37]{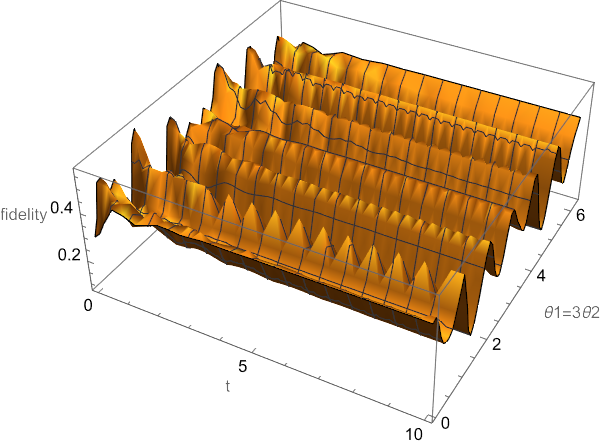}
       		\caption{$F_{H_5}$}
       	\end{subfigure}
       	\caption{The quantum teleportation fidelity under non-Markovian amplitude damping noise is plotted with respect to the channel parameter $t$ and the state parameter $\theta_2$. Different subfigure is generated for different hypergraphs. Also, here we fix $g=1$, and $\gamma=10$.}
       	\label{fig:Non_Markovian_Amplitude_Damping_Channel}
       \end{figure}

	 \subsection{Markovian and non-Markovian Dephasing channel:}
	    
	    The Markovian dephasing channel for a qubit system is represented by the Kraus operators
             \begin{center}
				\begin{tabular}{c c c c}
					$K_1 = \sqrt{(1 - p)} I_2,$ & $K_2 = \sqrt{\frac {p} {3}} \sigma_x,$ & $K_3 = \sqrt{\frac {p} {3}} \sigma_y,$ & $K_4 = \sqrt{\frac {p} {3}} \sigma_z.$
				\end{tabular}
			\end{center}
			For a qutrit system we generalize the Markovian dephasing channel using the following Kraus operators
			\begin{equation} \label{Markovian dephasing}
				{K_{r,s}} = \begin{cases}
					\sqrt{\frac{1 - p}{3}} I_3 & \text{when}~ r = 0, s=0, \\ 
					\sqrt{\frac {p} {24}} W_{r,s} & \text{for}~ 0\leq r,s \leq 2 ~\text{and}~ (r,s) \neq (0,0),
				\end{cases} 
			\end{equation}
			 where the Weyl operators $W_{r,s}, r,s = 0, 1, 2$ are mentioned in equation (\ref{Weyl operator}). Using Lemma \ref{multiqutrit_channel_lemma}, we can construct three-qutrit channel which can be described by the following Kraus operators:
	         
	         \small
			 \noindent\begin{tabular}{p{5.5cm} p{5.5cm} p{5.5cm}}
			 	\centering
					$K_1=\sqrt{\frac{1-p}{3}} (K_{0,0}\otimes I_3 \otimes I_3)$,& $K_2=\sqrt{\frac{p}{24}} (K_{0,1} \otimes I_3 \otimes I_3)$,& $K_3=\sqrt{\frac{p}{24}} (K_{0,2} \otimes I_3 \otimes I_3)$,\\
					$K_4=\sqrt{\frac{p}{24}} (K_{1,0} \otimes I_3 \otimes I_3)$,& $K_5=\sqrt{\frac{p}{24}} (K_{1,1} \otimes I_3 \otimes I_3)$,& $K_6=\sqrt{\frac{p}{24}} (K_{1,2} \otimes I_3 \otimes I_3)$,\\
					$K_7=\sqrt{\frac{p}{24}} (K_{2,0} \otimes I_3 \otimes I_3)$,& $K_8=\sqrt{\frac{p}{24}} (K_{2,1} \otimes I_3 \otimes I_3)$,& $K_9=\sqrt{\frac{p}{24}} (K_{2,2} \otimes I_3 \otimes I_3)$,\\
					$K_{10}=\sqrt{\frac{1-p}{3}} (I_3\otimes K_{0,0} \otimes I_3)$,& $K_{11}=\sqrt{\frac{p}{24}} (I_3 \otimes K_{0,1} \otimes I_3)$,& $K_{12}=\sqrt{\frac{p}{24}} (I_3 \otimes K_{0,2} \otimes I_3)$,\\
					$K_{13}=\sqrt{\frac{p}{24}} (I_3 \otimes K_{1,0} \otimes I_3)$,& $K_{14}=\sqrt{\frac{p}{24}} (I_3 \otimes K_{1,1} \otimes I_3)$,& $K_{15}=\sqrt{\frac{p}{24}} (I_3 \otimes K_{1,2} \otimes I_3)$,\\
					$K_{16}=\sqrt{\frac{p}{24}} (I_3 \otimes K_{2,0} \otimes I_3)$,& $K_{17}=\sqrt{\frac{p}{24}} (I_3 \otimes K_{2,1} \otimes I_3)$,& $K_{18}=\sqrt{\frac{p}{24}} (I_3 \otimes K_{2,2} \otimes I_3)$,\\
					$K_{19}=\sqrt{\frac{1-p}{3}} (I_3\otimes I_3 \otimes K_{0,0})$,& $K_{20}=\sqrt{\frac{p}{24}} (I_3\otimes I_3 \otimes K_{0,1})$,& $K_{21}=\sqrt{\frac{p}{24}} (I_3\otimes I_3 \otimes K_{0,2})$,\\
					$K_{22}=\sqrt{\frac{p}{24}} (I_3\otimes I_3 \otimes K_{1,0})$,& $K_{23}=\sqrt{\frac{p}{24}} (I_3\otimes I_3 \otimes K_{1,1})$,& $K_{24}=\sqrt{\frac{p}{24}} (I_3\otimes I_3 \otimes K_{1,2})$,\\
					$K_{25}=\sqrt{\frac{p}{24}} (I_3\otimes I_3 \otimes K_{2,0})$,& $K_{26}=\sqrt{\frac{p}{24}} (I_3\otimes I_3 \otimes K_{2,1})$,& $K_{27}=\sqrt{\frac{p}{24}} (I_3\otimes I_3 \otimes K_{2,2})$.
				\end{tabular}
				
			We transfer a qutrit quantum state mentioned in equation (\ref{quantum state vector}) using our quantum teleportation procedure, where the hypergraph states were distributed via a Markovian Dephasing channel. The teleportation fidelity are as follows: 	
			\begin{equation}
				\small
				\begin{aligned}
					F_{H_{1}}&=\frac{5}{20736}[8 | p - 1|  \{8 \sin ^2(\text{$\theta $}_1) (16 \cos ^2(\text{$\theta $}_1) \cos (2 \text{$\theta $}_2)+\sin ^2(\text{$\theta $}_1) (3 \sin (4 \text{$\theta $}_2)+\cos (4 \text{$\theta $}_2))\\
					& +\sin (2 \text{$\theta $}_1) (17 \sin (3 \text{$\theta $}_2)-4 \cos (3 \text{$\theta $}_2))+55 \sin (2 \text{$\theta $}_2))+4 \cos (2 \text{$\theta $}_1) (74 \sin ^2(\text{$\theta $}_1) \sin (2 \text{$\theta $}_2)\\
					&+\sin (2 \text{$\theta $}_1) (7 \sin (\text{$\theta $}_2)-4 \cos (\text{$\theta $}_2))+1)+4 \sin (2 \text{$\theta $}_1) (41 \sin (\text{$\theta $}_2)+20 \cos (\text{$\theta $}_2))-81 \cos (4 \text{$\theta $}_1)\\
					& +333\}+ p \{-8 \sin ^2(\text{$\theta $}_1) (16 \cos ^2(\text{$\theta $}_1) \cos (2 \text{$\theta $}_2)+\sin ^2(\text{$\theta $}_1) (9 \sin (4 \text{$\theta $}_2)+\cos (4 \text{$\theta $}_2))\\
					& +\sin (2 \text{$\theta $}_1) (5 \sin (3 \text{$\theta $}_2)-4 \cos (3 \text{$\theta $}_2))-35 \sin (2 \text{$\theta $}_2))+4 \cos (2 \text{$\theta $}_1) (-86 \sin ^2(\text{$\theta $}_1) \sin (2 \text{$\theta $}_2)\\
					& +\sin (2 \text{$\theta $}_1) (29 \sin (\text{$\theta $}_2)+4 \cos (\text{$\theta $}_2))-1)+4 \sin (2 \text{$\theta $}_1) (139 \sin (\text{$\theta $}_2)+28 \cos (\text{$\theta $}_2))\\
					& +81 \cos (4 \text{$\theta $}_1)+1971\}].
				\end{aligned}
			\end{equation}
			\begin{equation}
				\small
				\begin{aligned}
					F_{H_{2}}&=\frac{5}{10368}[-8 | p-1|  \{8 \sin ^2(\text{$\theta $}_1) (-26 \cos ^2(\text{$\theta $}_1) \cos (2 \text{$\theta $}_2)+\sin ^2(\text{$\theta $}_1) \cos (4 \text{$\theta $}_2)\\
					& +\sin (\text{$\theta $}_1) \cos (\text{$\theta $}_1) (7 \cos (3 \text{$\theta $}_2)-29 \sin (3 \text{$\theta $}_2))-2 \sin (2 \text{$\theta $}_2))+\cos (2 \text{$\theta $}_1) (80 \sin ^2(\text{$\theta $}_1) \sin (2 \text{$\theta $}_2)\\
					& +38 \sin (2 \text{$\theta $}_1) (\cos (\text{$\theta $}_2)-\sin (\text{$\theta $}_2))+4)-2 \sin (2 \text{$\theta $}_1) (29 \sin (\text{$\theta $}_2)+11 \cos (\text{$\theta $}_2))+9 \cos (4 \text{$\theta $}_1)-141\}\\
					& +4 p \sin (\text{$\theta $}_1) \{-52 \sin (\text{$\theta $}_1) \cos ^2(\text{$\theta $}_1) \cos (2 \text{$\theta $}_2)+2 \sin ^3(\text{$\theta $}_1) \cos (4 \text{$\theta $}_2)\\
					& +\cos (\text{$\theta $}_1) ((19 \cos (2 \text{$\theta $}_1)+13) \cos (\text{$\theta $}_2)+2 \sin ^2(\text{$\theta $}_1) (7 \cos (3 \text{$\theta $}_2)-11 \sin (3 \text{$\theta $}_2))+(35 \cos (2 \text{$\theta $}_1)+61) \sin (\text{$\theta $}_2))\\
					& +4 \sin (\text{$\theta $}_1) (11-4 \cos (2 \text{$\theta $}_1)) \sin (2 \text{$\theta $}_2)\}+p (4 \cos (2 \text{$\theta $}_1)+9 \cos (4 \text{$\theta $}_1)+1011)].
				\end{aligned}
			\end{equation}
			\begin{equation}
				\small
				\begin{aligned}
					F_{H_{3}}&=\frac{5}{20736}[8 | p-1|  \{8 \sin ^2(\text{$\theta $}_1) (28 \cos ^2(\text{$\theta $}_1) \cos (2 \text{$\theta $}_2)-\sin ^2(\text{$\theta $}_1) (3 \sin (4 \text{$\theta $}_2)+5 \cos (4 \text{$\theta $}_2))\\
					& +\sin (2 \text{$\theta $}_1) (8 \sin (3 \text{$\theta $}_2)-7 \cos (3 \text{$\theta $}_2))+13 \sin (2 \text{$\theta $}_2))+4 \cos (2 \text{$\theta $}_1) (14 \sin ^2(\text{$\theta $}_1) \sin (2 \text{$\theta $}_2)-5)\\
					& +2 (22 \sin (2 \text{$\theta $}_1)-19 \sin (4 \text{$\theta $}_1)) \cos (\text{$\theta $}_2)+64 \sin (2 \text{$\theta $}_1) (\cos (2 \text{$\theta $}_1)+2) \sin (\text{$\theta $}_2)-39 \cos (4 \text{$\theta $}_1)+315\}\\
					& +p \{8 \sin ^2(\text{$\theta $}_1) (-28 \cos ^2(\text{$\theta $}_1) \cos (2 \text{$\theta $}_2)+\sin ^2(\text{$\theta $}_1) (5 \cos (4 \text{$\theta $}_2)-9 \sin (4 \text{$\theta $}_2))\\
					& +\sin (2 \text{$\theta $}_1) (7 \cos (3 \text{$\theta $}_2)-2 \sin (3 \text{$\theta $}_2))-13 \sin (2 \text{$\theta $}_2))+4 \cos (2 \text{$\theta $}_1) (-14 \sin ^2(\text{$\theta $}_1) \sin (2 \text{$\theta $}_2)\\
					& +\sin (2 \text{$\theta $}_1) (26 \sin (\text{$\theta $}_2)+19 \cos (\text{$\theta $}_2))+5)-44 \sin (2 \text{$\theta $}_1) (\cos (\text{$\theta $}_2)\\
					& -2 \sin (\text{$\theta $}_2))+39 \cos (4 \text{$\theta $}_1)+1989\}].
				\end{aligned}
			\end{equation}
			\begin{equation}
				\small
				\begin{aligned}
					F_{H_{4}}&=\frac{5}{10368}[p \{8 \sin ^2(\text{$\theta $}_1) (-20 \cos ^2(\text{$\theta $}_1) \cos (2 \text{$\theta $}_2)+\sin ^2(\text{$\theta $}_1) \cos (4 \text{$\theta $}_2)+\sin (\text{$\theta $}_1) \cos (\text{$\theta $}_1) (\sin (3 \text{$\theta $}_2)+7 \cos (3 \text{$\theta $}_2))\\
					& +25 \sin (2 \text{$\theta $}_2))+2 \sin (2 \text{$\theta $}_1) (61 \sin (\text{$\theta $}_2)+13 \cos (\text{$\theta $}_2))+4 \cos (2 \text{$\theta $}_1) (19 \sin (\text{$\theta $}_1) \cos (\text{$\theta $}_1) \cos (\text{$\theta $}_2)\\
					& +5 \sin (\text{$\theta $}_1) \sin (\text{$\theta $}_2) (7 \cos (\text{$\theta $}_1)-4 \sin (\text{$\theta $}_1) \cos (\text{$\theta $}_2))+1)+15 \cos (4 \text{$\theta $}_1)+1005\}\\
					& -8 | p-1|  \{-58 \sin (2 \text{$\theta $}_1) \sin (\text{$\theta $}_2)+8 \sin ^2(\text{$\theta $}_1) (-20 \cos ^2(\text{$\theta $}_1) \cos (2 \text{$\theta $}_2)+\sin ^2(\text{$\theta $}_1) \cos (4 \text{$\theta $}_2)\\
					& +\sin (\text{$\theta $}_1) \cos (\text{$\theta $}_1) (7 \cos (3 \text{$\theta $}_2)-17 \sin (3 \text{$\theta $}_2))+\sin (2 \text{$\theta $}_2))-44 \sin (\text{$\theta $}_1) \cos (\text{$\theta $}_1) \cos (\text{$\theta $}_2)\\
					& +19 \sin (4 \text{$\theta $}_1) \cos (\text{$\theta $}_2)+\cos (2 \text{$\theta $}_1) (4 \sin (\text{$\theta $}_1) \sin (\text{$\theta $}_2) (52 \sin (\text{$\theta $}_1) \cos (\text{$\theta $}_2)-19 \cos (\text{$\theta $}_1))+4)\\
					& +15 \cos (4 \text{$\theta $}_1)-147\}].
				\end{aligned}
			\end{equation}
			\begin{equation}
				\small
				\begin{aligned}
					F_{H_{5}}&=\frac{5}{20736}[p \{8 \sin ^2(\text{$\theta $}_1) (-28 \cos ^2(\text{$\theta $}_1) \cos (2 \text{$\theta $}_2)+5 \sin ^2(\text{$\theta $}_1) \cos (4 \text{$\theta $}_2)+\sin (2 \text{$\theta $}_1) (\sin (3 \text{$\theta $}_2)-5 \cos (3 \text{$\theta $}_2))\\
					& +20 \sin (2 \text{$\theta $}_2))+4 \cos (2 \text{$\theta $}_1) (16 \sin ^2(\text{$\theta $}_1) \sin (2 \text{$\theta $}_2)+\sin (2 \text{$\theta $}_1) (7 \cos (\text{$\theta $}_2)-\sin (\text{$\theta $}_2))+5)\\
					& +4 \sin (2 \text{$\theta $}_1) (25 \sin (\text{$\theta $}_2)+\cos (\text{$\theta $}_2))+39 \cos (4 \text{$\theta $}_1)+1989\}-8 | p-1|  \{8 \sin ^2(\text{$\theta $}_1) (-28 \cos ^2(\text{$\theta $}_1) \cos (2 \text{$\theta $}_2)\\
					& +\sin ^2(\text{$\theta $}_1) (6 \sin (4 \text{$\theta $}_2)+5 \cos (4 \text{$\theta $}_2))+\sin (2 \text{$\theta $}_1) (7 \sin (3 \text{$\theta $}_2)-5 \cos (3 \text{$\theta $}_2))-10 \sin (2 \text{$\theta $}_2))\\
					& +4 \cos (2 \text{$\theta $}_1) (4 \sin ^2(\text{$\theta $}_1) \sin (2 \text{$\theta $}_2)+7 \sin (2 \text{$\theta $}_1) (\cos (\text{$\theta $}_2)-\sin (\text{$\theta $}_2))+5)+4 \sin (2 \text{$\theta $}_1) (\cos (\text{$\theta $}_2)\\
					& -41 \sin (\text{$\theta $}_2))+39 \cos (4 \text{$\theta $}_1)-315\}].
				\end{aligned}
			\end{equation}
       		
       		For $\ket{\phi}_a = \ket{+}$, $\frac{1}{\sqrt{2}}(\ket{0} + \ket{2})$, and $\ket{0}$ we plot the teleportation fidelity with respect to the channel parameter $p$ in Figure \ref{fig:MarkovianDephasing}.
	       
	       \begin{figure}
		       	\includegraphics[scale = .75]{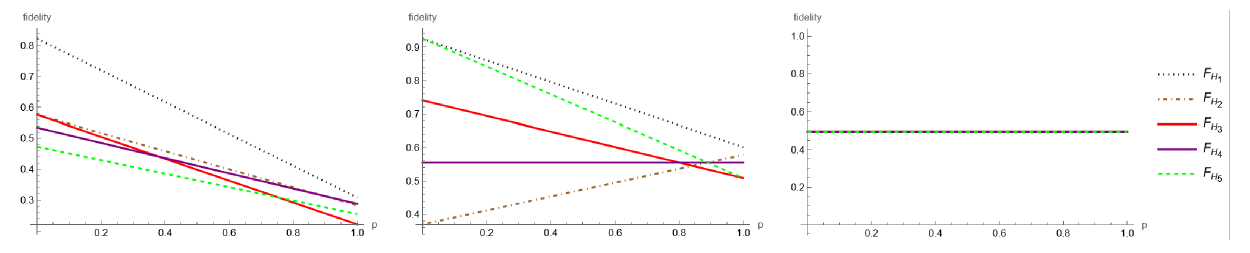}
		       	\caption{For all five hypergraph states we plot the teleportation fidelity with respect to the channel parameter $p$, in case of Markovian dephasing noise. Three subfigures consider different states $\ket{\phi}_a$ to teleport.  We consider $\ket{\phi}_a = \ket{+}$, $\ket{\phi}_a = \frac{1}{\sqrt{2}}(\ket{0} + \ket{2})$, and $\ket{\phi}_a = \ket{0}$ in the left, middle, and right subfigure, respectively.}
		       	\label{fig:MarkovianDephasing}
	       \end{figure}
       
       \begin{figure}
       	\centering
       	\begin{subfigure}{0.32\textwidth}
       		\centering
       		\includegraphics[scale = .37]{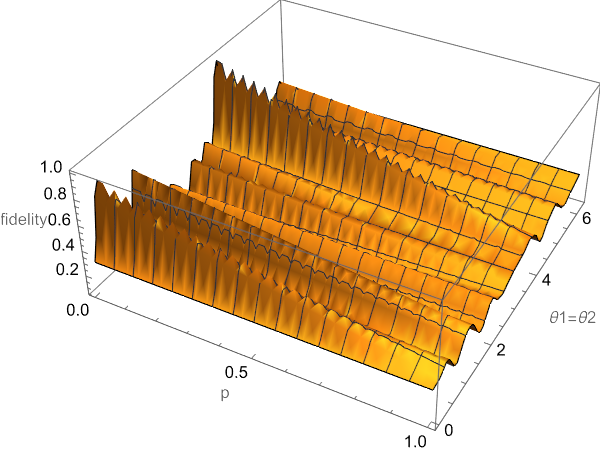}
       		\caption{$F_{H_1}$}
       	\end{subfigure}
       	\begin{subfigure}{0.32\textwidth}
       		\centering
       		\includegraphics[scale = .37]{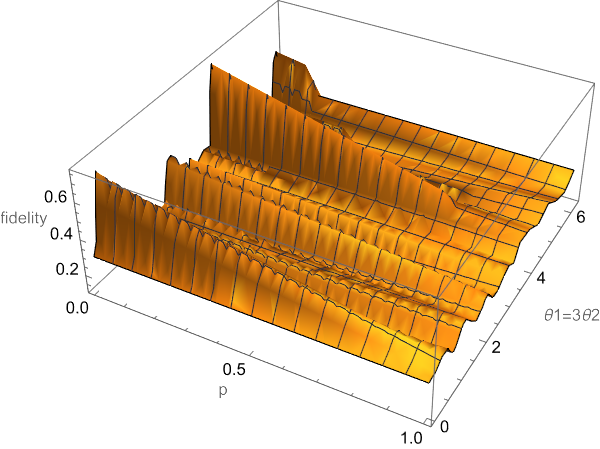}
       		\caption{$F_{H_2}$}
       \end{subfigure}
       	\begin{subfigure}{0.32\textwidth}
       		\centering
       		\includegraphics[scale = .37]{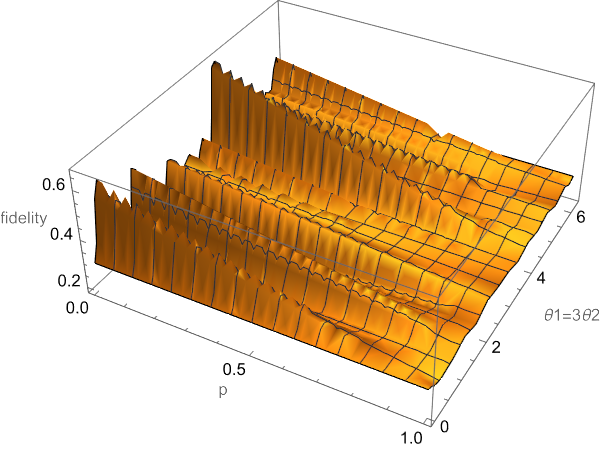}
       		\caption{$F_{H_3}$}
       	\end{subfigure}
       	\\
       	\begin{subfigure}{0.32\textwidth}
       		\centering
       		\includegraphics[scale = .37]{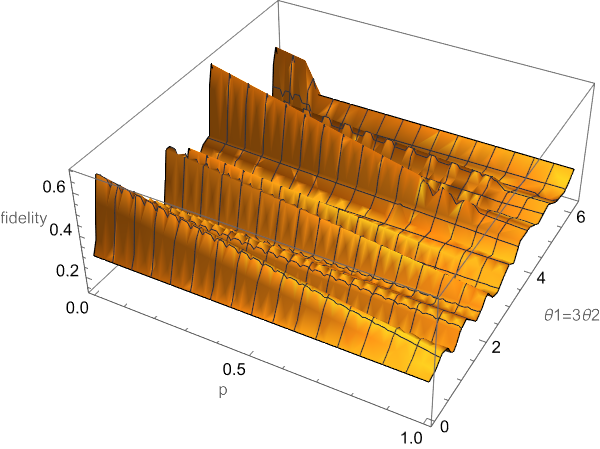}
       		\caption{$F_{H_4}$}
       	\end{subfigure}
       	\begin{subfigure}{0.32\textwidth}
       		\centering
       		\includegraphics[scale = .37]{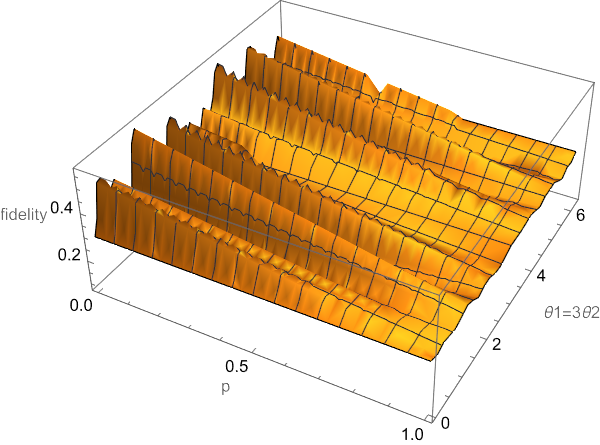}
       		\caption{$F_{H_5}$}
       \end{subfigure}
       	\label{fig:Markovian-Dephasing}
       	\caption{The quantum teleportation fidelity under Markovian dephasing noise is plotted with respect to the channel parameter $p$ and the state parameter $\theta_2$ assuming $\theta_1 = 3 \theta_2$. Different subfigure is generated for different hypergraph states.}
       \end{figure}
       
       The non-Markovian dephasing channel \cite{shrikant2018non} for a qutrit system can be represented by using the following Kraus operators
       \begin{equation}\label{Non Markovian dephasing}
       		{K_{r,s}} =  \begin{cases}
       			\sqrt{\frac{1 - \kappa}{3}} I_3 ~\text{when}~ r = 0, s=0, \\ 
       			\sqrt{\frac {\kappa} {24}} W_{r,s} ~\text{for}~ 0\leq r,s \leq 2 ~\text{and}~ (r,s) \neq (0,0).
       		\end{cases}
       \end{equation}
       Here we consider $\kappa(p) = p \frac{1 + \eta (1 - 2p) \sin(\beta p)}{1 + \eta (1 - 2p)}$. Also, $\eta$ and $\beta$ are two positive constants characterizing the strength and frequency of the channel, respectively. Consider the Kraus operators for the Markovian and non-Markovian dephasing channels as mentioned in equations (\ref{Markovian dephasing}), and (\ref{Non Markovian dephasing}), respectively. Note that, $p$ of equation (\ref{Markovian dephasing}) is replaced by $\kappa(p)$ in equation (\ref{Non Markovian dephasing}). The similar changes will be reflected in the expressions of teleportation Fidelity. But, in case of non-Markovian dephasing these expressions have two additional channel parameters $\eta$, and $\beta$.
       
       Now, we work out the expression of teleportation fidelity for the non-Markovian dephasing channel. We consider three qutrit states for our calculations. They are $\ket{+}$, $\frac{1}{\sqrt{2}} (\ket{0} + \ket{2})$, and $\ket{0}$. Fixing $\eta = 0.5$, and $\beta = 100$ we plot the teleportation fidelity with respect to time, which is available in Figure \ref{fig:non-Markovian Dephasing}. We observe the signature of non-Markovianity.
       
       If we consider $\theta_1 = 3 \theta_2$ then the number of state parameters is reduced to one, which is $\theta_2$. It assists us to plot fidelity as a function of the state parameter $\theta_2$ and the channel parameter $p$. The corresponding surface plots of fidelity for different hypergraph states are depict in Figure \ref{fig:non_Markovian_Dephasing}.
       
       \begin{figure}
       	\includegraphics[scale = .78]{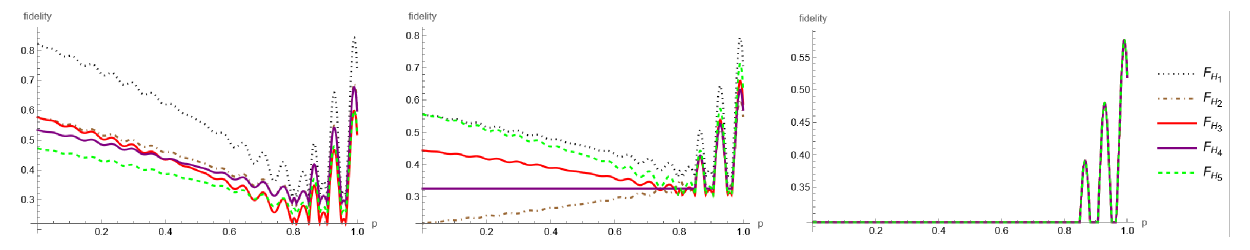}
       	\caption{For all five hypergraph states we plot the teleportation fidelity with respect to the channel parameter $p$, in case of non-Markovian dephasing noise. Three subfigures consider different states $\ket{\phi}_a$ to teleport.  We consider $\ket{\phi}_a = \ket{+}$, $\ket{\phi}_a = \frac{1}{\sqrt{2}}(\ket{0} + \ket{2})$, and $\ket{\phi}_a = \ket{0}$ in the left, middle, and right subfigure, respectively.}
       \label{fig:non-Markovian Dephasing}
       \end{figure}
       
       \begin{figure}[hbt!]
       	\centering
       	\begin{subfigure}{0.32\textwidth}
       		\centering
       		\includegraphics[scale = .37]{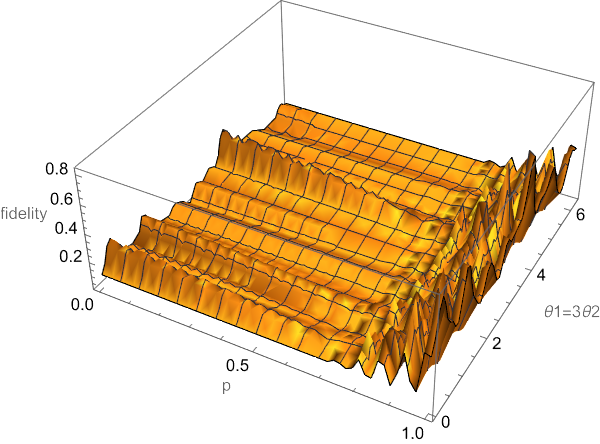}
       		\caption{$F_{H_1}$}
       	\end{subfigure}
       	\begin{subfigure}{0.32\textwidth}
       		\centering
       		\includegraphics[scale = .37]{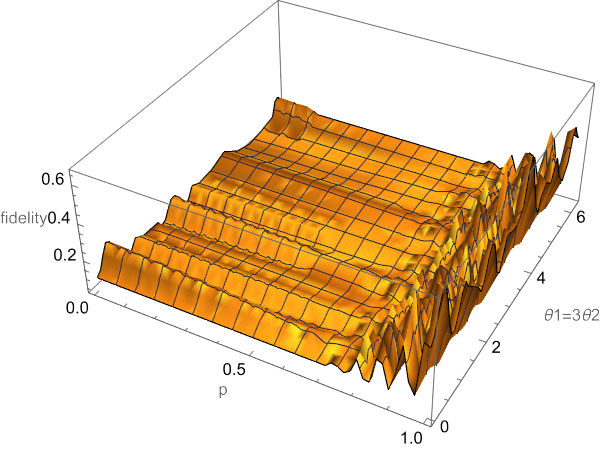}
       		\caption{$F_{H_2}$}
       	\end{subfigure}
       	\begin{subfigure}{0.32\textwidth}
       		\centering
       		\includegraphics[scale = .37]{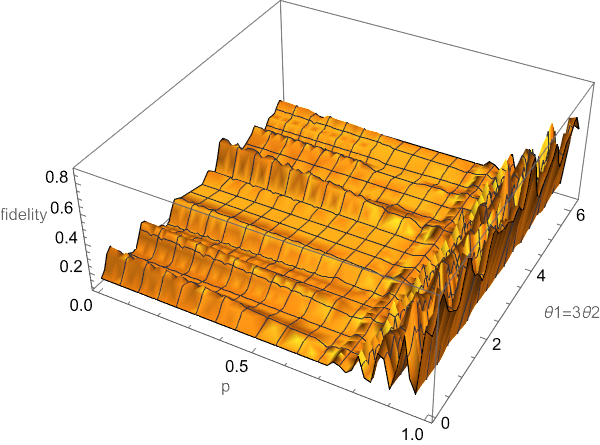}
       		\caption{$F_{H_3}$}
       	\end{subfigure}
       	\\
       	\begin{subfigure}{0.32\textwidth}
       		\centering
       		\includegraphics[scale = .37]{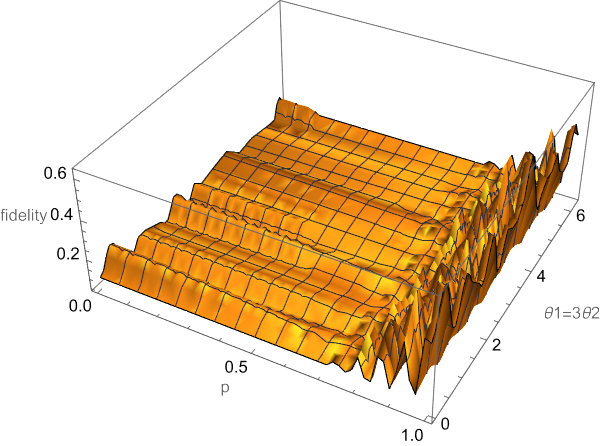}
       		\caption{$F_{H_4}$}
       	\end{subfigure}
       	\begin{subfigure}{0.32\textwidth}
       		\centering
       		\includegraphics[scale = .37]{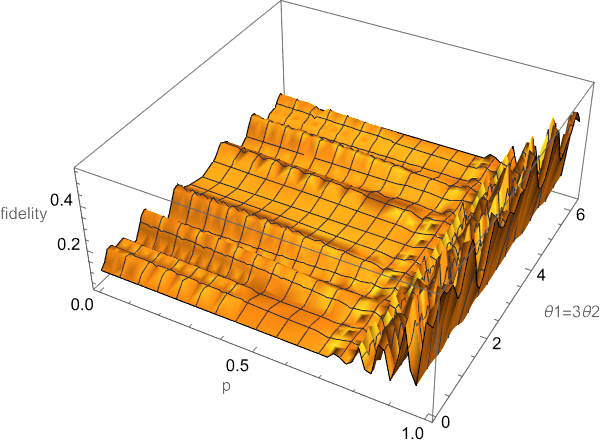}
       		\caption{$F_{H_5}$}
       	\end{subfigure}
       	\caption{The quantum teleportation fidelity under non-Markovian dephasing noise is plotted with respect to the channel parameter $p$ and the state parameter $\theta_2$ for different hypergraphs. To plot these figures we fixed $\eta = 0.5$, and $\beta = 100$.}
       \label{fig:non_Markovian_Dephasing}
       \end{figure}

		\subsection{Non-Markovian depolarization noise:}
		
			We generalize the qubit non-Markovian depolarization noise for qutrits. Then qubit non-Markovian depolarization noise represented by
		\begin{center}
			\begin{tabular}{c c}
				$K_1 = \sqrt{(1-\alpha p) (1-p)} I_2,$ & $K_2 = \sqrt{[1 + \alpha(1-p)] p} \sigma_z,$
			\end{tabular}
		\end{center}
		here we choose $0 \leq \alpha \leq 1$.
		
		For a qutrit system we generalize the non-Markovian depolarization noise using the following Kraus operators 
		\begin{equation} \label{Non Markovian depolarization}
			K_{r,s} = \begin{cases}
				 \sqrt{\frac{9-8p(1-p)}{27}} I_3 & \text{when}~ r = 0, s=0, \\
				 \sqrt{\frac{p(1-p)}{27}} W_{r,s} & \text{for}~ 0\leq r,s \leq 2 ~\text{and}~ (r,s) \neq (0,0).
			\end{cases}
		\end{equation} 
		where the Weyl operators $W_{0,0}, W_{0,1}, W_{0,2}, W_{1,0}, W_{1,1}, W_{1,2}, W_{2,0}, W_{2,1}, \text{and}~ W_{2,2}$ are mentioned in equation (\ref{Weyl operator}). Using Lemma \ref{multiqutrit_channel_lemma}, we construct three-qutrit channel which are described by the following Kraus operators:
		
		\small
		\begin{tabular}{p{5.5cm} p{5.5cm} p{5.5cm}}
			\centering
				$K_1=\sqrt{\frac{9-8p(1-p)}{27}} (K_{0,0}\otimes I_3 \otimes I_3)$,& $K_2=\sqrt{\frac{p(1-p)}{27}} (K_{0,1} \otimes I_3 \otimes I_3)$,& $K_3=\sqrt{\frac{p(1-p)}{27}} (K_{0,2} \otimes I_3 \otimes I_3)$,\\
				$K_4=\sqrt{\frac{p(1-p)}{27}} (K_{1,0} \otimes I_3 \otimes I_3)$, & $K_5=\sqrt{\frac{p(1-p)}{27}} (K_{1,1} \otimes I_3 \otimes I_3)$,& $K_6=\sqrt{\frac{p(1-p)}{27}} (K_{1,2} \otimes I_3 \otimes I_3)$,\\
				$K_7=\sqrt{\frac{p(1-p)}{27}} (K_{2,0} \otimes I_3 \otimes I_3)$,& $K_8=\sqrt{\frac{p(1-p)}{27}} (K_{2,1} \otimes I_3 \otimes I_3)$,& $K_9=\sqrt{\frac{p(1-p)}{27}} (K_{2,2} \otimes I_3 \otimes I_3)$,\\
				$K_{10}=\sqrt{\frac{9-8p(1-p)}{27}} (I_3\otimes K_{0,0} \otimes I_3)$,& $K_{11}=\sqrt{\frac{p(1-p)}{27}} (I_3 \otimes K_{0,1} \otimes I_3)$,& $K_{12}=\sqrt{\frac{p(1-p)}{27}} (I_3 \otimes K_{0,2} \otimes I_3)$,\\
				$K_{13}=\sqrt{\frac{p(1-p)}{27}} (I_3 \otimes K_{1,0} \otimes I_3)$,& $K_{14}=\sqrt{\frac{p(1-p)}{27}} (I_3 \otimes K_{1,1} \otimes I_3)$,& $K_{15}=\sqrt{\frac{p(1-p)}{27}} (I_3 \otimes K_{1,2} \otimes I_3)$,\\
				$K_{16}=\sqrt{\frac{p(1-p)}{27}} (I_3 \otimes K_{2,0} \otimes I_3)$,& $K_{17}=\sqrt{\frac{p(1-p)}{27}} (I_3 \otimes K_{2,1} \otimes I_3)$,& $K_{18}=\sqrt{\frac{p(1-p)}{27}} (I_3 \otimes K_{2,2} \otimes I_3)$,\\
				$K_{19}=\sqrt{\frac{9-8p(1-p)}{27}} (I_3\otimes I_3 \otimes K_{0,0})$,& $K_{20}=\sqrt{\frac{p(1-p)}{27}} (I_3\otimes I_3 \otimes K_{0,1})$,& $K_{21}=\sqrt{\frac{p(1-p)}{27}} (I_3\otimes I_3 \otimes K_{0,2})$,\\
				$K_{22}=\sqrt{\frac{p(1-p)}{27}} (I_3\otimes I_3 \otimes K_{1,0})$,& $K_{23}=\sqrt{\frac{p(1-p)}{27}} (I_3\otimes I_3 \otimes K_{1,1})$,& $K_{24}=\sqrt{\frac{p(1-p)}{27}} (I_3\otimes I_3 \otimes K_{1,2})$,\\
				$K_{25}=\sqrt{\frac{p(1-p)}{27}} (I_3\otimes I_3 \otimes K_{2,0})$,& $K_{26}=\sqrt{\frac{p(1-p)}{27}} (I_3\otimes I_3 \otimes K_{2,1})$,& $K_{27}=\sqrt{\frac{p(1-p)}{27}} (I_3\otimes I_3 \otimes K_{2,2})$.
			\end{tabular}
		
		We now use our quantum teleportation procedure described in Procedure \ref{teleportation_procedure} to transmit any qutrit quantum state via one of the hypergraph states. The teleportation fidelity of the different hypergraph states are as follows: 
		
		\begin{equation}
			\small
			\begin{split}
				F_{H_{1}} = & \frac{5}{23328}[p | p-1|  \{-8 \sin ^2(\text{$\theta $}_1) (16 \cos ^2(\text{$\theta $}_1) \cos (2 \text{$\theta $}_2)+\sin ^2(\text{$\theta $}_1) (9 \sin (4 \text{$\theta $}_2)+\cos (4 \text{$\theta $}_2))\\
				& +\sin (2 \text{$\theta $}_1) (5 \sin (3 \text{$\theta $}_2)-4 \cos (3 \text{$\theta $}_2))-35 \sin (2 \text{$\theta $}_2))+4 \cos (2 \text{$\theta $}_1) (-86 \sin ^2(\text{$\theta $}_1) \sin (2 \text{$\theta $}_2)\\
				& +29 \sin (2 \text{$\theta $}_1) \sin (\text{$\theta $}_2)+8 \sin (\text{$\theta $}_1) \cos (\text{$\theta $}_1) \cos (\text{$\theta $}_2)-1)+4 \sin (2 \text{$\theta $}_1) (139 \sin (\text{$\theta $}_2)+28 \cos (\text{$\theta $}_2))\\
				& +81 \cos (4 \text{$\theta $}_1)+1971\}+(8 (p-1) p+9) \{8 \sin (\text{$\theta $}_1) (3 \sin ^3(\text{$\theta $}_1) \sin (4 \text{$\theta $}_2)+55 \sin (\text{$\theta $}_1) \sin (2 \text{$\theta $}_2)\\
				& +16 \sin (\text{$\theta $}_1) \cos ^2(\text{$\theta $}_1) \cos (2 \text{$\theta $}_2)+\sin ^3(\text{$\theta $}_1) \cos (4 \text{$\theta $}_2)+\cos (\text{$\theta $}_1) (2 \sin ^2(\text{$\theta $}_1) (17 \sin (3 \text{$\theta $}_2)\\
				& -4 \cos (3 \text{$\theta $}_2))+41 \sin (\text{$\theta $}_2)))+4 \cos (2 \text{$\theta $}_1) (74 \sin ^2(\text{$\theta $}_1) \sin (2 \text{$\theta $}_2)+\sin (2 \text{$\theta $}_1) (7 \sin (\text{$\theta $}_2)-4 \cos (\text{$\theta $}_2))+1)\\
				& +80 \sin (2 \text{$\theta $}_1) \cos (\text{$\theta $}_2)-81 \cos (4 \text{$\theta $}_1)+333\}].
			\end{split}
		\end{equation}
		\begin{equation}
			\small
			\begin{split}
				F_{H_{2}} = & \frac{5}{11664}[p | p-1|  \{8 \sin ^2(\text{$\theta $}_1) (-26 \cos ^2(\text{$\theta $}_1) \cos (2 \text{$\theta $}_2)+\sin ^2(\text{$\theta $}_1) \cos (4 \text{$\theta $}_2)\\
				& +\sin (\text{$\theta $}_1) \cos (\text{$\theta $}_1) (7 \cos (3 \text{$\theta $}_2)-11 \sin (3 \text{$\theta $}_2))+22 \sin (2 \text{$\theta $}_2))+\cos (2 \text{$\theta $}_1) (-64 \sin ^2(\text{$\theta $}_1) \sin (2 \text{$\theta $}_2)\\
				& +2 \sin (2 \text{$\theta $}_1) (35 \sin (\text{$\theta $}_2)+19 \cos (\text{$\theta $}_2))+4)+2 \sin (2 \text{$\theta $}_1) (61 \sin (\text{$\theta $}_2)+13 \cos (\text{$\theta $}_2))+9 \cos (4 \text{$\theta $}_1)+1011\}\\
				& -(8 (p-1) p+9) \{-58 \sin (2 \text{$\theta $}_1) \sin (\text{$\theta $}_2)-4 \sin (\text{$\theta $}_1) (4 \sin (\text{$\theta $}_1) \sin (2 \text{$\theta $}_2)+52 \sin (\text{$\theta $}_1) \cos ^2(\text{$\theta $}_1) \cos (2 \text{$\theta $}_2)\\
				& -2 \sin ^3(\text{$\theta $}_1) \cos (4 \text{$\theta $}_2)+\cos (\text{$\theta $}_1) (2 \sin ^2(\text{$\theta $}_1) (29 \sin (3 \text{$\theta $}_2)-7 \cos (3 \text{$\theta $}_2))+11 \cos (\text{$\theta $}_2)))\\
				& +\cos (2 \text{$\theta $}_1) (80 \sin ^2(\text{$\theta $}_1) \sin (2 \text{$\theta $}_2)+38 \sin (2 \text{$\theta $}_1) (\cos (\text{$\theta $}_2)-\sin (\text{$\theta $}_2))+4)+9 \cos (4 \text{$\theta $}_1)-141\}].
			\end{split}
		\end{equation}
		\begin{equation}
			\small
			\begin{split}
				F_{H_{3}} = & \frac{5}{23328} [p | p-1|  \{8 \sin ^2(\text{$\theta $}_1) (-28 \cos ^2(\text{$\theta $}_1) \cos (2 \text{$\theta $}_2)+\sin ^2(\text{$\theta $}_1) (5 \cos (4 \text{$\theta $}_2)-9 \sin (4 \text{$\theta $}_2))\\
				& +\sin (2 \text{$\theta $}_1) (7 \cos (3 \text{$\theta $}_2)-2 \sin (3 \text{$\theta $}_2))-13 \sin (2 \text{$\theta $}_2))+4 \cos (2 \text{$\theta $}_1) (-14 \sin ^2(\text{$\theta $}_1) \sin (2 \text{$\theta $}_2)\\
				& +\sin (2 \text{$\theta $}_1) (26 \sin (\text{$\theta $}_2)+19 \cos (\text{$\theta $}_2))+5)-44 \sin (2 \text{$\theta $}_1) (\cos (\text{$\theta $}_2)-2 \sin (\text{$\theta $}_2))+39 \cos (4 \text{$\theta $}_1)+1989\}\\
				& +(8 (p-1) p+9) \{8 \sin ^2(\text{$\theta $}_1) (28 \cos ^2(\text{$\theta $}_1) \cos (2 \text{$\theta $}_2)-\sin ^2(\text{$\theta $}_1) (3 \sin (4 \text{$\theta $}_2)+5 \cos (4 \text{$\theta $}_2))\\
				& +\sin (2 \text{$\theta $}_1) (8 \sin (3 \text{$\theta $}_2)-7 \cos (3 \text{$\theta $}_2))+13 \sin (2 \text{$\theta $}_2))+4 \cos (2 \text{$\theta $}_1) (14 \sin ^2(\text{$\theta $}_1) \sin (2 \text{$\theta $}_2)-5)\\
				& +2 (22 \sin (2 \text{$\theta $}_1)-19 \sin (4 \text{$\theta $}_1)) \cos (\text{$\theta $}_2)+64 \sin (2 \text{$\theta $}_1) (\cos (2 \text{$\theta $}_1)+2) \sin (\text{$\theta $}_2)-39 \cos (4 \text{$\theta $}_1)+315\}].
			\end{split}
		\end{equation}
		\begin{equation}
			\small
			\begin{split}
				F_{H_{4}}&=\frac{5}{11664}[p | p-1|  \{8 \sin ^2(\text{$\theta $}_1) (-20 \cos ^2(\text{$\theta $}_1) \cos (2 \text{$\theta $}_2)+\sin ^2(\text{$\theta $}_1) \cos (4 \text{$\theta $}_2) +\sin (\text{$\theta $}_1) \cos (\text{$\theta $}_1) (\sin (3 \text{$\theta $}_2)\\
				& +7 \cos (3 \text{$\theta $}_2))+25 \sin (2 \text{$\theta $}_2))+2 \sin (2 \text{$\theta $}_1) (61 \sin (\text{$\theta $}_2)+13 \cos (\text{$\theta $}_2))+4 \cos (2 \text{$\theta $}_1) (19 \sin (\text{$\theta $}_1) \cos (\text{$\theta $}_1) \cos (\text{$\theta $}_2)\\
				& +5 \sin (\text{$\theta $}_1) \sin (\text{$\theta $}_2) (7 \cos (\text{$\theta $}_1)-4 \sin (\text{$\theta $}_1) \cos (\text{$\theta $}_2))+1) +15 \cos (4 \text{$\theta $}_1)+1005\}\\
				& -(8 (p-1) p+9) \{-58 \sin (2 \text{$\theta $}_1) \sin (\text{$\theta $}_2)+8 \sin ^2(\text{$\theta $}_1) (-20 \cos ^2(\text{$\theta $}_1) \cos (2 \text{$\theta $}_2)+\sin ^2(\text{$\theta $}_1) \cos (4 \text{$\theta $}_2)\\
				& +\sin (\text{$\theta $}_1) \cos (\text{$\theta $}_1) (7 \cos (3 \text{$\theta $}_2)-17 \sin (3 \text{$\theta $}_2))+\sin (2 \text{$\theta $}_2))-44 \sin (\text{$\theta $}_1) \cos (\text{$\theta $}_1) \cos (\text{$\theta $}_2)\\
				& +19 \sin (4 \text{$\theta $}_1) \cos (\text{$\theta $}_2)+\cos (2 \text{$\theta $}_1) (4 \sin (\text{$\theta $}_1) \sin (\text{$\theta $}_2) (52 \sin (\text{$\theta $}_1) \cos (\text{$\theta $}_2)-19 \cos (\text{$\theta $}_1))+4)\\
				& +15 \cos (4 \text{$\theta $}_1)-147\}].
			\end{split}
		\end{equation}
		\begin{equation}
			\small
			\begin{split}
				F_{H_{5}} = & \frac{5}{23328}[p | p-1|  \{8 \sin ^2(\text{$\theta $}_1) (-28 \cos ^2(\text{$\theta $}_1) \cos (2 \text{$\theta $}_2)+5 \sin ^2(\text{$\theta $}_1) \cos (4 \text{$\theta $}_2)+\sin (2 \text{$\theta $}_1) (\sin (3 \text{$\theta $}_2)-5 \cos (3 \text{$\theta $}_2))\\
				& +20 \sin (2 \text{$\theta $}_2))+4 \cos (2 \text{$\theta $}_1) (16 \sin ^2(\text{$\theta $}_1) \sin (2 \text{$\theta $}_2)+\sin (2 \text{$\theta $}_1) (7 \cos (\text{$\theta $}_2)-\sin (\text{$\theta $}_2))+5)\\
				& +4 \sin (2 \text{$\theta $}_1) (25 \sin (\text{$\theta $}_2)+\cos (\text{$\theta $}_2))+39 \cos (4 \text{$\theta $}_1)+1989\}\\
				& -(8 (p-1) p+9) \{8 \sin ^2(\text{$\theta $}_1) (-28 \cos ^2(\text{$\theta $}_1) \cos (2 \text{$\theta $}_2)+\sin ^2(\text{$\theta $}_1) (6 \sin (4 \text{$\theta $}_2)+5 \cos (4 \text{$\theta $}_2))\\
				& +\sin (2 \text{$\theta $}_1) (7 \sin (3 \text{$\theta $}_2)-5 \cos (3 \text{$\theta $}_2))-10 \sin (2 \text{$\theta $}_2))+4 \cos (2 \text{$\theta $}_1) (4 \sin ^2(\text{$\theta $}_1) \sin (2 \text{$\theta $}_2)\\
				& +7 \sin (2 \text{$\theta $}_1) (\cos (\text{$\theta $}_2)-\sin (\text{$\theta $}_2))+5)+4 \sin (2 \text{$\theta $}_1) (\cos (\text{$\theta $}_2)-41 \sin (\text{$\theta $}_2))+39 \cos (4 \text{$\theta $}_1)-315\}].
			\end{split}
		\end{equation}
		
		For $\ket{\phi}_a = \ket{+}$, $\frac{1}{\sqrt{2}}(\ket{0} + \ket{2})$ and $\ket{0}$ we plot the teleportation fidelity for different hypergraph states with respect to channel parameter $p$ in Figure \ref{fig:Non_Markovian_depolarization}. In Figure \ref{fig:Non_Markovian_Depolarization}, we plot the teleportation fidelity for different hypergraph state with respect to the channel parameter $p$ and channel parameter $\theta_2$ considering $\theta_1 = 3 \theta_2$.

    \begin{figure}
    	\includegraphics[scale = .78]{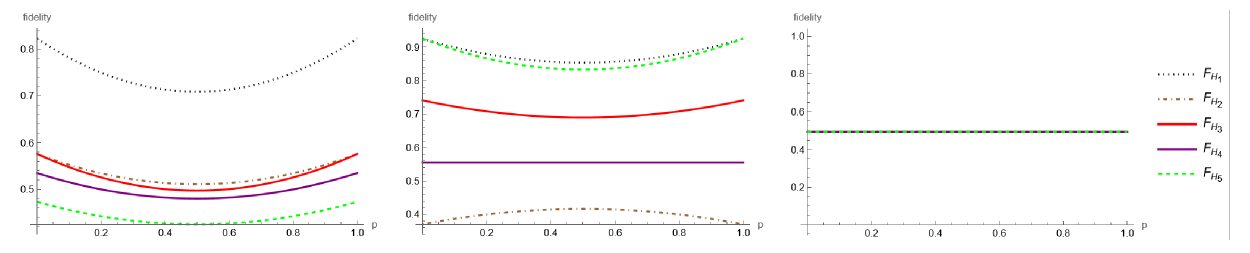}
    	\caption{For all five hypergraph states we plot the teleportation fidelity with respect to the channel parameter $p$, in case of non-Markovian depolarization noise. Three subfigures consider different states $\ket{\phi}_a$ to teleport.  We consider $\ket{\phi}_a = \ket{+}$, $\ket{\phi}_a = \frac{1}{\sqrt{2}}(\ket{0} + \ket{2})$, and $\ket{\phi}_a = \ket{0}$ in the left, middle, and right subfigure, respectively.}
    	\label{fig:Non_Markovian_depolarization}
    \end{figure}
    
    \begin{figure}
    	\centering
    	\begin{subfigure}{0.32\textwidth}
    		\centering
    		\includegraphics[scale = .35]{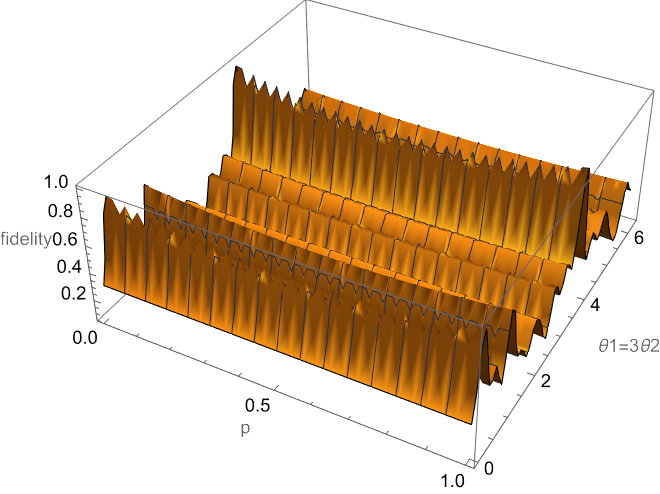}
    		\caption{$F_{H_1}$}
    	\end{subfigure}
    	\begin{subfigure}{0.32\textwidth}
    		\centering
    		\includegraphics[scale = .37]{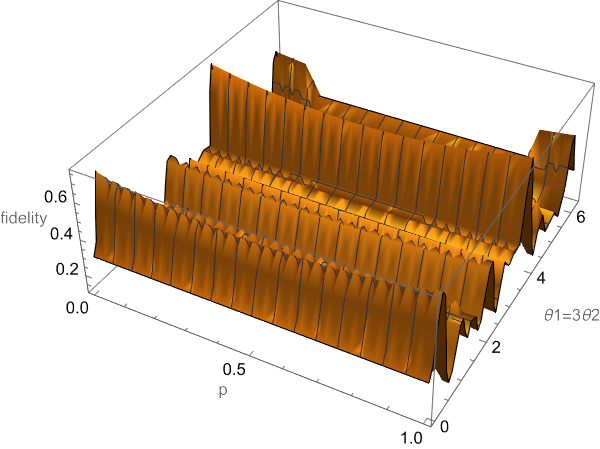}
    		\caption{$F_{H_2}$}
    	\end{subfigure}
    	\begin{subfigure}{0.32\textwidth}
    		\centering
    		\includegraphics[scale = .37]{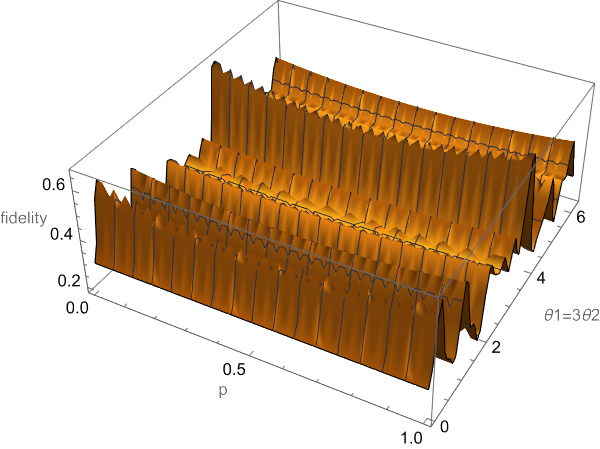}
    		\caption{$F_{H_3}$}
    	\end{subfigure}
    	\\
    	\begin{subfigure}{0.32\textwidth}
    		\centering
    		\includegraphics[scale = .37]{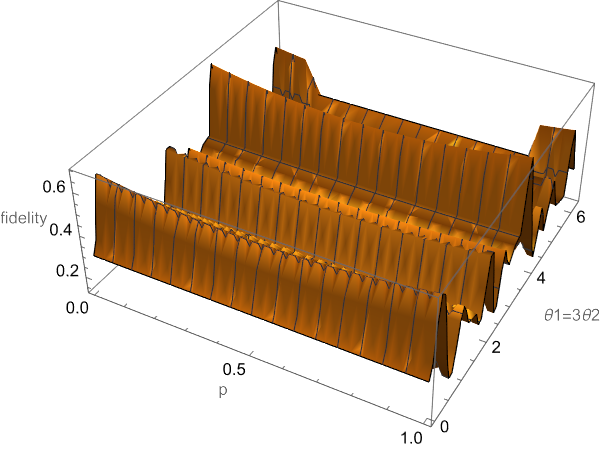}
    		\caption{$F_{H_4}$}
    	\end{subfigure}
    	\begin{subfigure}{0.32\textwidth}
    		\centering
    		\includegraphics[scale = .37]{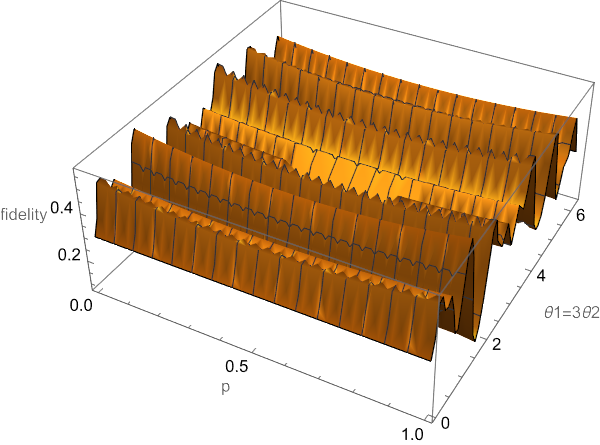}
    		\caption{$F_{H_5}$}
    	\end{subfigure}
    	\caption{The quantum teleportation fidelity under non-Markovian depolarization noise is plotted with respect to the channel parameter $p$ and the state parameter $\theta_2$. Different subfigure is generated for different hypergraphs. Here we consider the non-Markovian depolarization noise.}
    	\label{fig:Non_Markovian_Depolarization}
    \end{figure}

	\section{Discussions and conclusion}
	
		This article is at the interface of quantum communication with qutrits and hypergraph theory. In our model of quantum teleportation, we need three-qutrit entangled hypergraph states. There are five connected hypergraphs with three vertices which represents three qutrit entangled states. We transport them via noisy quantum channels and share them between two parties, Alice and Bob, for quantum teleportation. Alice wants to teleport a general qutrit state to Bob via the distributed state between them. In this article, we calculate analytical expression of teleportation fidelity which depends on the hypergraph state, noisy channel, and qutrit state to communicate. We have considered three specific states $\ket{\phi}_a  = \ket{+}, \frac{1}{\sqrt{2}}(\ket{0} + \ket{2})$, and $\ket{0}$ for teleportation. We plot the teleportation fidelity for them which depends on channel parameters. 
		
		Figure \ref{fig:Qutritflip noise}, \ref{fig:Qutrit-phase-flip noise}, \ref{fig:Depolarizing}, \ref{fig:Markovian_amplitude_damping_channel}, \ref{fig:Non-Markovian Amplitude Damping channel}, \ref{fig:MarkovianDephasing}, \ref{fig:non-Markovian Dephasing}, and \ref{fig:Non_Markovian_depolarization} suggests that the teleportation fidelity depends on the initial states, noise parameters and the hypergraph states distributed between parties. To teleport $\ket{\phi}_a = \ket{0}$ all the hypergraph states are equally efficient indifferent of noises. But, for the other values of $\ket{\phi}_a$, we observe the effect of hypergraph states and noises on teleportation fidelity.
		
		The hypergraph state $\ket{H_1}$ is most efficient to teleport the state $\ket{\phi}_a = \ket{+}$. In Figure \ref{fig:Qutritflip noise}, \ref{fig:Qutrit-phase-flip noise}, \ref{fig:Depolarizing}, \ref{fig:Markovian_amplitude_damping_channel}, \ref{fig:Non-Markovian Amplitude Damping channel}, \ref{fig:MarkovianDephasing}, \ref{fig:non-Markovian Dephasing}, and \ref{fig:Non_Markovian_depolarization} we mark $F_{H_{1}}$, which is the teleportation fidelity for $\ket{H_1}$, with a dotted line. Figure \ref{fig:Qutritflip noise} corresponds the qutrit-flip noise, it indicates that $F_{H_{1}}$ increases with respect to the channel parameter. Whereas the teleportation fidelity $F_{H_{i}}: i = 2, 3, 4, 5$ decrease with respect to the channel parameter. For the other channels all the teleportation fidelity $F_{H_{i}}: i = 1, 2, 3, 4, 5$ decreases. But $F_{H_{1}}$ has the maximum value in all the cases.
		
		For our investigations we consider $\ket{\phi}_a = \frac{1}{\sqrt{2}} (\ket{0} + \ket{2})$, also. The teleportation fidelity $F_{H_{4}}$ for the hypergraph state $\ket{H_4}$ remains constant for qutrit-phase-flip channel, depolarizing channel, Markovian dephasing channel, and non-Markovian depolarization channel. Behavior of all the teleportation fidelity $F_{H_{i}}$ for $i = 1, 2, \dots 5$ are similar for  qutrit-phase-flip channel, depolarizing channel, and Markovian dephasing channel. For the qutrit-flip noise $F_{H_{1}}$ remains constant for all values of the channel parameter. Also, $F_{H_{1}} > F_{H_{i}}$, for $i = 2, 3, \dots 5$, in this case. Except qutrit-flip noise, $F_{H_{1}}$ and $F_{H_{5}}$ decrease from same point with respect to the channel parameter. For all these channel $F_{H_{5}}$ decreases faster than $F_{H_{1}}$. We mark $F_{H_{1}}$ and $F_{H_{5}}$ with a dotted line, and a green dashed line in Figure \ref{fig:Qutrit-phase-flip noise}, \ref{fig:Depolarizing}, \ref{fig:Markovian_amplitude_damping_channel}, \ref{fig:Non-Markovian Amplitude Damping channel}, \ref{fig:MarkovianDephasing}, \ref{fig:non-Markovian Dephasing}, and \ref{fig:Non_Markovian_depolarization}.

	\section*{Appendix}
		
		In this article, we consider five hypergraphs $H_1, H_2, H_3, H_4$, and $H_5$ for our investigations which are depicted in Figure \ref{Hypergraph}. Now, we present our calculation for $H_5$, which is drawn in Figure \ref{fig:$H5$}. The calculation for the other hypergraphs are similar. 
		
		Recall the construction of the qutrit hypergraph states, which is mentioned in Section \ref{Quantum_hypergraph_states_with_qutrits}.  Corresponding to every vertex we assign a $\ket{+}$ state. As there are only three vertices in every hypergraph, the combined state is $\ket{+++}$, which is mentioned in equation (\ref{+++}). There are three hyperedges with cardinality two in the hypergraph $H_5$, which are $e = (0, 1), (1, 2)$ and $(2, 0)$. The controlled-$Z^{(3)}$ operator applicable for these hyperedges is defined in equation (\ref{controlled_z_for_cardinality_2}). Also, there is one hyperedge with cardinality $3$, which is $(0, 1, 2)$. The corresponding controlled-$Z^{(3)}$ gate is described in equation (\ref{controlled_z_for_cardinality_3}). Recall that we can apply the control $Z^{(3)}$ gates in any order.
		
		We may begin with applying $CCZ_{(0,1,2)}^{(3)}$ on $\ket{+++}$, to construct $\ket{H_5}$. The new state is
		\begin{equation} \label{g1}
			\small
			\begin{split}
				\ket{T_1} = &CCZ_{(0,1,2)}^{(3)}\ket{+++} \\
				= & \frac{1}{3\sqrt{3}} (\ket{000} + \ket{001} + \ket{002} + \ket{010} + \ket{011} + \ket{012} + \ket{020} + \ket{021} + \ket{022} + \ket{100} + \ket{101} + \ket{102} \\
				& + \ket{110}+\ket{111}+\ket{112}+\ket{120}+\ket{121}+\ket{122}+\ket{200}+\ket{201}+\ket{202}+\ket{210}+\ket{211}+\ket{212}\\
				& + \ket{220}+ \omega \ket{221}+ \omega^2 \ket{222}). 
			\end{split}
		\end{equation}
		Corresponding to the hyperedge $(2, 0)$ we apply $CZ_{(0, 2)}^{(3)}$ on the state $\ket{T_1}$. The new state is 
		\begin{equation} \label{f3}
			\begin{split}
				\ket{T_2} = &CZ_{(0, 2)}^{(3)} CCZ_{(0,1,2)}^{(3)}\ket{+++}\\
				= & \frac{1}{3\sqrt{3}} (\ket{000} + \ket{001} + \ket{002} + \ket{010} + \ket{011} + \ket{012} + \ket{020} + \ket{021} + \ket{022} + \ket{100} + \ket{101} + \omega \ket{102} \\
				& + \ket{110}+\ket{111}+ \omega \ket{112}+\ket{120}+\ket{121}+ \omega \ket{122}+\ket{200}+\ket{201}+ \omega^2 \ket{202}+\ket{210}+\ket{211}+ \omega^2 \ket{212}\\
				& + \ket{220}+ \omega \ket{221}+ \omega \ket{222}). 
			\end{split}
		\end{equation}
        Similarly, applying the operators corresponding to the remaining hyperedges we get 
		\begin{equation}
			\begin{split}
				& CZ_{(0, 1)}^{(3)} CZ_{(1, 2)}^{(3)} CZ_{(0, 2)}^{(3)} CCZ_{(0,1,2)}^{(3)}\ket{+++}\\
				= & \frac{1}{3\sqrt{3}} \sum_{q_0 = 0}^{2} \sum_{q_1 = 0}^{2} \sum_{q_2 = 0}^{2} f_{(0, 1)}(q_0, q_1, q_2) f_{(1, 2)}(q_0, q_1, q_2) f_{(2, 0)}(q_0, q_1, q_2) g_{(0, 1, 2)} (q_0, q_1, q_2) \ket{q_0, q_1, q_2}\\
				= & \frac{1}{3\sqrt{3}} (\ket{000} + \ket{001} + \ket{002} + \ket{010} + \ket{011} + \ket{012} + \ket{020} + \omega \ket{021} + \omega^2 \ket{022} + \ket{100} + \ket{101} + \omega \ket{102} \\
				& + \ket{110}+\ket{111}+ \omega \ket{112}+\ket{120}+ \omega \ket{121}+ \ket{122}+\ket{200}+ \ket{201}+ \omega^2 \ket{202}+ \omega \ket{210}+ \omega \ket{211}+ \ket{212}\\
				& + \omega^2 \ket{220}+ \omega \ket{221}+ \omega^2 \ket{222}), 
			\end{split}
		\end{equation} 
		which is the qutrit hypergraph state $\ket{H_5}$ corresponding to $H_5$. The density matrix corresponding to $\ket{H_5}$ is 
		\begin{equation} \label{rhoh5}
			\hat{\rho_{5}} = \ket{H_{5}}_{123} \bra{H_{5}}.
		\end{equation}
		We can represent these operations as a qutrit quantum circuit, which is depicted in Figure \ref{H_5 circuit}.
		
		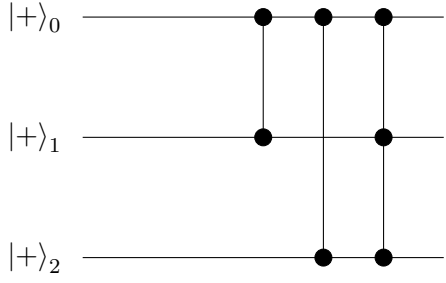
\begin{figure}
			\centering  
			\begin{tikzpicture}[scale=1.6]
				\draw (0, 0) -- (3, 0);
				\draw (0, 1) -- (3, 1);
				\draw (0, 2) -- (3, 2);
				\node at (-.4, 0) {$\ket{+}_2$};
				\node at (-.4, 1) {$\ket{+}_1$};
				\node at (-.4, 2) {$\ket{+}_0$};
				\draw[fill] (1.5, 1) circle [radius = 2pt];
				\draw[fill] (1.5, 2) circle [radius = 2pt];
				\draw (1.5, 1) -- (1.5, 2);
				\draw[fill] (2, 2) circle [radius = 2pt];
				\draw[fill] (2, 0) circle [radius = 2pt];
				\draw (2, 0) -- (2,2);
				\draw (2.5, 1) -- (2.5, 2);
				\draw (2.5, 0) -- (2.5, 1);
				\draw[fill] (2.5, 1) circle [radius = 2pt];
				\draw[fill] (2.5, 2) circle [radius = 2pt];
				\draw[fill] (2.5, 0) circle [radius = 2pt];
			\end{tikzpicture}
			\caption{The quantum circuit corresponding to the hypergraph $H_5$ depict in Figure \ref{fig:$H5$}. The horizontal lines represent $\ket{+}$ states corresponding to the vertices $0, 1$, and $2$. The first vertical line between $\ket{+}_0$ and $\ket{+}_1$ representes $CZ^{(3)}$ operation on $\ket{+}_0$ and $\ket{+}_1$ which corresponds the hyperedge $(0, 1)$. Similarly the second vertical line between $\ket{+}_0$ and $\ket{+}_2$ represents $CZ^{(3)}$ operation on $\ket{+}_0$ and $\ket{+}_2$ which corresponds the hyperedge $(0, 2)$. Also the third vertical line corresponds the $CCZ^{(3)}$ operation related to the hyperedge $(0, 1, 2)$.}
			\label{H_5 circuit}
		\end{figure}
		
		According to our quantum teleportation protocol, Dave prepares $\ket{H_5}$ and sends two particles to Alice and one to Bob via a noisy quantum channel. A number of quantum channels are consider in this work. For this discussion we consider the qutrit-flip noisy channel on $\ket{H_5}$, which is represented by the Kraus operators mentioned in equation (\ref{Qutrit flip noise}). After transmitting the hypergraph state $\ket{H_5}$ via the qutrit-flip channel we get a new state represented by the density matrix $\hat {\hat {\rho_5}} = \sum K_m \hat{\rho_{5}} K^\dagger_{m}$, where $K_m$ are mentioned in subsection \ref{Qutrit-flip noise and quantum teleportation}. As Alice wants to transfer the state in equation (\ref{quantum state vector}) using $\hat {\hat {\rho_5}}$ the combined state is $\rho_{H_5} = \rho_{in} \otimes \hat {\hat {\rho_5}}$. After Von-Neumann measurement and recovery of states using classical communications and application of unitary operation Bob receive the state 
		\begin{equation}
			\small
			\begin{split}
				\rho^{out}_{H_{5}} = & \sum_{l=1}^{4} Tr[M_{l}^{\dagger} M_{l} \rho_{H_{5}}] \frac{U_lTr_{a01}[M_l\rho_{H_{5}}M^{\dagger}_l]U^{\dagger}_l}{Tr[M^{\dagger}_lM_l\rho_{H_{5}}]}\\
				= & \big[\frac{1}{324}\{ 2 (2 \text{p}-1) \sin ^2(\text{$\theta $}_1) \sin (2 \text{$\theta $}_2)+(4-3 \text{p}) \sin (2 \text{$\theta $}_1) \sin (\text{$\theta $}_2)+(3 \text{p}-2) \sin (2 \text{$\theta $}_1) \cos (\text{$\theta $}_2)+16\} \ket{0}\\
				& + \frac{1}{324}\{\sin (\text{$\theta $}_1) (\cos (\text{$\theta $}_1) ((-7 i \sqrt{3} \text{p}+13 \text{p}+6 i \sqrt{3}-6) \sin (\text{$\theta $}_2)+(3 i \sqrt{3} \text{p}-19 \text{p}-6 i \sqrt{3}+18) \cos (\text{$\theta $}_2))\\
				& +(-6-6 i \sqrt{3}) (\text{p}-1) \sin (\text{$\theta $}_1) \sin (\text{$\theta $}_2) \cos (\text{$\theta $}_2))\}\ket{1} + \frac{1}{972}\{((6+2 i \sqrt{3}) \text{p}-3 i \sqrt{3}-3) \sin ^2(\text{$\theta $}_1) \sin (2 \text{$\theta $}_2)\\
				& +(6+(-3-i \sqrt{3}) \text{p}) \sin (2 \text{$\theta $}_1) \sin (\text{$\theta $}_2)+2 (2 (6-i \sqrt{3}) \text{p}+3 i \sqrt{3}-9) \sin ^2(\text{$\theta $}_1) \cos (2 \text{$\theta $}_2)\\
				& +2 ((3-4 i \sqrt{3}) \text{p}+3 i \sqrt{3}+3) \sin (2 \text{$\theta $}_1) \cos (\text{$\theta $}_2)+((-6+4 i \sqrt{3}) \text{p}-3 i \sqrt{3}+9) \cos (2 \text{$\theta $}_1)-6 i (\sqrt{3}-2 i) \text{p}\\
				& +3 i \sqrt{3}+15\}\ket{2}\big]\bra{0} + \big[\frac{1}{648}\{6 i (\sqrt{3}+i) (\text{p}-1) \sin ^2(\text{$\theta $}_1) \sin (2 \text{$\theta $}_2)\\
				& +(7 i \sqrt{3} \text{p}+13 \text{p}-6 i \sqrt{3}-6) \sin (2 \text{$\theta $}_1) \sin (\text{$\theta $}_2)+((-19-3 i \sqrt{3}) \text{p}+6 i \sqrt{3}+18) \sin (2 \text{$\theta $}_1) \cos (\text{$\theta $}_2)\}\ket{0}\\
				& + \frac{1}{324}\{(4-3 \text{p}) \sin ^2(\text{$\theta $}_1) \sin (2 \text{$\theta $}_2)-2 (\text{p}-1) \sin (2 \text{$\theta $}_1) (2 \cos (\text{$\theta $}_2)-\sin (\text{$\theta $}_2))+16\}\ket{1}\\
				& + \frac{1}{972}\{3 (7 i \sqrt{3} \text{p}-4 \text{p}-5 i \sqrt{3}+5) \sin ^2(\text{$\theta $}_1) \sin (2 \text{$\theta $}_2)+3 i (4 \sqrt{3} \text{p}+i-3 \sqrt{3}) \sin (2 \text{$\theta $}_1) \sin (\text{$\theta $}_2)\\
				& +2 (4 i \sqrt{3} \text{p}+9 \text{p}-3 i \sqrt{3}-9) \sin ^2(\text{$\theta $}_1) \cos (2 \text{$\theta $}_2)+(-4 i \sqrt{3} \text{p}+12 \text{p}+3 i \sqrt{3}-3) \sin (2 \text{$\theta $}_1) \cos (\text{$\theta $}_2)\\
				& +((-15-2 i \sqrt{3}) \text{p}+3 i \sqrt{3}+9) \cos (2 \text{$\theta $}_1)-15 \text{p}+4 i \sqrt{3} \text{p}-3 i \sqrt{3}+15\}\ket{2} \big]\bra{1}\\
				& + \big[\frac{1}{972}\{((6-2 i \sqrt{3}) \text{p}+3 i \sqrt{3}-3) \sin ^2(\text{$\theta $}_1) \sin (2 \text{$\theta $}_2)+(6+i (\sqrt{3}+3 i) \text{p}) \sin (2 \text{$\theta $}_1) \sin (\text{$\theta $}_2)\\
				& +2 (2 (6+i \sqrt{3}) \text{p}-3 i \sqrt{3}-9) \sin ^2(\text{$\theta $}_1) \cos (2 \text{$\theta $}_2)+2 ((3+4 i \sqrt{3}) \text{p}-3 i \sqrt{3}+3) \sin (2 \text{$\theta $}_1) \cos (\text{$\theta $}_2)\\
				& +((-6-4 i \sqrt{3}) \text{p}+3 i \sqrt{3}+9) \cos (2 \text{$\theta $}_1)+6 i (\sqrt{3}+2 i) \text{p}-3 i \sqrt{3}+15\}\ket{0}\\
				& + \frac{1}{972}\{\sin ^2(\text{$\theta $}_1) (3 (1+i \sqrt{3}) (5 \sin (2 \text{$\theta $}_2)+2)-3 i \text{p} (2 \sqrt{3}+(7 \sqrt{3}-4 i) \sin (2 \text{$\theta $}_2))\\
				& +2 ((9-4 i \sqrt{3}) \text{p}+3 i \sqrt{3}-9) \cos (2 \text{$\theta $}_2))+i \sin (2 \text{$\theta $}_1) (3 (-4 \sqrt{3} \text{p}+i+3 \sqrt{3}) \sin (\text{$\theta $}_2)\\
				& +(4 (\sqrt{3}-3 i) \text{p}+3 i-3 \sqrt{3}) \cos (\text{$\theta $}_2))+(24+(-30-2 i \sqrt{3}) \text{p}) \cos ^2(\text{$\theta $}_1)\}\ket{1}\\
				&  + \frac{1}{324}\{2 ((2-3 \text{p}) \sin ^2(\text{$\theta $}_1) \sin (2 \text{$\theta $}_2)+(2-3 \text{p}) \sin (2 \text{$\theta $}_1) \sin (\text{$\theta $}_2)+8)+(\text{p}-2) \sin (2 \text{$\theta $}_1) \cos (\text{$\theta $}_2)\}\ket{2}\big]\bra{2}.
			\end{split}
		\end{equation}
		After that, we work out the teleportation fidelity which is 
		\begin{equation}
			\small
			\begin{split}
				F_{H_{5}} = & \frac{5}{2592}[12 (5 p-4) \sin ^4(\text{$\theta $}_1) \sin (4 \text{$\theta $}_2)+4 (13 p+20) \sin ^2(\text{$\theta $}_1) \sin (2 \text{$\theta $}_2)\\
				& +8 (4 p-5) \sin ^4(\text{$\theta $}_1) \cos (4 \text{$\theta $}_2)+8 (10-9 p) \sin ^3(\text{$\theta $}_1) \cos (\text{$\theta $}_1) \cos (3 \text{$\theta $}_2)\\
				& +16 (13 p-7) \sin ^3(\text{$\theta $}_1) \cos (\text{$\theta $}_1) \sin (3 \text{$\theta $}_2)+4 (14-17 p) \sin ^2(2 \text{$\theta $}_1) \cos (2 \text{$\theta $}_2)\\
				& +4 \cos (2 \text{$\theta $}_1) ((31 p-4) \sin ^2(\text{$\theta $}_1) \sin (2 \text{$\theta $}_2)+4 p-5)+2 \sin (2 \text{$\theta $}_1) \cos (\text{$\theta $}_2) ((19 p-14) \cos (2 \text{$\theta $}_1)\\
				& +5 p-2) +4 \sin (2 \text{$\theta $}_1) \sin (\text{$\theta $}_2) ((p+7) \cos (2 \text{$\theta $}_1)-37 p+41)+(38 p-39) \cos (4 \text{$\theta $}_1)\\
				& -54 p+315].
			\end{split}
		\end{equation}

		Putting $\theta_1 = \sin^{-1} \left(\sqrt{\frac{2}{3}}\right)$, and $\theta_2 = \frac{\pi}{4}$ we find $\ket{\phi}_a = \ket{+}$. For these values of $\theta_1$ and $\theta_2$, the teleportation fidelity is
		\begin{equation}
			\small
			\begin{split}
				F_{H_{5}} = & \frac{25}{729} (13 - 3 p).
			\end{split}
		\end{equation}
		Now, $F_{H_{5}}$ is a function of the channel parameter $p$, only. We plot $F_{H_{5}}$ in Figure \ref{fig:Qutritflip noise} with green dashed line. Similarly, putting other values of $(\theta_1,\theta_2) = (\frac{\pi}{4}, \frac{\pi}{2})$ and $(\theta_1,\theta_2) = (0, \frac{\pi}{2})$ which are mentioned in equation (\ref{zero+two_state}), and (\ref{zero_state}) we find the teleportation fidelity for $\frac{1}{\sqrt{2}}(\ket{0} + \ket{2})$ and $\ket{0}$. They are also the function of channel parameter $p$, only. They are depicted in middle, and right subfigure of Figure \ref{fig:Qutritflip noise}, respectively.
	 
		The calculation for the other hypergraph states and channels are similar.

	\section*{Acknowledgment}
	We are thankful to Javid Naikoo for some programs.
	
	\section*{Funding}
		
		SG is supported by a doctoral fellowship from MoE, Government of India. SD is supported by a sponsored project entitled ``Transmission of quantum information using perfect state transfer" (CRG/2021/001834) funded by Science and Engineering Research Board, Government of India.
	
	\section*{Data availability statement}
	No data is generated in this study.


\end{document}